\documentclass[version=preprint,notanonymous]{iacrcc}

\license{CC-by}

\title[running  = {Drawing the $\LINE$},
      ]{Drawing the $\LINE$: Cryptographic Analysis and Security Improvements for the $\LINE$ E2EE Protocol}
\addauthor[%orcid    = {0000-0003-1010-8157},
           inst     = {1},
           email    = {benjamin.dowling@kcl.ac.uk},
           surname  = {Dowling}
          ]{Benjamin Dowling}

\addauthor[%orcid   = {0000-0001-7890-5430},
           inst    = {2},
           email   = {p.gope@sheffield.ac.uk},
           surname = {Gope},
          ]{Prosanta Gope}

\addauthor[%orcid   = {0000-0001-7890-5430},
           inst    = {2},
           email   = {munilyas1@sheffield.ac.uk},
           surname = {Nisa},
          ]{Mehr U Nisa}

\addauthor[%orcid    = {0000-0002-7688-504X},
           inst     = {1},
           email    = {bhagya.wimalasiri@kcl.ac.uk},
           surname  = {Wimalasiri}
          ]{Bhagya Wimalasiri}

\usepackage{bm}
 \usepackage[
 n, 
 advantage, 
 operators, 
 sets, 
 adversary, 
 landau, 
 probability, 
 notions, 
 logic, 
 ff, 
 mm, 
 primitives, 
 events, 
 complexity, 
 oracles, 
 asymptotics, 
 keys
 ]{cryptocode}
\usepackage{xcolor}
\usepackage{needspace}
 \usetikzlibrary{shapes,arrows, calc}
 \usetikzlibrary {patterns,patterns.meta, backgrounds}
 \usepackage{multicol}
 \usepackage{adjustbox}
 \usepackage{tcolorbox}
\usepackage{diagbox}
 \usepackage{makecell}
 \usepackage[shortlabels,inline]{enumitem}
 \usepackage{todonotes}
 \usepackage{subfig}
 \usepackage{amsthm}
 \usepackage[normalem]{ulem}
 \usepackage{algorithm}
 \usepackage[noend]{algpseudocode}
 \usepackage{framed}
 \usepackage{algorithm, setspace}
 \usepackage{tikz}
\usepackage{tikzpeople}
\usetikzlibrary{positioning, shapes.multipart}
\tcbuselibrary{breakable}
\tcbset{%any default parameters
	width=0.7\textwidth,
	halign=justify,
	center,
	breakable,
	colback=white    
}
\usepackage{wrapfig}
\usepackage{caption}
\newcommand{\valid}{\ensuremath{\mathbf{valid}}}

\newcommand{\asym}{\ensuremath{\mathsf{asym}}}
\newcommand{\initial}{\ensuremath{\mathsf{inital}}}
\newcommand{\LL}{\ensuremath{\mathsf{LL}}}
\newcommand{\EL}{\ensuremath{\mathsf{EL}}}
\newcommand{\EE}{\ensuremath{\mathsf{EE}}}
\newcommand{\sym}{\ensuremath{\mathsf{sym}}}
\newcommand{\rootkey}{\ensuremath{\mathsf{rk}}}
\newcommand{\asymindexsender}{\ensuremath{\mathsf{i_s}}}
\newcommand{\asymindexrcvr}{\ensuremath{\mathsf{i_r}}}
\newcommand{\symindexsender}{\ensuremath{\mathsf{j_s}}}
\newcommand{\symindexrcvr}{\ensuremath{\mathsf{j_r}}}

\newcommand{\chainkey}{\ensuremath{\mathsf{ck}}}
\newcommand{\msgkey}{\ensuremath{\mathsf{mk}}}
\newcommand{\LINEDR}{\ensuremath{\mathsf{LINEvDR}}}
\newcommand{\PKL}{\ensuremath{\mathtt{PKL}}}
\newcommand{\tou}{\ensuremath{\overset{\cup}{\rightarrow}}}

\newcommand{\initS}{\ensuremath{\mathsf{Init \mhyphen S}}}
\newcommand{\initC}{\ensuremath{\mathsf{Init \mhyphen C}}}
\newcommand{\sendS}{\ensuremath{\mathsf{Send \mhyphen S}}}
\newcommand{\sendC}{\ensuremath{\mathsf{Send \mhyphen C}}}
\newcommand{\rcvS}{\ensuremath{\mathsf{Rcv \mhyphen S}}}
\newcommand{\rcvC}{\ensuremath{\mathsf{Rcv \mhyphen C}}}
\newcommand{\stateS}{\ensuremath{\mathsf{St_S}}}
\newcommand{\stateC}{\ensuremath{\mathsf{St_C}}}
\newcommand{\KeyGenS}{\ensuremath{\mathsf{KeyGen \mhyphen S}}}
\newcommand{\ACCE}{\ensuremath{\mathsf{ACCE}}}

\newcommand{\stateACCE}{\ensuremath{\mathsf{\Pi_\ACCE}}}
\newcommand{\stateMSKE}{\ensuremath{\mathsf{\Pi_\MSKE}}}

\newcommand{\pkS}{\ensuremath{pk_S}}
\newcommand{\skS}{\ensuremath{sk_S}}

\newcommand{\msg}{\ensuremath{m}}
\newcommand{\ciphertxt}{\ensuremath{c}}

\newcommand{\win}{\ensuremath{\mathbf{win}}}
\newcommand{\ifsm}{\ensuremath{\mathbf{if}}}
\newcommand{\whilesm}{\ensuremath{\mathbf{while}}}

\newcommand{\rand}{\ensuremath{\mathsf{rand}}}

\newcommand{\PK}{\ensuremath{\mathbf{PK}}}

\newcommand{\AlgI}{\ensuremath{\mathsf{AlgI}}}
\newcommand{\AlgR}{\ensuremath{\mathsf{AlgR}}}

\newcommand{\keylbl}{\ensuremath{\mathtt{Key}}}
\newcommand{\ivlbl}{\ensuremath{\mathtt{IV}}}
\newcommand{\nonce}{\ensuremath{n}}
\newcommand{\salt}{\ensuremath{s}}
\newcommand{\AAD}{\ensuremath{AAD}}
\newcommand{\RID}{\ensuremath{RID}}
\newcommand{\SID}{\ensuremath{SID}}
\newcommand{\IV}{\ensuremath{IV}}
\newcommand{\vers}{\ensuremath{\mathtt{vers}}}
\newcommand{\ctype}{\ensuremath{\mathtt{ctype}}}
\newcommand{\fromfield}{\ensuremath{\mathtt{from}}}
\newcommand{\tofield}{\ensuremath{\mathtt{to}}}
\newcommand{\totypefield}{\ensuremath{\mathtt{toType}}}
\newcommand{\idfield}{\ensuremath{\mathtt{id}}}
\newcommand{\timefield}{\ensuremath{\mathtt{createdTime}}}
\newcommand{\hasfield}{\ensuremath{\mathtt{hascontent}}}
\newcommand{\contentfield}{\ensuremath{\mathtt{contentType}}}
\newcommand{\eversfield}{\ensuremath{\mathtt{e2eeVersion}}}
\newcommand{\sidfield}{\ensuremath{\mathtt{sessionId}}}
\newcommand{\chunksfield}{\ensuremath{\mathtt{chunks}}}
\newcommand{\delfield}{\ensuremath{\mathtt{deliveredTime}}}
\newcommand{\seqfield}{\ensuremath{\mathtt{seq}}}
\newcommand{\contentMDfield}{\ensuremath{\mathtt{contentMetadata}}}
\newcommand{\bottagfield}{\ensuremath{\mathtt{BOT\_TAG2}}}
\newcommand{\botoriginfield}{\ensuremath{\mathtt{BOT\_ORIGIN}}}
\newcommand{\botcheckfield}{\ensuremath{\mathtt{BOT\_CHECK}}}
\newcommand{\bottrackfield}{\ensuremath{\mathtt{BOT\_TRACK}}}
\newcommand{\textfield}{\ensuremath{\mathtt{text}}}

\newcommand{\KGen}{\ensuremath{\mathsf{KGen}}}
\newcommand{\Activate}{\ensuremath{\mathsf{Activate}}}
\newcommand{\Run}{\ensuremath{\mathsf{Run}}}

\newcommand{\KGEN}{\ensuremath{\mathsf{KGen}}}
\newcommand{\DS}{\ensuremath{\mathsf{DS}}}
\newcommand{\AEAD}{\ensuremath{\mathsf{AEAD}}}
\newcommand{\SIGN}{\ensuremath{\mathsf{Sign}}}

\newcommand{\ENC}{\ensuremath{\mathsf{Encrypt}}}
\newcommand{\DEC}{\ensuremath{\mathsf{Decrypt}}}

\newcommand{\KID}{\ensuremath{\mathsf{Kid}}}
\newcommand{\MSKE}{\ensuremath{\mathsf{MSKE}}}

\newcommand{\SHAtwo}{\ensuremath{\mathsf{SHA\mhyphen256}}}
\newcommand{\AESGCM}{\ensuremath{\mathsf{AES\mhyphen GCM}}}
\newcommand{\AESCBC}{\ensuremath{\mathsf{AES\mhyphen CBC}}}
\newcommand{\AESECB}{\ensuremath{\mathsf{AES\mhyphen ECB}}}
\newcommand{\LINEone}{\ensuremath{\mathsf{LINEv1}}}
\newcommand{\LINEtwo}{\ensuremath{\mathsf{LINEv2}}}
\newcommand{\CtoS}{\ensuremath{\mathsf{C2S}}}

\newcommand{\ClientHello}{\ensuremath{\mathtt{ClientHello}}}
\newcommand{\ServerHello}{\ensuremath{\mathtt{ServerHello}}}
\newcommand{\ServerCertificate}{\ensuremath{\mathtt{ServerCertificate}}}
\newcommand{\ServerKeyExchange}{\ensuremath{\mathtt{ServerKeyExchange}}}
\newcommand{\ServerCertificateVerify}{\ensuremath{\mathtt{ServerCertificateVerify}}}
\newcommand{\ClientKeyExchange}{\ensuremath{\mathtt{ClientKeyExchange}}}
\newcommand{\ClientFinished}{\ensuremath{\mathtt{ClientFinished}}}
\newcommand{\ServerFinished}{\ensuremath{\mathtt{ServerFinished}}}
\newcommand{\CH}{\ensuremath{\mathtt{CH}}}
\newcommand{\SH}{\ensuremath{\mathtt{SH}}}
\newcommand{\SC}{\ensuremath{\mathtt{SC}}}
\newcommand{\SKE}{\ensuremath{\mathtt{SKE}}}
\newcommand{\SCV}{\ensuremath{\mathtt{SCV}}}
\newcommand{\CKE}{\ensuremath{\mathtt{CKE}}}
\newcommand{\CF}{\ensuremath{\mathtt{CF}}}
\newcommand{\SF}{\ensuremath{\mathtt{SF}}}

\newcommand{\etal}{et~al.\xspace{}}

\mathchardef\mhyphen="2D

\newcommand{\styleSecurityNotion}[1]{\ensuremath{\mathsf{#1}}}

\newcommand{\styleConstants}[1]{\ensuremath{\mathtt{#1}}}

\renewcommand{\O}[1]{\ensuremath{\mathsf{#1}}}

\newcommand{\getsr}{\ensuremath{\overset{\$}{\leftarrow}}}
\newcommand{\tor}{\ensuremath{\overset{\$}{\rightarrow}}}

\newcommand{\bits}{\ensuremath{\{0,1\}}}
\newcommand{\rreplace}[1]{\ensuremath{\widetilde{#1}}}

\newcommand{\KeySpace}{\ensuremath{\mathcal{K}}}

\newcommand{\SaltSpace}{\ensuremath{\mathcal{S}}}

\newcommand{\groupgen}{\ensuremath{\mathcal{G}}}

\newcommand{\Group}{\ensuremath{\groupgen}}

\renewcommand{\adversary}{\ensuremath{\mathcal{A}}}
\newcommand{\bdversary}{\ensuremath{\mathcal{B}}}
\newcommand{\challenger}{\ensuremath{\mathcal{C}}}

\newcommand{\KeyGen}{\ensuremath{\mathsf{KeyGen}}}

\newcommand{\Hash}{\ensuremath{\mathsf{H}}}
\newcommand{\Enc}{\ensuremath{\mathsf{Enc}}}
\newcommand{\Dec}{\ensuremath{\mathsf{Dec}}}
\newcommand{\KDF}{\ensuremath{\mathsf{KDF}}}
\newcommand{\DDH}{\ensuremath{\mathsf{DDH}}}
\newcommand{\GDH}{\ensuremath{\mathsf{GDH}}}

\newcommand{\PRF}{\ensuremath{\mathsf{PRF}}}

\newcommand{\RO}{\ensuremath{\mathsf{RO}}}
\newcommand{\diffihell}{\ensuremath{\mathsf{DH}}}
\newcommand{\Adv}[2]{\ensuremath{\mathsf{Adv}^{#1}_{#2}}}
\newcommand{\Exp}[2]{\ensuremath{\mathsf{Exp}^{#1}_{#2}}}

\newcommand{\ctr}{\ensuremath{ctr}}

\newcommand{\freshpredicate}{\ensuremath{\mathsf{fresh}}}
\newcommand{\matchpredicate}{\ensuremath{\mathsf{match}}}

\newcommand{\numParties}{\ensuremath{n_P}}
\newcommand{\numSessions}{\ensuremath{n_S}}
\newcommand{\numMsgs}{\ensuremath{n_M}}
\newcommand{\numStages}{\ensuremath{n_D}}

\newcommand{\init}{\ensuremath{\styleConstants{init}}}
\newcommand{\resp}{\ensuremath{\styleConstants{resp}}}

\newcommand{\inprogress}{\ensuremath{\styleConstants{active}}}
\newcommand{\rejected}{\ensuremath{\styleConstants{reject}}}
\newcommand{\accepted}{\ensuremath{\styleConstants{accept}}}
\newcommand{\trueflag}{\ensuremath{\styleConstants{true}}}
\newcommand{\falseflag}{\ensuremath{\styleConstants{false}}}
\newcommand{\tested}{\ensuremath{\styleConstants{tested}}}

\newcommand{\revrand}{\ensuremath{\styleConstants{rev\_rand}}}
\newcommand{\revstate}{\ensuremath{\styleConstants{rev\_state}}}
\newcommand{\hand}{\ensuremath{\styleConstants{hand}}}
\newcommand{\send}{\ensuremath{\styleConstants{send}}}
\newcommand{\accept}{\ensuremath{\styleConstants{accept}}}
\newcommand{\receive}{\ensuremath{\styleConstants{receive}}}

\newcommand{\session}{\ensuremath{\pi}}

\newcommand{\testsess}{\ensuremath{\pi^u_{i}}}

\newcommand{\partnersess}{\ensuremath{\pi^v_j}}
\newcommand{\partsess}{\partnersess}

\newcommand{\status}{\ensuremath{\alpha}}
\newcommand{\pid}{\ensuremath{pid}}
\newcommand{\kid}{\ensuremath{kid}}
\newcommand{\id}{\ensuremath{id}}
\renewcommand{\state}{\ensuremath{\pi}}

\newcommand{\role}{\ensuremath{\rho}}
\newcommand{\peerid}{\pid}
\newcommand{\revsesskey}{\ensuremath{\mathtt{rev\_sesskey}}}
\newcommand{\revltk}{\ensuremath{\mathtt{rev\_ltk}}}

\newcommand{\Test}{\ensuremath{\mathsf{Test}}}

\newcommand{\RevState}{\ensuremath{\mathsf{RevState}}}
\newcommand{\RevRand}{\ensuremath{\mathsf{RevRand}}}
\newcommand{\RevLongTermKey}{\ensuremath{\mathsf{RevLongTermKey}}}
\newcommand{\RevSessKey}{\ensuremath{\mathsf{RevSessKey}}}

\newcommand{\Send}{\ensuremath{\mathsf{Send}}}

\newcommand{\coll}{\styleSecurityNotion{coll}}
\newcommand{\ddh}{\styleSecurityNotion{ddh}}

\newcommand{\aead}{\styleSecurityNotion{aead}}
\newcommand{\auth}{\styleSecurityNotion{auth}}
\newcommand{\conf}{\styleSecurityNotion{conf}}

\newcommand{\msind}{\styleSecurityNotion{ms\mhyphen ind}}
\renewcommand{\indcca}{\styleSecurityNotion{ind \mhyphen cca}}

\newcommand{\mehr}[1]{\textcolor{black}{#1}}
\newcommand{\bw}[1]{\textcolor{orange}{BW:#1}}

\newcommand{\pms}{\ensuremath{pms}}
\newcommand{\AD}{\ensuremath{\mathit{AD}}}
\newcommand{\LINE}{\ensuremath{\mathsf{LINE}}}

\newcommand{\ClientAction}[1]{
	\node[right] at (\InitX, \Y) {#1};
}
\newcommand{\ServerAction}[1]{
	\node[left] at (\RespX, \Y) {#1};
}
\newcommand{\SharedAction}[1]{
	\node at ($1/2*(\InitX, \Y)+1/2*(\RespX, \Y)$) {#1};
}
\newcommand{\AdversaryAction}[1]{
	\node at ($1/2*(\InitX, \Y)+1/2*(\RespX, \Y)$) {#1};
}

\newcommand{\ClientToServer}[3][->]{
	\NextLine[0.5]
	\draw[#1] (\ArrowLeft,\Y) -- node[above] {#2} node[below] {#3} (\ArrowRight,\Y) ;
	\NextLine[0.5]
}
\newcommand{\ServerToClient}[3][->]{
	\NextLine[0.5]
	\draw[#1] (\ArrowRight,\Y) -- node[above] {#2} node[below] {#3} (\ArrowLeft,\Y) ;
	\NextLine[0.5]
}
\newcommand{\ClientToAdversary}[3][->]{
	\NextLine[0.5]
	\draw[#1] (\ArrowLeft,\Y) -- node[above] {#2} node[below] {#3} (\ArrowCenter,\Y) ;
	\NextLine[0.5]
}
\newcommand{\ServerToAdversary}[3][->]{
	\NextLine[0.5]
	\draw[#1] (\ArrowRight,\Y) -- node[above] {#2} node[below] {#3} (\ArrowCenter,\Y) ;
	\NextLine[0.5]
}

\newcommand{\AdversaryToClient}[3][->]{
	\NextLine[0.5]
	\draw[#1] (\ArrowCenter,\Y) -- node[above] {#2} node[below] {#3} (\ArrowLeft,\Y) ;
	\NextLine[0.5]
}
\newcommand{\AdversaryToServer}[3][->]{
	\NextLine[0.5]
	\draw[#1] (\ArrowCenter,\Y) -- node[above] {#2} node[below] {#3} (\ArrowRight,\Y) ;
	\NextLine[0.5]
}

\newcommand{\NextLine}[1][1.0]{
	\pgfmathparse{\Y+#1}
	\edef\Y{\pgfmathresult}
}

\newcommand{\Separator}[2][1.0]{
	\draw[very thick,dotted,darkgray] (\InitX,\Y+0.5) node[above=-0.1cm,anchor=north west] {\bf #2} -- (\RespX,\Y+0.5);
}

\newcommand{\game}[1]{\par\vspace{1em}\noindent\textbf{Game #1}}
\renewcommand{\case}[1]{\par\vspace{1em}\noindent\textbf{Case #1}}
\newcommand{\gameref}[1]{\ensuremath{\mathbf{Game}~ \mathbf{#1}}}
\newcommand{\caseref}[1]{\ensuremath{\mathbf{Case}~ \mathbf{#1}}}

\newcommand{\gameabort}[1]{\ensuremath{abort_{#1}}\xspace}
\newcommand{\linkgame}[2]{\hyperref[#1]{G#2}}
\newcounter{Bdversary}

\begin{document}

\maketitle

% Provide the keywords *before* the abstract
% When keywords contain macros provide the text version as the optional argument
%\newcommand{\Dirac}{Dirac}

% Provide the abstract of your paper

\begin{abstract}
$\LINE$ has emerged as one of the most popular communication platforms in many East Asian countries, including Thailand and Japan, with millions of active users. Therefore, it is essential to understand its security guarantees. In this work, we present the first provable security analysis of the $\LINE$ version two ($\LINEtwo$ messaging protocol, focusing on its cryptographic guarantees in a real-world setting. We capture the architecture and security of the $\LINE$ messaging protocol by modifying the Multi-Stage Key Exchange (MSKE) model, a framework for analysing cryptographic protocols under adversarial conditions. While $\LINEtwo$ achieves basic security properties such as key indistinguishability and message authentication, we highlight the lack of forward secrecy (FS) and post-compromise security (PCS). To address this, we introduce a stronger version of the $\LINE$ protocol, introducing FS and PCS to $\LINE$, analysing and benchmarking our results. 
%\keywords{LINE protocol  \and Secure messaging \and Formal analysis.}
\end{abstract}
%\keywords{Secure Messaging, End-to-End Encryption, Generic Privacy Model (GPM), $\LINE$ Messaging Protocol.}
\section{Introduction} \label{sec:intro}

In our fast-paced and increasingly digital world, the importance of messaging applications for communication cannot be overstated. Platforms like WhatsApp, Signal, Facebook Messenger, Telegram, and many others boast vast user bases, already having (combined) over two billion monthly active users \cite{MAU}. This growth reflects both their popularity and our increasing reliance on them for personal, professional, and social interactions. One such messaging application that has gained significant popularity in East Asian countries such as Thailand and Japan is $\LINE$, which remains relatively underexplored in the literature despite its widespread adoption. 
$\LINE$ has become deeply embedded in the digital ecosystems of Japan and Thailand, extending far beyond its original purpose as a simple communication platform. In Japan, it is used by over 78\% of the population for messaging, payments, news, and government services \cite{datareportal2024japan,lycorp2025mediaguide}. Similarly, in Thailand, $\LINE$ serves as a key platform for not only personal communication but also banking, e-commerce, and healthcare-related services, with more than 54~million active users nationwide \cite{lycorp2025thailand}. This broad adoption underscores $\LINE$'s  expansive reach as a multi-purpose digital infrastructure across East Asia. However, as $\LINE$'s applications broaden, so does its attack surface to potential online threats.

In fact, a growing body of academic research recognizes the security implications associated with the pervasive integration of messaging applications into the daily routines of millions of users. In response, an increasing number of studies have sought to formally analyse the security claims made by messaging applications such as WhatsApp \cite{whatsappSASM}, Signal \cite{signal-1}, Threema \cite{paterson2023three}, Telegram \cite{Telegram_attacks}, and Matrix \cite{10351027}. However, despite its widespread adoption, $\LINE$ is yet to undergo a cryptographic security analysis to substantiate its claimed guarantees. This work addresses this significant research gap by providing the first rigorous examination of the protocol's claimed security properties.

Provable security analysis of modern messaging protocols relies on security models that capture the properties deemed desirable for such constructions. Typically, these properties include: end-to-end encryption (\textbf{E2EE}), which ensures that only the intended communication partners (i.e., sender and receiver) can decrypt messages; forward secrecy (\textbf{FS}), which guarantees that past communications remain confidential even if long-term keys are compromised; and post-compromise security (\textbf{PCS}), which mitigates the impact of a compromise on any future sessions. More generally, \emph{key indistinguishability} prevents adversaries from distinguishing real cryptographic keys from random values, while \emph{message authentication} ensures message integrity and may verify the sender's identity. The absence of these essential properties introduces critical vulnerabilities that can compromise the security of data for millions of users worldwide.  
For instance, flaws in Telegram's message sequencing have exposed users to potential security breaches \cite{Telegram_attacks}, while Threema's closed-source server-side code limits the extent of independent security verification \cite{paterson2023three}. Furthermore, despite offering E2EE, the Matrix protocol was later discovered to have several issues, including out-of-band verification attacks, semi-trusted to fully trusted impersonation attacks, and attacks against message confidentiality \cite{10351027}. These instances highlight that rigorous security analysis of vulnerabilities is not merely a theoretical exercise; it is the essential preventative measure that would ensure the robustness of secure messaging protocols, ultimately safeguarding the privacy and security of users worldwide.

That being said, the value of security analysis and the identification of potential security flaws in protocol constructions must be reinforced empirically through practical adoption and evaluation of more secure implementations. To this end, our work first analyzes the security of the generic $\LINEtwo$ protocol against its claimed security guarantees. We further support our cryptographic security analysis with empirical findings that reveal critical security flaws in the current $\LINEtwo$ architecture. Building on these insights, we introduce a stronger construction, $\LINEDR$, and provide formal cryptographic proofs of its security. Finally, we implement and benchmark both $\LINEtwo$ and our enhanced $\LINEDR$, evaluating the associated costs and demonstrating that stronger security guarantees can be achieved without incurring substantial performance overhead.
\subsection{Related Work} \label{sec: RW}

% Espinoza et al. \cite{LINE-1} identified the absence of forward secrecy and replay protection in LINEv1; Isobe and Minematsu \cite{isobe2018breaking} demonstrated message forgery; and Shi and Yoneyama \cite{shi2019verification} used ProVerif to discover additional forward secrecy weaknesses. However, these evaluations targeted outdated versions of LINE, and a formal cryptographic analysis of the current protocol, LINEv2, has not been published. In parallel, the Multi-Stage Key Exchange (MSKE) model \cite{MSKE-1} has emerged as a powerful tool for analyzing secure messaging protocols under realistic adversarial conditions, and has been applied to Signal \cite{signal-1}, TLS \cite{BenTLS1.3}, and QUIC. In this work, we adapt the MSKE model to analyze LINEv2, capturing the protocol's stateless nature. 

 %Even though $\LINE$ has been praised for its easy-to-use design and wide range of communication tools, it has also been questioned about security and privacy, just like any other digital platform. In response, various forensic analyses have been conducted to identify potential vulnerabilities within the application's infrastructure. These studies look deep into how $\LINE$ works, trying to find any flaws that could give hackers access to user data or harm the platform. 
Although $\LINE$ messenger enjoys widespread popularity across East Asia, existing studies have largely focused on its application security rather than its underlying cryptographic foundations. Chang and Chang~\cite{Line_forensic} present a forensic analysis of the $\LINE$ messenger on Android, outlining methods for extracting traces and evidence from both volatile and non-volatile memory, including chat logs, user account data, and database files. Similarly, Riadi \etal~\cite{Line_forensic2} recovered $\LINE$ conversational data 
 from RAM---including timestamps, message content, and status information---that can serve as digital evidence in online fraud investigations. However, such vulnerabilities stem from application-level flaws rather than weaknesses in the underlying cryptographic design, and are therefore beyond the scope of this work. 
 
Beyond these forensic investigations, only limited attention has been given to the cryptographic security of the protocol itself. A few studies have analysed earlier versions of the $\LINE$ protocol, focusing specifically on the now-deprecated $\LINEone$ design. Espinoza \etal~\cite{LINE-1} identified the absence of replay resilience and forward secrecy, while Shi and Yoneyama~\cite{shi2019verification} corroborated these findings through a ProVerif analysis and further revealed a new attack compromising forward secrecy. Additionally, Isobe \etal~\cite{isobe2018breaking} demonstrated an attack that allowed a malicious server to forge messages. In this work, we focus on $\LINEtwo$, as described in the LINE Technical Whitepaper~\cite{LineWP2022}, which introduces updates intended to address the security vulnerabilities identified in $\LINEone$.

In recent years, numerous studies have examined the formal security of secure messaging protocols \cite{whatsappSASM,signal-1,DBLRAT,paterson2023three,Telegram_attacks,10351027}. Among such work, Cohn-Gordon \etal~\cite{signal-1} provide a comprehensive formal analysis of the Signal protocol, modelling its X3DH and double ratchet mechanisms within the Multi-Stage Key Exchange (MSKE) framework to capture its distinctive ratcheting key-update structure. The MSKE framework models key exchange protocols that derive cryptographic keys over multiple stages, capturing a range of adversarial capabilities relevant to such constructions. Fischlin and G{\"u}nter~\cite{MSKE-1} introduced the MSKE model to analyse the QUIC protocol, formalizing security notions such as message confidentiality, key indistinguishability, forward secrecy, and authentication. Dowling \etal~\cite{BenTLS1.3} later extended this framework to capture multi-stage pre-shared-secret key exchange (MS-PSKE) security, providing a formal analysis of the TLS~1.3 handshake.

\mehr{Our analysis found that bot communications are transmitted in plaintext, lacking the end-to-end encryption applied to standard one-to-one messages (see Appendix \ref{sec:Exp}). This represents a significant security risk for sensitive bot interactions---such as those involving banking or healthcare services—--and raise significant privacy concerns. This issue, which we identify as a novel finding of our work, has also been examined concurrently and independently by Chou \etal~\cite{chatbotExploitation}.}

There has been growing attention in recent years toward improving the security of modern messaging protocols. Albrecht \etal~\cite{Telegram_attacks} introduce minor modifications to the Telegram protocol and formally verify that the use of symmetric encryption in MTProto 2.0 can provide the fundamental security guarantees expected from a bidirectional channel. The work of Paterson \etal~\cite{paterson2023three} identified multiple attacks against Threema Secure Messenger, including exploits that targeted its lack of forward secrecy. Their findings have been instrumental in subsequent security improvements to the Threema protocol, notably the introduction of forward secrecy as a core security guarantee. The ongoing transition toward retrofitting internet protocols with post-quantum primitives has also motivated similar efforts in secure messaging systems. For example, Apple's iMessage has integrated post-quantum security guarantees through PQ3, which employs a Signal-\emph{like} mechanism incorporating a post-quantum key encapsulation mechanism (KEM) within the double ratchet \cite{stebila2024security}. Building on this, Dodis \etal~\cite{dodis2025triple} improve the communication efficiency of PQ3 by reducing its complexity and introduce a post-quantum secure \emph{triple ratchet} construction. While our work introduces several security enhancements to the existing $\LINEtwo$ design, we do not address post-quantum security, as it falls outside the scope of the currently deployed LINE design. Furthermore, as shown in \cite{LineBreak2025}, the $\LINEtwo$ protocol's stateless design and unauthenticated metadata enable a Man-in-the-Middle or malicious server to replay ciphertexts indefinitely, with replayed messages accepted as valid by the recipient. \mehr{Our analysis concurrently identifies similar replay vulnerabilities, reinforcing the need for the security enhancements we introduce to the existing $\LINEtwo$. While our modifications strengthen the $\LINEtwo$ against these specific threats, post-quantum security is outside our current scope, and remains consistent with the limitations of the currently deployed $\LINE$ protocol.}

Despite existing analysis of application-level vulnerabilities and formal studies of other secure messaging protocols, the cryptographic security of $\LINEtwo$ remains largely uncharacterised. This gap motivates a provable security  examination of its security guarantees and the design of improvements to bring it closer to state-of-the-art messaging protocols. In this work, we conduct the first rigorous cryptographic analysis of $\LINEtwo$. Thus, we summarize the main technical contributions of our work as follows:
\begin{enumerate}
    \item Isolating and formally describing the $\LINE$ cryptographic E2EE protocol from the Line Technical Whitepaper~\cite{LineWP2022}.
     \item Cryptographic security analysis of the $\LINEtwo$ using a tailored MSKE construction capable of modelling the security of stateless messaging protocols.
    \item Building on insights from our analysis, we propose an enhanced version of the $\LINE$ E2EE protocol, denoted $\LINEDR$, which integrates the double ratchet mechanism from the Signal protocol.
    % \item A comparative analysis with other widely-used messaging protocols for which formal analyses already exist.
    \item Reference implementation of standard $\LINEtwo$ and our enhanced $\LINEDR$, alongside empirical evaluation of their computational and communication overhead, highlighting the practical viability of strengthening security within the $\LINE$ architecture.
\end{enumerate}

\section{Overview of LINE Messaging Protocol} \label{sec:ovr}

The LINE messaging application provides end-to-end encryption (E2EE) through its proprietary Letter Sealing protocol. Two versions of this protocol are currently deployed in parallel. Letter Sealing version 2 (henceforth will be referred to as $\LINEtwo$) is enabled by default for both one-to-one and one-to-many messaging whenever supported by the sender's client. If a recipient does not support $\LINEtwo$, the protocol is downgraded to Letter Sealing version 1 (henceforth will be referred to as $\LINEone$) to maintain compatibility \cite{LineWP}. For reference and comparison, we provide a formal description of $\LINEone$ in Appendix \ref{Sec:line1}. As discussed in Section \ref{sec:intro}, the limitations of $\LINEone$ have already been resolved by LINE~\cite{LineWP2022}. Therefore, our analysis concentrates exclusively on the one-to-one messaging setting of $\LINEtwo$, which represents the operative mechanism underpinning $\LINE$'s E2EE guarantees.

\subsection{$\LINEtwo$: $\LINE$ Letter Sealing Protocol version 2}\label{sec:linetwo}

$\LINEtwo$ can be described in terms of four main phases: a \emph{Registration} phase, where communicating parties register with the $\LINE$ messaging server, upload public keys and are assigned with the public key identifiers ($\kid$s); a \emph{Session Establishment} phase, where clients retrieve public keys of their intended communicating partners from $\LINE$ messaging server and generate a pre-master secret; an \emph{Encryption} phase, where messages are encrypted under session keys; and a  \emph{Decryption} phase, through which the receiver retrieves the original content of encrypted messages. An algorithmic description of $\LINEtwo$ is given in Figure \ref{fig:line-vers-2}, where we present the protocol from the perspective of two communicating parties, Alice and Bob, and note that the same procedure generalizes to any communicating participants.  Both $\LINEone$ and $\LINEtwo$ follow the same structural framework, differing solely in the encryption phase. We note that sender and recipient identifiers are not encrypted, and therefore, E2E privacy cannot be formally guaranteed. A formal treatment of privacy against external adversaries is left for future work.

\subsubsection{Registration Phase}
Each client generates a long-term elliptic curve Diffie-Hellman (ECDH) key pair for secure communication and registers the public key with $\LINE$'s messaging server. These keys will be used to enable a static-static ECDH key exchange process between Alice and Bob, facilitated by a server. Both parties independently generate long-term Diffie-Hellman key pairs: Alice generates $(x, g^x)$ and Bob generates $(y, g^y)$. Alice sends her public key $g^x$ to the server, where the server generates Alice's key identifier and stores Alice's public key in its database under the identifier $\PK[\kid_A]$. Similarly, Bob sends his public key $g^y$ to the server, where the server generates Bob's key identifier and stores Bob's public key in its database under the identifier $\PK[\kid_B]$.

\subsubsection{Session Establishment}

In what we describe as the \emph{session establishment} phase, the server facilitates the exchange of public keys between communicating partners. This is done by distributing Alice's public key $g^x$ to Bob and Bob's public key $g^y$ to Alice upon request. Alice and Bob initially retrieve each other's public keys from the server, where Alice obtains Bob's public key $g^y$ (identified by $\kid_B$), and Bob retrieves Alice's public key $g^x$ (identified by $\kid_A$).

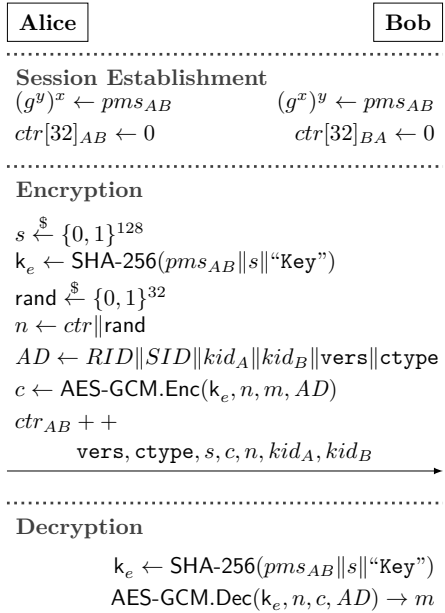
\begin{wrapfigure}{l}
{0.5\textwidth}
    \centering
\scalebox{0.85}{
	\centering
\begin{tikzpicture}[yscale=-0.5,>=latex, xscale = 0.75]
\edef\InitX{0}
\edef\ArrowLeft{0}
\edef\ArrowCenter{4.5}
\edef\ArrowRight{9}
\edef\AdvX{4.5}
\edef\RespX{9}
% Set the starting Y coordinate
\edef\Y{0}

% Draw header boxes
\node [rectangle,draw,inner sep=5pt,right] at (\InitX,\Y) {\textbf{Alice} };
% \node [rectangle, draw, inner sep=5pt, centered] at (\AdvX,\Y) {\textbf{Server}};
\node [rectangle,draw,inner sep=5pt,left] at (\RespX,\Y) {\textbf{Bob}};
\NextLine[0.5]
\Separator{Session Establishment}{}
\NextLine[2]
\ClientAction{$(g^y)^x \gets \pms_{AB}$}
\ServerAction{$(g^x)^y \gets \pms_{AB}$}
\NextLine[1]
\ClientAction{$\ctr[32]_{AB} \gets 0$}
\ServerAction{$\ctr[32]_{BA} \gets 0$}
\NextLine[0.5]

%%%%%%%%%%%%%%%%%%%%%%% Activation Phase %%%%%%%%%%%%%%%%%%
\Separator{Encryption}{}
\NextLine[2.5]
\ClientAction{$\salt \getsr \{0,1\}^{128}$}
\NextLine[1]
\ClientAction{$\key_e \gets \SHAtwo(\pms_{AB} \| \salt \| ``\keylbl" )$}
\NextLine[1]
\ClientAction{$\rand \getsr \bits^{32}$}
\NextLine[1]
\ClientAction{$\nonce \gets \ctr \| \rand$}
\NextLine[1]
\ClientAction{$\AD \gets \RID \| \SID \| \kid_A \| \kid_B \| \vers \| \ctype$}
\NextLine[1]
\ClientAction{$c \gets \AESGCM.\Enc(\key_e, \nonce, m, \AD)$}
\NextLine[1]
\ClientAction{$\ctr_{AB}++$}
\NextLine[1]
\ClientToServer{$\vers,\ctype, \salt, c, \nonce, \kid_A, \kid_B$}{}
\Separator{Decryption}{}
%%%%%%%%%%%% Sending Message %%%%%%%%%%%%%%%%%%%%%%%
\NextLine[2.5]
\ServerAction{$\key_e \gets \SHAtwo(\pms_{AB} \| \salt \| ``\keylbl" )$}
\NextLine[1]
\ServerAction{$\AESGCM.\Dec(\key_e,\nonce,c,\AD)\to m$}

\end{tikzpicture}
}

\caption{
$\LINEtwo$ protocol phases (excluding registration). Given our analysis focuses on session establishment and message exchange, the registration phase is omitted for brevity (see Appendix~\ref{app:line2-detail} for details). Registration and session establishment are illustrated concurrently for both parties but may occur independently.}
\label{fig:line-vers-2}
\end{wrapfigure}

Once they have each other's public keys, they independently compute a shared secret, known as the pre-master secret $\pms_{AB}$. Alice calculates this by raising Bob's public key $g^y$ to her private key $x$, resulting in $g^{xy}$, while Bob does the same by raising Alice's public key $g^x$ to his private key $y$, also deriving $g^{xy}$. This shared secret $\pms_{AB}$, forms the basis for secure communication between them. In addition to the key exchange, both parties initialize a 32-bit counter $\ctr[32]_{AB} \gets 0 $, which ensures keys are freshly derived per message and will be incremented after each encryption/decryption operation.

\subsubsection{Encryption}
When Alice encrypts a message, she begins by generating a 128-bit random salt $\salt$, which is used to derive the fresh encryption key \sloppy{$ \key_e \gets \SHAtwo(\pms_{AB}\|\salt\|$ $''\keylbl'')$}, computed by hashing the long-term premaster secret $\pms_{AB}$, random salt $\salt$, and the string "$\keylbl$" using $\SHAtwo$. This ensures that the derived key is unique for the session. A 32-bit $\nonce$ is then created by concatenating the shared counter $\ctr_{AB}$ and a random value $\rand$. The $\nonce$ ensures that each encryption instance remains unique and, with state management. The additional authenticated data ($\AAD$) comprises metadata including the receiver and sender identifiers ($\RID$, $\SID$), key identifiers ($\mathit{kid}_A$, $\mathit{kid}_B$), the E2EE version ($\vers$), and content type ($\ctype$). It is provided as associated data, which remains unencrypted but authenticated. The actual message $m$ is encrypted using AES-GCM with the derived key $\key_e$, the nonce $\nonce$, and the $\AAD$, resulting in the ciphertext $c \gets \AESGCM.\Enc(\key_e,\nonce,m,\AAD)$. The ciphertext $c$, and metadata ($\salt$, $\nonce$, $\kid_A$, $\kid_B$, $\vers$, $\ctype$), are sent to the recipient. After encryption, the counter is incremented to ensure future messages use a unique nonce.

\subsubsection{Decryption}
The decryption process mirrors the encryption. The recipient re-derives the encryption key $\key_e$ using the same pre-master secret $\pms_{AB}$, and salt $\salt$ provided in the ciphertext. The ciphertext is then decrypted using $\AESGCM$, employing the same $\nonce$ and $\AAD$ used during encryption, ensuring the message and its associated metadata remain unaltered. 

\subsubsection{$\LINEtwo$ Packet Analysis Experiment}
To ensure our understanding of the $\LINEtwo$ protocol, we deployed an experimental testbed to analyse network traffic between two end parties. The study, summarized here and detailed in Appendix \ref{sec:Exp}, investigated whether the cryptographic design described in the official white paper \cite{LineWP} matched observations in a live network environment. Although the white paper describes the overall E2EE design, it's unclear how the client-to-server secure channel ($\CtoS$) layer and E2EE layer (Letter Sealing) are specified. Our analysis confirms that $\LINE$ uses a layered architecture, where client-to-server communication is protected by a TLS 1.2 layer that encapsulates the client-to-client E2E messages. Additionally, we verify the end-to-end encryption used aligns with the specifications outlined in the white paper. \mehr{As discussed in section \ref{sec: RW}, our analysis confirms that bot communications lack the E2EE applied to standard messages, a finding consistent with the concurrent work of Chou \etal~\cite{chatbotExploitation}.}

% While the core E2E message structure generally aligned with the white paper specifications, we discovered additional, unencrypted metadata fields providing message context such as sender and responder IDs and content type etc. Crucially , we found that bot communications are transmitted in plaintext, completely lacking the E2E protection applied to 1:1 messages. This finding, which poses a significant security risk for sensitive bot interactions, constitutes a novel contribution of our work and is explored concurrently by Chou \etal~\cite{chatbotExploitation}. The lack of encryption in bot interactions constitutes a significant privacy weakness, leaving sensitive information shared with banking or healthcare bots exposed to potential interception.

\section{Limitations in $\LINEtwo$ protocol}

Our security analysis of $\LINEtwo$, presented in Section~\ref{sec:MSKE}, reveals both strengths and weaknesses in the protocol's design. $\LINE$ achieves key-indistinguishability, end-to-end encryption, and message authentication. However, it lacks post-compromise security (PCS) and forward secrecy (FS) because it relies on a static-static Diffie-Hellman key exchange. \mehr{Due to its stateless design, $\LINEtwo$ lacks robust replay protection, enabling an active adversary or malicious server to replay ciphertexts that may be accepted by the recipient. This was independently confirmed in \cite{LineBreak2025}.}
In contrast, other state-of-the-art messaging protocols provide stronger security guarantees. For instance, Signal achieves both PCS and FS through the double ratchet mechanism \cite{signal-1}. Matrix likewise supports FS \cite{matrix-comparison}, while Threema not only ensures FS but also incorporates robust replay protection \cite{paterson2023three}. Telegram, on the other hand, provides only partial FS alongside otherwise standard features \cite{telegramFA}. These examples illustrate the gap between $\LINE$'s stateless design and the richer security guarantees adopted elsewhere.

\subsection{Why Not Replace $\LINE$ with Signal?}

Given these considerations, one may ask why $\LINE$ does not adopt a well-established secure messaging protocol such as Signal, which already offers stronger, formally verified security guarantees. Our design of $\LINEDR$ was based on deployment realities. A direct one-to-one adoption of the Signal protocol within the $\LINE$ protocol architecture is not straightforward; it would require a fundamental restructuring of identity management and a complex reconfiguration of the existing infrastructure, risking service interruptions for an established user base of millions.

These difficulties arise largely from differences in the underlying cryptographic building blocks. \mehr{Signal relies on a significant pre-public key bundle infrastructure involving signed pre-keys, one-time pre-keys, and digital signatures. In contrast, $\LINE$'s current design is more rudimentary; it does not use digital signatures or pre-generated bundles and instead relies on a single static Diffie-Hellman public key per user. Furthermore, Signal's inclusion of post-quantum keys presents an additional obstacle, as LINE has yet to adopt post-quantum security guarantees. Consequently, retrofitting Signal’s full architecture into $\LINE$'s existing environment is practically unfeasible.}

To strike a practical balance, we propose an evolutionary path that seamlessly incorporates Signal's core cryptographic component---the double ratchet mechanism--- to fit within $\LINE$'s existing environment. \mehr{Our approach avoids massive infrastructure changes by remaining backward compatible: no signatures, pre-key bundles, or server-side modifications are required. Instead, additional public keys are exchanged alongside existing messages, enabling the protocol to transition directly to a double ratchet-style key-update mechanism, even after session establishment.} This approach provides the same robust security guarantees, specifically FS and PCS, while preserving the performance of a global platform. By optimizing for the specific needs of the LINE environment, we deliver an operationally viable solution that is more efficient than a generic protocol replacement.

 In the following section, we present $\LINEDR$, a double ratchet-based enhancement of $\LINEtwo$ that addresses these limitations without requiring a complete migration to Signal's infrastructure.

\section{Improving the Cryptographic Design of $\LINE$}

$\LINEtwo$ relies solely on static–static $\diffihell$ for key establishment, which limits the level of security the protocol can achieve. To overcome these limitations, we propose $\LINEDR$ (Figure~\ref{fig:line-dblrat-version}), a double ratchet–based enhancement that augments $\LINEtwo$ with forward secrecy, post-compromise security, replay protection, and KCI resilience, while preserving $\LINE$’s overall architecture.

Table~\ref{tab:linev2_comparison} summarizes the security properties of $\LINEtwo$ and our enhanced Letter Sealing variants. While $\LINEtwo$ combines static Diffie–Hellman with $\AEAD$, it does not provide forward secrecy or strong compromise resilience. In contrast, $\LINEDR$ achieves continuous key evolution and robust replay protection, bringing the design closer to the guarantees expected of modern secure messaging protocols.

% \vspace{-1cm}
\setlength{\tabcolsep}{8pt}
\renewcommand{\arraystretch}{1.2}
\begin{table*}[h!tp]
\Huge
\centering
\resizebox{\textwidth}{!}{
\begin{tabular}{|l|c|c|c|}
\hline
\textbf{Security Property} & \textbf{$\LINEtwo$ (Static-Static $\diffihell$)} & \textbf{$\LINEDR$ (Double Ratchet)} \\
\hline
\textbf{End-to-End Encryption} & $\surd$ (AEAD)  & $\surd$ ($\AEAD$ + Key Update) \\
\hline
\textbf{Authentication} & $\surd$ (Implicit via $\diffihell$ + $AD$)  & $\surd$ (Implicit via $\diffihell$ + $AD$) \\
\hline
\textbf{Forward Secrecy} & $\times$  & $\surd$ (Continuous with $\diffihell$ Ratchet) \\
\hline
\textbf{Post-Compromise Security}  & $\times$ & $\surd$  (via Key Update) \\
\hline
\textbf{Replay Protection} & $\times$  & $\surd$ (via Key Update) \\
\hline
\textbf{KCI Resistance} & $\times$  & $\surd$ (TOFU via Key Update) \\
\hline
\end{tabular}
}
\vspace{0.1cm}
\caption{Comparison of cryptographic security properties of LINE Letter Sealing protocol variants.}
\label{tab:linev2_comparison}
\end{table*}

\subsection{$\LINEDR$: LINE Double Ratchet}

\mehr{Figure \ref{fig:line-dblrat-version} illustrates the message encryption and decryption workflow of the proposed protocol, focusing on session establishment and message exchange (the registration phase is provided in Appendix \ref{fig:line-vers-DR(a)}). The protocol follows a double ratchet–style structure, where asymmetric and symmetric ratchets are advanced independently based on the sender's state as described below.}

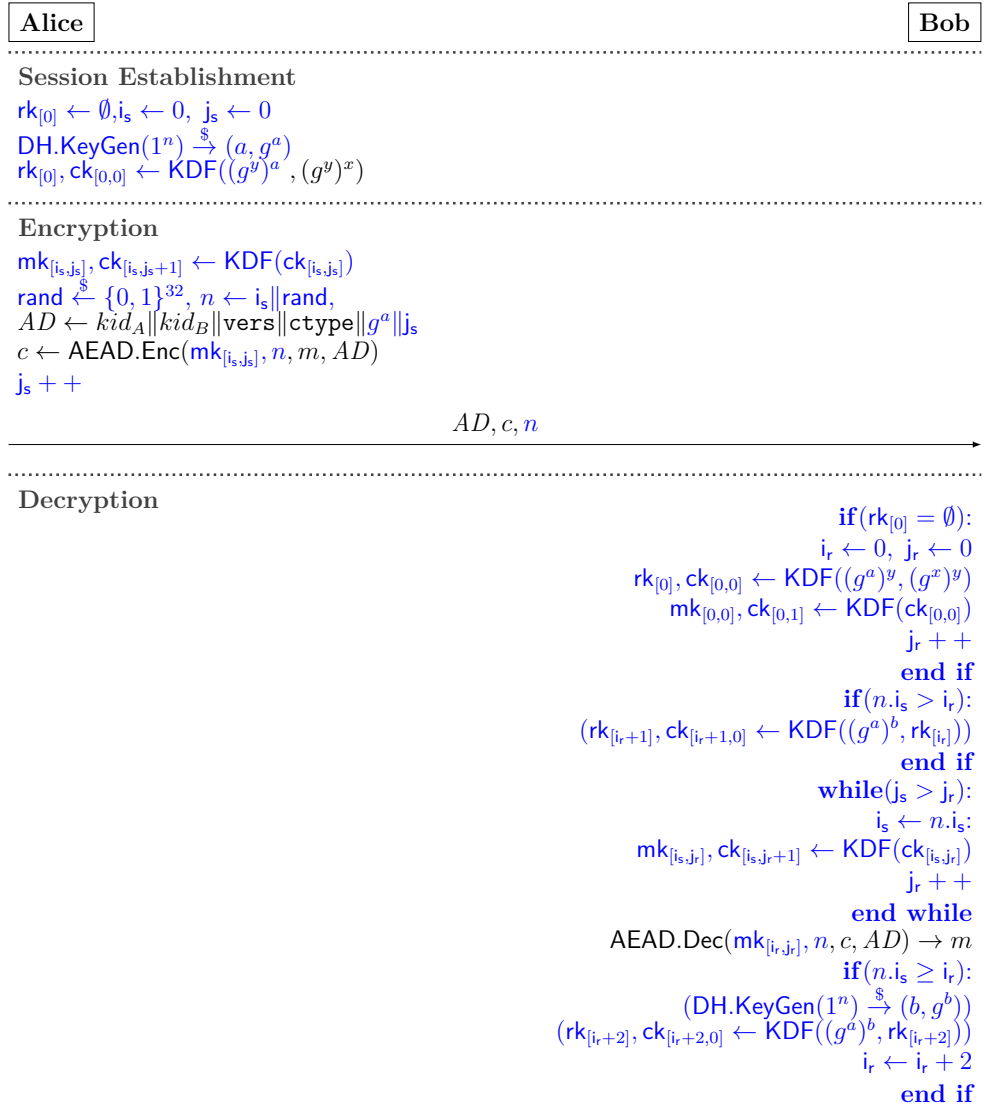
\begin{figure}[htbp]
    \centering
\scalebox{0.8}{
\large
	\centering
\begin{tikzpicture}[yscale=-0.5,>=latex]
\edef\InitX{0}
\edef\ArrowLeft{0}
\edef\ArrowCenter{8}
\edef\ArrowRight{16}
\edef\AdvX{8}
\edef\RespX{16}
% Set the starting Y coordinate
\edef\Y{0}

% Draw header boxes
\node [rectangle,draw,inner sep=5pt,right] at (\InitX,\Y) {\textbf{Alice} };
% \node [rectangle, draw, inner sep=5pt, centered] at (\AdvX,\Y) {\textbf{Server}};
\node [rectangle,draw,inner sep=5pt,left] at (\RespX,\Y) {\textbf{Bob}};
\NextLine[0.5]
\Separator{Session Establishment}{}
\NextLine[2.5]
\ClientAction{\color{blue}$\rootkey_{[0]} \gets \emptyset ,$$\asymindexsender \gets 0, ~\symindexsender \gets 0 $}
% \NextLine[1]
% \ClientAction{\color{blue}$\asymindexsender \gets 0, ~\stageindex \gets 0 $}
\NextLine[1]
\ClientAction{\color{blue}$\diffihell.\KeyGen(1^\secpar) \tor (a,g^{a})$}
\NextLine[1]
\ClientAction{{\color{blue}$\rootkey_{[0]}, \chainkey_{[0,0]}\gets \KDF((g^y)^{a}$} $,(g^y)^x)$}
\NextLine[0.5]

\Separator{Encryption}{}
%%%%%%%%%%%%%%%%%%%%%%% Activation Phase %%%%%%%%%%%%%%%%%%
\NextLine[2.5]
%\ClientAction{$\diffihell.\KeyGen(1^\secpar) \tor (x',g^{x'})$}
%\NextLine[1]
%\ClientAction{$\salt \getsr \{0,1\}^{128}$}
%\NextLine[1]
\ClientAction{\color{blue}$\msgkey_{[\asymindexsender,\symindexsender]},\chainkey_{[\asymindexsender,\symindexsender+1]} \gets \KDF(\chainkey_{[\asymindexsender,\symindexsender]})$}
\NextLine[1]
\ClientAction{\color{blue} $\rand \getsr \bits^{32}$, $\nonce  \gets \asymindexsender \| \rand$, }
% \NextLine[1]
% \ClientAction{\color{blue}$(\nonce,\nonce_r)  \gets \asymindexsender \| \rand$}
\NextLine[1]
\ClientAction{$\AD \gets \kid_A \| \kid_B \| \vers \| \ctype \| \color{blue}g^{a} \| \symindexsender$}
\NextLine[1]
\ClientAction{$\ciphertxt \gets \AEAD.\Enc({\color{blue}\msgkey_{[\asymindexsender,\symindexsender]}, \nonce }, \msg, \AD)$}
\NextLine[1]
\ClientAction{\color{blue}$\symindexsender++$}
\NextLine[1.5]
\ClientToServer{$ \AD, \ciphertxt, {\color{blue}\nonce}$}{}

%%%%%%%%%%%% Decryption %%%%%%%%%%%%%%%%%%%%%%%
\Separator{Decryption}{}
\NextLine[2]
%\ServerAction{$\ifsm (\rootkey_{[0]} = \emptyset)$}
%\NextLine[1]
\ServerAction{\color{blue}$\ifsm (\rootkey_{[0]} = \emptyset)$:}
\NextLine[1]
\ServerAction{\color{blue}$\asymindexrcvr \gets 0, ~\symindexrcvr \gets 0 $}
% \NextLine[1]
% \ServerAction{\color{blue}$\asymindexrcvr \gets 0, ~\stageindex \gets 0 $}
\NextLine[1]
\ServerAction{\color{blue}$\rootkey_{[0]}, \chainkey_{[0,0]}\gets \KDF((g^a)^{y},(g^x)^y)$}
\NextLine[1]
\ServerAction{\color{blue}$\msgkey_{[0,0]}, \chainkey_{[0,1]} \gets \KDF(\chainkey_{[0,0]})$}
\NextLine[1]
\ServerAction{\color{blue}$\symindexrcvr++$}
\NextLine[1]
\ServerAction{\color{blue}$\textbf{end~if}$}
\NextLine[1]
\ServerAction{\color{blue}$\ifsm (\nonce.\asymindexsender > \asymindexrcvr)$:}
\NextLine[1]
\ServerAction{\color{blue}$(\rootkey_{[\asymindexrcvr+1]}, \chainkey_{[\asymindexrcvr+1,0]} \gets \KDF((g^{a})^{b},\rootkey_{[\asymindexrcvr]}))$}
\NextLine[1]
\ServerAction{\color{blue}$\textbf{end~if}$}
\NextLine[1]
\ServerAction{\color{blue}$\whilesm (\symindexsender > \symindexrcvr)$:}
\NextLine[1]
\ServerAction{\color{blue}$\asymindexsender \gets \nonce.\asymindexsender$:}
\NextLine[1]
\ServerAction{\color{blue}$\msgkey_{[\asymindexsender,\symindexrcvr]},\chainkey_{[\asymindexsender,\symindexrcvr+1]} \gets \KDF(\chainkey_{[\asymindexsender,\symindexrcvr]}) $}
\NextLine[1]
\ServerAction{\color{blue}$ \symindexrcvr++$}
\NextLine[1]
\ServerAction{\color{blue}$\textbf{end~while}$}
\NextLine[1]
\ServerAction{$\AEAD.\Dec({\color{blue}\msgkey_{[\asymindexrcvr,\symindexrcvr]},\nonce}, \ciphertxt,\AD)\to \msg$}
\NextLine[1]
\ServerAction{\color{blue}$\ifsm (\nonce.\asymindexsender \geq \asymindexrcvr)$:}
\NextLine[1]
\ServerAction{\color{blue}$(\diffihell.\KeyGen(1^\secpar) \tor (b,g^{b}))$}
\NextLine[1]
\ServerAction{\color{blue}$(\rootkey_{[\asymindexrcvr+2]}, \chainkey_{[\asymindexrcvr+2,0]}\gets \KDF((g^{a})^{b},\rootkey_{[\asymindexrcvr+2]}) )$}
\NextLine[1]
\ServerAction{\color{blue}$\asymindexrcvr \gets \asymindexrcvr+2$}
\NextLine[1]
\ServerAction{\color{blue}$\textbf{end~if}$} 
\end{tikzpicture}
}
\caption{$\LINEDR$ Letter Sealing double ratchet protocol phases (excluding registration). Given our analysis focuses on session establishment and message exchange, the registration phase is omitted for brevity (see Appendix~\ref{app:lineDR-detail} for details). Text in {\color{blue}blue} highlights our modifications to $\LINEtwo$, notably the integration of double ratchet mechanisms. In $[i,j]$, $i$ denotes the asymmetric ratchet and $j$ the symmetric ratchet maintained by both parties; even $i$ values indicate Alice as the sender, and odd $i$ values indicate Bob.}
\label{fig:line-dblrat-version}
\end{figure}
\mehr{
During session establishment, Alice initializes her state to prepare for secure communication. She sets her root key ($\rootkey$) to empty and sets her asymmetric ($\asymindexsender$) and symmetric ($\symindexsender$) ratchet indices to zero. To establish a shared secret, Alice performs a Diffie-Hellman key generation to produce an ephemeral key pair $(a, g^a)$. She then invokes a Key Derivation Function (KDF) to derive the initial root key $\rootkey_{[0]}$ and the first chain key $\chainkey_{[0,0]}$. This derivation is keyed by a combined Diffie-Hellman shared secret using Alice and Bob’s static pre-shared keys ($g^x$, $g^y$ respectively), ensuring the session is linked to the participants' identities.
When Alice prepares to send a message, she performs a symmetric ratchet step to ensure that each message is protected by a unique key. Alice derives a specific message key ($\msgkey_{[\asymindexsender,\symindexsender]}$) and the subsequent chain key ($\chainkey_{[\asymindexsender,\symindexsender+1]}$) from the current chain key. A unique nonce is then constructed by concatenating the current asymmetric index with a 32-bit random value ($\nonce \gets \asymindexsender \| \text{rand}$). Alice further constructs the Associated Data (AD), which includes identifiers for both parties, the protocol version, her ephemeral public key $g^a$, and the current symmetric index. The plaintext is encrypted using an AEAD scheme. Finally, Alice increments her symmetric index ($\symindexsender{++}$) and transmits the AD, ciphertext, and nonce to Bob.
Bob’s side involves a state-matching mechanism to handle potential out-of-order messages or ratchet advances. If Bob’s root key is uninitialized, he performs the initial KDF steps to synchronize with Alice’s starting state. Upon receipt of a message, Bob checks if Alice has advanced her asymmetric ratchet ($\nonce.\asymindexsender > \asymindexrcvr$). If she has, Bob performs a Diffie-Hellman calculation using his private key $b$ and Alice’s public key $g^a$ to update the root key. Bob then enters a symmetric ratchet loop, iteratively deriving intermediate message keys until his index ($\symindexrcvr$) matches Alice’s sender index. Once synchronized, Bob uses the derived message key and the provided nonce to decrypt the ciphertext. To provide future secrecy, if Alice has advanced, Bob generates a new ephemeral key pair $(b, g^b)$, updates the root key again, and advances his receiver index.}

\section{Formal Security Analysis of $\LINE$}

This section introduces our security definitions for the Multi-Stage Key Exchange (MSKE) model, which underpin our analysis of $\LINE$'s security. We describe a variant of the MSKE model designed to analyse a class of stateless secure messaging protocols lacking FS and PCS security guarantees. The standard MSKE model provides a framework for evaluating the security of key exchange protocols. It builds upon earlier Bellare-Rogaway-style models while extending the security definitions to capture advanced real-world features such as state reveal and corruption. MSKE has been widely used in the cryptographic literature to analyse the security guarantees of various secure messaging protocols, including Signal \cite{signal-1}, and TLS \cite{BenTLS1.3}. Choosing MSKE in our analysis of $\LINE$ allows us to rigorously assess the protocol's resistance to a broad range of realistic adversarial capabilities and threats, providing important security assurances for the millions of users who rely on the $\LINE$ application. For a detailed definition of the MSKE model and an illustrative execution flow, we refer the reader to Appendix \ref{appen:MSKE}.

\subsection{Multi Stage Key Exchange Model} \label{sec:MSKE}

\begin{figure}[htbp!]
\centering
\begin{adjustbox} {
				max width=0.9\columnwidth,max height=0.6\textheight}
			\input{mskeExp}
		\end{adjustbox}
    \caption{Pseudocode description of the MSKE Experiment}
    \label{fig:MSKEgame}
\end{figure}

A Multi-Stage Key Exchange protocol is a two-party protocol that securely establishes a sequence of shared secret keys among communicating parties over an insecure communication channel, ensuring that only the intended parties can derive the same secret key without revealing it to unauthorised parties or adversaries. The security of MSKE is captured as a game played between an adversary and a challenger simulating parties executing an MSKE protocol. In the game, the adversary interacts with the challenger's MSKE executions via oracle queries that represent and abstract the capabilities of real-world attackers. For an MSKE protocol, the adversary aims to distinguish between a ``real'' session key and a random key (as provided in the $\Test$ query) where breaking key indistinguishability implies \emph{some} adversarial knowledge of the keys established in the MSKE protocol. If the adversary can do this with any non-negligible advantage (i.e., better than random guessing), then the MSKE protocol is considered broken. Once we prove key indistinguishability for a given MSKE, it is trivial to prove message confidentiality and authentication (when the key is used in an authenticated encryption scheme, as in $\LINEtwo$).

A high-level description of a potential security game for an MSKE protocol follows; a detailed description is shown in Figure \ref{fig:MSKEgame}. The game is broadly divided into three phases: Setup, Queries and End-of-Game. During the setup phase, the challenger initializes the game, sets-up any necessary parameters, and creates key pairs for the participating entities. The adversary can then interact with the challenger via a series of adversarial queries, allowing them to direct the execution of the MSKE between many simulated parties. Finally, the adversary outputs a bit guessing whether a given session key was honestly computed, or instead randomly sampled from the same distribution. 

During the experiment, the challenger maintains a series of variables and session states. Each session $\session$ that participates in the execution of the security model maintains the following session variables:

\begin{itemize}
\item $\role \in \{\init,\resp\}$: the role of the session.
\item $\session.\peerid \in [1,\ldots,\numParties]$: the identifier of the party owner of the partner session. 
\item $\session.peerpk$: the peer's long-term public key. 
\item $\pi.\status[s]\in{\inprogress,\accepted,\rejected}$: the current state of the session $\pi$ in stage $s$, set to $\inprogress$ upon start of a new stage, and set to $\accepted$ or $\rejected$ during the protocol execution.  
\item $\session.k[s] \in \KeySpace$: the session key output by session $\session$ in  stage $s$. 
\item $\session.rand[s] \in \bits^*$: the randomness used by session $\session$ in stage $s$. 
\item $\session.st[s]$: any additional protocol state maintained by the session in stage $s$. 
\item $\session.T[s]$: the transcript of messages sent and received by session $\session$ in stage $s$. 
\end{itemize}

Here we describe the queries the adversary can make to the challenger. These oracles model different kinds of information that the adversary might gain access to in a real-world scenario. We thus capture an active adversary in control of the network, who attempts to compromise sessions via the following oracle queries:
\begin{itemize}
\item $\Test(u,i,s)$: Returns either a real or random key to the adversary from a session $\testsess$ in stage $s$.
\item \mehr{$\RevSessKey(u,i,s)\to \testsess.\key[s]$: Returns the session key used in a specified $\testsess$ in stage $s$. This models the scenario where the session key of a participant is compromised.} 
\item $\RevLongTermKey(u)\to \sk_u$: Returns the long-term secret key $\sk_u$ of a specified entity $u$. This models the scenario where the long-term key of a participant is compromised. 
\item $\RevRand(u,i,s) \to \testsess.rand[s]$: Returns the random values (such as $\nonce$ or $\salt$) used in a specified session $\testsess$ in stage $s$. This models the case where the adversary gets to see some of the random values used during the key exchange. 
\item $\RevState(u,i,s) \to \testsess.st[s]$: Returns the internal state of a specified entity in a given session. This might include intermediate values and computations. 
\item $\Send(u,i,m) \to m'$: Sends the message $m$ to the session $\testsess$, which processes the message and (potentially) returns a message $m'$. The oracle then returns the response $m'$. This allows the adversary to completely control communication between sessions, capturing an untrusted channel.     
\end{itemize}
 
Finally, the adversary terminates and outputs a guess bit $b' \in \bits$. If the guess matches the hidden coin flip from the $\Test$ query, then the adversary is able to distinguish the real key from a randomly sampled key, and the adversary wins. We formalise this notion in Definition \ref{def:ms-ind}.

\begin{definition} \label{def:ms-ind}
    Let $\Pi$ be a multi-stage key exchange protocol. Let $\numParties,\numMsgs,\numSessions \in  \mathbb{N}$ be the maximum number of parties, messages and sessions in the security game. Let $\adversary$ be a probabilistic algorithm that runs in time polynomial in the security parameter. We define the adversary's advantage in winning the $\msind$ game as:
 %\begin{align*}
   $\Adv{\msind}{\Pi,\numParties,\numMsgs,\numSessions}(\adversary) = 
  |2 \cdot \Pr[\Exp{\msind}{\Pi,\numParties,\numSessions,\numMsgs}(\adversary) = 1 ] - 1|$
 %\end{align*}
 where the security experiment $\Exp{\msind}{\Pi,\numParties,\numMsgs,\numSessions}(\adversary)$ is as defined in Figure \ref{fig:MSKEgame}. We say that an MSKE protocol $\Pi$ is secure if $\Adv{\msind}{\Pi,\numParties,\numMsgs,\numSessions}(\adversary)$ is negligible.
\end{definition}

Note that the challenger in $\Exp{\msind}{\Pi,\numParties,\numMsgs,\numSessions}(\adversary)$ maintains the following set of global variables, used to determine whether the adversary has trivially broken the protocol via the freshness predicate $\freshpredicate$:
\begin{itemize}
    \item $b$: a challenge bit;
    \item $\tested = (u,i,s)$ or $\bot$: records the inputs to the query $\Test(u,i,s)$ or $\bot$ if no $\Test$ query has been issued;
   \item $\testsess.\revsesskey[s] \in \{\trueflag,\falseflag\}$: records whether $\RevSessKey(u,i,s)$ was issued, set by default to $\falseflag$.
   %\ben{This isn't referred to in the pseudocode, nor in the write-up. Should this be here? Discuss.}
    \item $\testsess.\revrand[s] \in \{\trueflag,\falseflag\}$: records whether $\RevRand(u,i,s)$ was issued, set by default to $\falseflag$.
    \item $\testsess.\revstate[s] \in \{\trueflag,\falseflag\}$: records whether $\RevState(u,i,s)$ was issued, set by default to $\falseflag$.
\end{itemize}

We now introduce \emph{freshness conditions}, which restrict the adversary from issuing arbitrary queries and specify the combinations of queries under which an MSKE---and in particular the $\LINE$ messaging protocol---remains secure. We first present the freshness conditions and analysis for $\LINEtwo$, and then proceed to the freshness conditions and analysis for $\LINEDR$.
\begin{figure}[t]
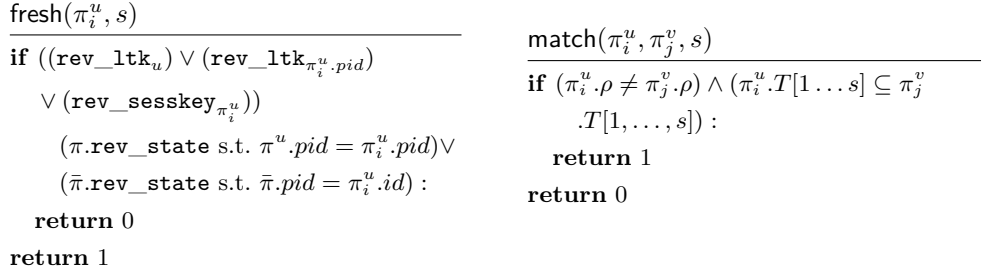

\centering
\noindent
\begin{minipage}{0.5\textwidth}
\vspace{-0cm}
  \procedureblock{$\freshpredicate(\testsess,s)$}
    {\small
    \pcif ((\revltk_u) \lor (\revltk_{\testsess.\peerid}) \\ \t \lor (\revsesskey_{\testsess})) \\ 
    \t \t (\session.\revstate ~\mathrm{s.t.}~\session^u.\peerid = \testsess.\pid) \lor\\
    \t \t  (\bar{\session}.\revstate ~\mathrm{s.t.}~\bar{\session}.\peerid = \testsess.\id):\\
    \t \pcreturn 0\\
    \pcreturn 1\\
    }
    %\caption{$\LINE$ Freshness Predicate}
    %\label{fig:fresh}
\end{minipage}%
\begin{minipage}{0.5\textwidth}
\procedureblock{$\matchpredicate(\testsess,\partsess,s)$}
   {\small
    \pcif (\testsess.\role \neq \partsess.\role) \land (\testsess.T[1\ldots s] \subseteq \partsess \t \t \\
    \t \t .T[1,\ldots,s]):\\
    \t \pcreturn 1\\
    \pcreturn 0\\
    \t \\
  } 
\end{minipage}
\caption{Left: $\LINE$ Freshness Predicate and Right: $\LINE$ Matching Predicate}
\label{fig:fresh+match}
\end{figure}

\subsection{Security of $\LINEtwo$}

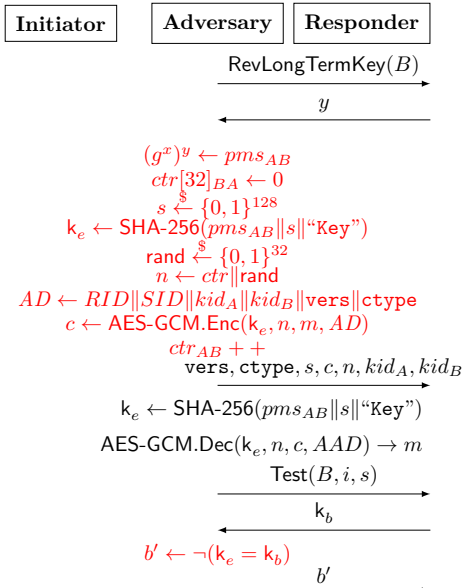
\begin{wrapfigure}{l}{0.5\textwidth}
    \centering
    \scalebox{0.8}{
\centering
\begin{tikzpicture}[yscale=-0.4,>=latex]
\edef\InitX{0}
\edef\ArrowLeft{0}
\edef\ArrowCenter{3.5}
\edef\ArrowRight{7}
\edef\AdvX{3.5}
\edef\RespX{7}
% Set the starting Y coordinate
\edef\Y{0}

% Draw header boxes
\node [rectangle,draw,inner sep=5pt,right] at (\InitX,\Y) {\textbf{Initiator} };
\node [rectangle, draw, inner sep=5pt, centered] at (\AdvX,\Y) {\textbf{Adversary}};
\node [rectangle,draw,inner sep=5pt,left] at (\RespX,\Y) {\textbf{Responder}};

\NextLine[2]
\AdversaryToServer{$\RevLongTermKey(B)$}{}
\NextLine[0.5]
\ServerToAdversary{$y$}{}
\NextLine[1]
\AdversaryAction{\color{red}$(g^x)^y \gets \pms_{AB}$}
\NextLine[1]
\AdversaryAction{\color{red}$\ctr[32]_{BA} \gets 0$}
\NextLine[1]
\AdversaryAction{\color{red}$\salt \getsr \{0,1\}^{128}$}
\NextLine[1]
\AdversaryAction{\color{red}$\key_e \gets \SHAtwo(\pms_{AB} \| \salt \| ``\keylbl" )$}
\NextLine[1]
\AdversaryAction{\color{red}$\rand \getsr \bits^{32}$}
\NextLine[1]
\AdversaryAction{\color{red}$\nonce \gets \ctr \| \rand$}
\NextLine[1]
\AdversaryAction{\color{red}$\AD \gets \RID \| \SID \| \kid_A \| \kid_B \| \vers \| \ctype$}
\NextLine[1]
\AdversaryAction{\color{red}$\ciphertxt \gets \AESGCM.\Enc(\key_e, \nonce, \msg, \AD)$}
\NextLine[1]
\AdversaryAction{\color{red}$\ctr_{AB}++$}
\NextLine[1]
\AdversaryToServer{$\vers,\ctype, \salt, c, \nonce, \kid_A, \kid_B$}{}
\NextLine[0.5]
\ServerAction{$\key_e \gets \SHAtwo(\pms_{AB} \| \salt \| ``\keylbl" )$}
\NextLine[1.5]
\ServerAction{$\AESGCM.\Dec(\key_e,\nonce,c,\AAD)\to m$}
\NextLine[1.5]
\AdversaryToServer{$\Test(B,i,s)$}{}
\NextLine[0.5]
\ServerToAdversary{$\key_b$}{}
\NextLine[0.5]
\AdversaryAction{\color{red}$b' \gets \lnot(\key_e = \key_b)$}
\NextLine[1]
\AdversaryToServer{$b'$}{}
\end{tikzpicture}
}
    \vspace{-0.5cm}
    \caption{A trivial KCI attack.}
    \label{fig:ATTACK}
\end{wrapfigure}

\subsubsection{Matching and Freshness Predicates of $\LINEtwo$}
Previously, we provided the adversary with powerful queries with which it can break any protocol. We now define a \textit{freshness predicate} $\freshpredicate$ to constrain that power, effectively specifying the details of the threat model. The restriction on these adversarial queries is not arbitrary; allowing them would enable the adversary to trivially break the security of MSKE. Thus, we assume any successful adversary does not violate freshness predicates to prevent such trivial attacks. 

To demonstrate, we allow the adversary in Figure~\ref{fig:ATTACK} to violate the freshness predicates to perform an example Key Compromise Impersonation attack. The adversary (actions in \textcolor{red}{red}) issues a $\RevLongTermKey$ query to the responder, computes the initiator's encryption key, and forges a message to the responder, thereby breaking key indistinguishability. When defining $\freshpredicate$ we aim to capture the strongest security conditions provable for each of $\LINE$'s message keys given the protocol design---i.e., relative to the maximal combinations of secrets learned by the adversary. We then use the protocol structure to identify attacks that are impossible to prevent and exclude them by appropriately restricting the adversary.

Thus, $\LINE$ does not provide KCI resilience, nor PCS: if an adversary obtains the long-term keys of either party, they would be able to impersonate that party indefinitely. Thus, our freshness condition forbids the adversary from issuing long-term key reveals to the session owner or its communicating partner, as shown at left of Figure \ref{fig:fresh+match}. Additionally, once the first message is sent and received by either session, the pre-master secret $\pms$ is computed and stored locally, which is the only secret value used to generate the session keys used by the sessions. Thus, we must also forbid the adversary from issuing a $\RevState$ query to any session $\session$ such that $\session.\id = i$ and $\session.\peerid = \testsess.\peerid$ (or $\session.\id = \testsess.\peerid$ and $\session.\peerid = \session.\id$, respectively.). We formalize this notion in Definition \ref{def:fresh-LINE}.

\begin{definition}[$\LINE$ Freshness Predicate]
\label{def:fresh-LINE}
    Let $\session$ and $\bar{\session}$ be any session such that $\session.\id = \testsess.\id $ and $\bar{\session}.\id = \testsess.\pid$, respectively. For a session $\testsess$ in stage $s$ we say that $\testsess$ is  \emph{fresh} if $\freshpredicate(\testsess,s)=1$.      
\end{definition}

 However, we also need to restrict the adversary from issuing similar queries to its communicating partner, which we identify via a \textit{matching} condition. On a high level, a matching session $\session'$ to a session $\session$ is any session where the corruption of secrets of $\session'$ would impact $\session$. We use a prefix-matching definition here: a matching session is one where the transcript of messages maintained by the sessions are equal (to a certain point, in order to capture sessions where the adversary simply drops the final message) as shown at right of Figure \ref{fig:fresh+match}.

\begin{definition}[$\LINE$ Matching Predicate]\label{def:match}
For a session $\testsess$, we say that $\partsess$ is a \emph{matching session} in stage $s$ (or $\partsess$ \emph{matches} $\testsess$) if $\matchpredicate(\testsess,\partsess,s)=1$. 
\end{definition}

In Theorem \ref{thm:line-secure} we prove that $\LINEtwo$ indeed achieves key-indistinguishability security against an adversary limited by our freshness predicate (Definition \ref{def:fresh-LINE}). Our definitions of cryptographic primitives can be found in the Appendix \ref{appen:CP}.

%\iffullversion
\begin{theorem}[Key Indistinguishability Security of $\LINEtwo$]\label{thm:line-secure}
Let $\LINEtwo$ be the $\LINE$ protocol version 2 described in Figure \ref{fig:line-vers-2}. Assuming that the $\ddh$ assumption holds, modelling $\SHAtwo$ as a $\PRF$\footnote{Following the methodology of Bergsma et al. \cite{PRF-def}, we model the keyed use of SHA-256 as a PRF. One could also model $\SHAtwo$ as a random oracle \mehr{(as in the proof of Theorem \ref{thm:linedr-secure})}, to which the proof proceeds near identically, without needing a hash collision-resistance replacement.}, and $\SHAtwo$ is collision-resistant, then $\LINEtwo$ is $\msind$-secure. In particular, there exist algorithms $\bdversary_1, \ldots, \bdversary_3$ described in the proof of the theorem, such that, for all PPT adversaries $\adversary$, we have: $\Adv{\msind}{\LINE, \numParties,\numSessions, \adversary}  \leq \Adv{\coll}{\Hash,\adversary} (\bdversary^{\adversary}_1) + \numParties^2 \cdot \Big(\Adv{\ddh}{g,p,\adversary} (\bdversary^{\adversary}_2) + \Adv{\prf}{\PRF,\adversary} (\bdversary^{\adversary}_3)\Big)$.
\end{theorem}

\begin{proof}
From Theorem \ref{thm:line-secure}, $\LINE$ is $\msind$-secure if we show that the advantage \\ $\Adv{\msind}{\LINE, \numParties,\numSessions, \adversary}$of any adversary ($\adversary$) of breaking the key indistinguishability security of $\LINE$  is negligible. 

In what follows, our proof proceeds as a series of game hops. In each game hop, we make an incremental change to the initial security game. In the final game, we show that the computed session keys are already uniformly random keys independent of the protocol flow. As a result, the adversary learns nothing by being returned a real-or-random key, and thus, their advantage is equal to simply guessing. We define $\Adv{\msind}{\gameref{X},\adversary}(\lambda)$ as the advantage of the adversary $\adversary$ in Game $X$.

\game{0} This is the original multi-stage key exchange game described in Definition \ref{def:ms-ind}. Thus we have $\Adv{\msind}{\LINE,\numParties,\numSessions}(\adversary) = \Adv{\msind}{\gameref{0},\adversary}(\lambda)$.

\game{1} In this game, we exclude hash collisions. Specifically, we introduce an abort event $\gameabort{1}$ that triggers whenever the challenger $\challenger$ would compute two hash outputs that collide on different inputs. We define reduction $\bdversary_1$ that acts identically to \gameref{0}, but at the beginning of the experiment, $\bdversary_1$ initiates a hash collision challenger $\challenger_{\coll}$. Whenever the reduction $\bdversary_1$ must compute a hash, $\bdversary_1$ also maintains a list $(in, \Hash(in))$. For each hash $\bdversary_1$ decides if there exists an entry $(in',\Hash(in'))$ such that $in \neq in'$ but $\Hash(in) = \Hash(in')$. If such a collision occurs, $\bdversary_1$ sends $in,in'$ to the hash collision challenger $\challenger_{\coll}$ and aborts the game. Note that the only change introduced to the security game execution is if a hash collision occurs. Thus $\Pr(\gameabort{1}) = \Adv{\coll}{\Hash}(\adversary)$, and we have 
$\Adv{\msind}{\gameref{0},\adversary}(\lambda) \leq \Adv{\msind}{\gameref{1},\adversary}(\lambda) + \Adv{\coll}{\Hash}(\adversary)$.

\game{2} In this game we guess the party owners of our $\Test(u,i,s)$ query, and the session $\testsess$ and stage $s$. Specifically, we introduce an abort event $\gameabort{2}$ that triggers whenever we guess incorrectly. At the beginning of the experiment we guess indexes $u \in [\numParties]$, $v \in [\numParties]$, $i \in [\numSessions]$, $s \in [\numStages]$ and abort if $\Test(u^*,i^*,s^*)$ is issued such that $\pi^{u^*}_i.\pid \ne v$ and $u \ne u^*$, $i^* \neq i$, $s^* \neq s$. Since there are at most $\numParties$ parties, $\numSessions$ sessions and $\numStages$ stages, this reduces the adversary's advantage of success by a factor of $\numParties^2 \cdot \numSessions \cdot \numStages$. Thus we have $\Adv{\msind}{\gameref{1},\adversary}(\lambda) \leq \numParties^2 \cdot \numSessions \cdot \numStages \cdot \Adv{\msind}{\gameref{2},\adversary}(\lambda)$. Note that from what follows, we know that by the freshness predicate $\freshpredicate$, the adversary cannot issue either $\RevLongTermKey(u)$ or $\RevLongTermKey(v)$ without violating $\freshpredicate$ and gaining no advantage. 

\game{3}
In this game, we replace the value $\pms = g^{xy}$ computed by the test session $\testsess$ and its matching partner (if one exists) with a random value $\rreplace{pms}$. By the definition of the MSKE game, the adversary cannot maliciously inject or modify the long-term keys of honest protocol participants. Any adversary $\adversary$ that can distinguish this game from the previous game can directly be used to construct an adversary $\bdversary_3$
 that can break the $\DDH$ assumption: let $(g, g^x, g^y , g^z)$ be the $\DDH$ challenge. We set the long-term key of $u$ as $g^x$, the long-term key of $v$ as $g^y$ and $\rreplace{\pms}=g^z$. If $z = xy$, then we
have $\pms = \rreplace{\pms}$, and we are in \gameref{2}, otherwise we are in \gameref{3}. Thus $\Adv{\msind}{\gameref{2},\adversary}(\lambda) \leq \Adv{\msind}{\gameref{3},\adversary}(\lambda) + \Adv{\ddh}{g,p,q}(\adversary) $.

\game{4}
In this game, we replace the honestly computed encryption keys $\key_e$ of $u$, $v$ sessions with uniformly random values $\widetilde{k_e}$ by modelling the hash function $\SHAtwo$ as a $\PRF$.
Specifically, we define a reduction $\bdversary_4$ that acts identically to the challenger $\challenger$ in \gameref{3}, with the following changes: At the beginning of the experiment, $\bdversary_4$ initialises a $\PRF$ challenger $\challenger_{\PRF}$. Whenever $\bdversary_4$ needs to compute a encryption key $\key_e = \Hash({\pms'}||\salt||\keylbl)$ such that $\widetilde{\pms}_{u,v} = \pms'$ (i.e. it is a session owned by $u$ (or v, respectively) with $\testsess.\pid = v$ (or $\partsess.\pid = u$, respectively), $\bdversary_4$ instead uses its $\PRF$ oracle for the derivation of $\key_e$, submitting the query $(\salt \|\keylbl)$ to $\challenger_{\PRF}$, and replaces the computation of $\key_e$ with the output of $\challenger_{\PRF}$, $\widetilde{\key_e}$. 
When computing $\key_e$ for the test session $\testsess$ or its matching partner $\partsess$ (if $\matchpredicate(\testsess,\partsess,s)=1$), $\bdversary_4$ issues the challenge query $(\salt \|\keylbl)$ to $\challenger_{\PRF}$, and replaces the computation of $\key_e^s$ with the output of $\challenger_{\PRF}$, $\widetilde{\key^s_e}$. Observe that, if the bit $b$ sampled by $\challenger_{\PRF}$ is $b=0$, then $\widetilde{k^s_e}$  is computed honestly, and $\widetilde{k^s_e} = \Hash(\widetilde{\pms}_{u,v},\mathit{salt} \|\mathtt{key})$, and thus we perfectly simulate \gameref{3}. 
Otherwise, $\widetilde{k^s_e} \getsr \bits^{|\PRF|}$, and we are in \gameref{4}. 
The simulation is sound because $\widetilde{\pms}_{u,v}$ is already uniformly random and independent of the protocol flow by the change in \gameref{3}. 
Any adversary that can detect the change between \gameref{3} and \gameref{4} can be used by the reduction $\bdversary_4$ to break the $\prf$ security of $\SHAtwo$, and thus we have:
$\Adv{\msind}{\gameref{3},\adversary}(\lambda) \leq \Adv{\msind}{\gameref{4},\adversary}(\lambda) + \Adv{\prf}{\PRF}(\adversary)$.

We note now that the response to the $\Test(u,i,s)$ query issued by the adversary $\adversary$ ($\widetilde{k^s_e}$) is now uniformly random and independent of the protocol flow regardless of the bit $b$ sampled by the challenger $\challenger$ at the beginning of the game, and thus $\Adv{\msind}{\gameref{4},\adversary}(\lambda) = 0$.
\end{proof}

At this point, we can trivially extend our proof to show that $\LINEtwo$ achieves message authentication and message confidentiality. We showed that any key $\key_e$ used to encrypt the messages from $\LINEtwo$ is indistinguishable from a randomly-sampled key from the same distribution (via Theorem \ref{thm:line-secure}). These keys are then used in $\AESGCM$ to encrypt message content. Thus, assuming that $\AESGCM$ achieves $\AEAD$ security, any adversary that could forge a message against $\LINEtwo$ can thus be trivially transformed into an efficient adversary against the $\auth$-security of $\AESGCM$, and any adversary that could distinguish encrypted messages could be used to break $\indcca$-security. 

However, as Kampankis \etal~\cite{practicalaesgcm} highlight, $\AESGCM$ is not resilient with respect to $\IV$ misuse and re-use. We note here that $\LINEtwo$ uses a semi-deterministic $\IV$ construction: the $\nonce$ used as $\IV$ in $\AESGCM$ is a 64-bit counter concatenated with a 32-bit random value. Thus, users must send at least $2^{64}$ messages before any $\nonce$ re-use is possible in their own computed ciphertexts. However, we note that \emph{both sides} initialise their counter at 0, and use the same $\pms_{AB}$ as the basis for their key derivation. This means collisions on $\ctr_{AB}$ and $\ctr_{BA}$ are common. Thus, to provide an upper bound on the probability of collision, we need to consider the event that either party collides on both $\salt \in \bits^{128}$ and $\rand \in \bits^{32}$. Since we model the keyed derivation of $\key_e$ using $\SHAtwo$ as a collision-resistant hash function, colliding on $\salt$ will derive the same key but will generate distinct outputs otherwise. Functionally, the event that either party collides on both $\salt \in \bits^{128}$ and $\rand \in \bits^{32}$ is the same probability of colliding on a single $160$-bit string. In order to prove message authentication, we reason about the probability of such a collision in Theorem \ref{thm:line-auth-secure}.

%\iffullversion
\begin{theorem}[Message Authentication of $\LINE$]\label{thm:line-auth-secure}
Let $\LINEtwo$ be the protocol described in Figure \ref{fig:line-vers-2}. Assuming that the $\ddh$ assumption holds, $\SHAtwo$ is $\prf$-secure, $\SHAtwo$ is collision-resistant, the $\AEAD$ scheme is $\auth$-secure, then $\LINEtwo$ achieves message authentication. There exist algorithms $\bdversary_1, \ldots, \bdversary_4$ described in the proof of the theorem, such that, for all PPT adversaries $\adversary$, we have: $\Adv{\auth}{\LINE, \numParties,\numSessions, \adversary}  \leq  \Adv{\coll}{\Hash,\adversary} (\bdversary^{\adversary}_1) + \numParties^2 \cdot \numSessions \cdot \numStages \cdot (\Adv{\ddh}{g,p,\adversary} (\bdversary^{\adversary}_2) 
+ \Adv{\prf}{\PRF,\adversary} (\bdversary^{\adversary}_3)
+ \Adv{\auth}{\AEAD,\adversary} (\bdversary^{\adversary}_4) + 2\cdot\numStages \cdot 2^{-160})$.
\end{theorem}
\begin{proof}
    Our proof begins as in the proof of Theorem \ref{thm:line-secure}, and proceeds identically to \gameref{4}. Thus, we have:
    $\Adv{\auth}{\gameref{3},\adversary}(\lambda) \leq \Adv{\auth}{\gameref{4},\adversary}(\lambda) + \Adv{\prf}{\PRF}(\adversary)$.

    \game{5}
    In this game, we introduce an abort event that triggers when the guessed party $i$ collides on the same $\salt$, $\rand$ strings in some stage $s$ as their guessed communicating partner $j$ sharing $\pms_{ij}$. Since we guessed the stage, session and communicating parties in which $\adversary$ broke message authentication, this is simply the probability that $i$ collides with $j$ in any of their stages and thus:
    $\Adv{\auth}{\gameref{4},\adversary}(\lambda) \leq \Adv{\auth}{\gameref{5},\adversary}(\lambda) + 2\times\numStages \cdot 2^{-160}$.

    \game{6}
    In this game, we introduce an abort event that triggers when any session decrypts a ciphertext validly using $\rreplace{\key_e}$, but the ciphertext was not produced by $\testsess$, thus breaking message authentication. Specifically, we introduce a reduction $\bdversary_4$ that acts identically to the challenger in \gameref{5}, except when encrypting using $\rreplace{\key_e}$. Instead, $\bdversary_4$ initializes an $\AEAD$ $\auth$ challenger $\challenger_{\auth}$, and queries the plaintext to be encrypted to $\challenger_{\auth}$ instead. Note that by \gameref{4} $\rreplace{\key_e}$ is a uniformly random and independent value, which by \gameref{5} never collides with another key using the same $\IV$. We note that it is possible for a key to collide within the same session: in this case, $\bdversary_4$ will simply forward the plaintext to $\challenger_{\auth}$. If the abort event triggers, then the adversary has managed to forge an $\AEAD$ ciphertext that was not output by $\challenger_{\auth}$, and thus we can submit the ciphertext to win the $\AEAD$ $\auth$ security game. Thus we have:
    $\Adv{\auth}{\gameref{6},\adversary}(\lambda) \leq \Adv{\auth}{\gameref{6},\adversary}(\lambda) + \Adv{\auth}{\AEAD}(\adversary)$.
    Since we now abort before $\adversary$ triggers the winning condition, $\Adv{\auth}{\gameref{6},\adversary}(\lambda)=0$ and we find our bound.
 \end{proof}

\mehr{We emphasize that our security model is designed to maintain close correspondence to the $\LINEtwo$ protocol's actual deployment, including its inherent limitations. Consequently, the model explicitly permits replay attacks: replayed ciphertexts do not violate freshness conditions and therefore remain admissible in the model. However, the matching condition checks for matching messages within a transcript and thus disallows reveal keys for \textit{any} replayed ciphertexts. This concretely reflects $\LINEtwo$'s lack of replay resilience. By capturing these weaknesses rather than abstracting them away, our model provides an assessment of the protocol's actual security properties.}

\subsection{Security of $\LINEDR$}

\subsubsection{Matching and Freshness Predicates of $\LINEDR$} We now turn to analysing the security of our new $\LINEDR$ variant. We first introduce new matching and freshness conditions that capture the improved security guarantees of $\LINEDR$.

\begin{wrapfigure}{r}{0.35\textwidth}
\centering \procedureblock{$\freshpredicate(\testsess,s)$}
    {\pcif \valid(u,i,s) \land \freshpredicate_{\tau}(u,i,s) \\
    \quad \pcreturn 1
    }
    \caption{$\LINE$ Double Ratchet Freshness Predicate}
    \label{fig:FreshDR}
\end{wrapfigure}
Since our construction is inspired by the double ratchet protocol, our freshness conditions for $\LINEDR$ is a modification of the general multi-stage AKE model of Cohn-Gordon et al. for Signal \cite{signal-1}.
A test session $\testsess$ at stage $s$ is considered fresh if and only if it is both $\valid$ and $\freshpredicate$ as shown in Figure \ref{fig:FreshDR}. The validity predicate $\valid(u,i,s)$ prevents the most trivial attack: it ensures that the adversary cannot win key-indistinguishability simply by revealing the session key of the tested stage or of any partnered session.
\begin{align*}\valid(u,i,s) \equiv (\testsess.\status[s] = \accepted) \land (\neg \testsess.\revsesskey[s]) \land \\
    (\forall j: \matchpredicate(\testsess, \partnersess, s) = 1 \Rightarrow \neg\partnersess.\revsesskey[s])
\end{align*}

Within $\freshpredicate$ we introduce a per-stage freshness predicate $\freshpredicate_{\tau}$. Since $\LINEDR$, like Signal, consists of three types of stages (initial, asymmetric, and symmetric stages depending on the ratcheting that occurs), $\freshpredicate_{\tau}$ branches similarly i.e. $\tau \in \{\initial,\asym,\sym\}$. To capture this complexity, we modify our stage counter $s$: instead of simply an integer as before, our stage counter is now a pair of integers $s=[x,y]$ (where $x,y \in \mathbb{N}$), $x$ indexes the current asymmetric ratchet and $y$ indexes the current symmetric ratchet. Thus $s=[0,0]$ indicates the $\initial$ stage, $s=[x,0]~ (x > 0)$ indicates an $\asym$ stage, and finally $s=[x,y]~ (y > 0)$ indicates a $\sym$ stage. We now describe the freshness predicate $\freshpredicate_{\tau}$.

During the \textbf{Initial Stage} ($s = [0,0])$ we require that neither of the long-term keys have been revealed ($\freshpredicate_{\LL}$) or that neither the initiator's ephemeral key nor the responder's long-term key has been revealed ($\freshpredicate_{\EL}$). This is captured in $\freshpredicate_{\initial}(u,i,[0,0]) = \freshpredicate_{\LL}(u,i) \lor \freshpredicate_{\EL}(u,i,[0,0])$ where $\freshpredicate_{\LL}(u,i) = \neg \revltk_u \land \neg\revltk_{\testsess.\pid}$, and when $\matchpredicate(\testsess,\partsess,s)=1$:
$$\freshpredicate_{\EL}(u,i,s) = \begin{cases}
    \pcif \testsess.\role =\init: \neg\testsess.\revrand[0,0] \land \neg \revltk_v\\
    \pcif \testsess.\role = \resp: \neg\partnersess.\revrand[0,0] \land \neg \revltk_u
\end{cases}$$

At the \textbf{Asymmetric Stages} ($s=[x,0], x>0$), we require that neither of the ephemeral secrets sampled by either party used in the asymmetric stage is revealed ($\freshpredicate_{\EE})$, \emph{or} the previous asymmetric stage was fresh (i.e. $\freshpredicate_{\asym}(u,i,[x-1,0])$ when $x>1$, or $\freshpredicate_{\initial}(u,i,[0,0])$), \emph{and} neither party's session state has been revealed at this stage ($\freshpredicate_{\state}$). Thus, $\freshpredicate_{\state}(u,i,[x,y]) = \neg\testsess.\revstate[x,y] \land  
(\forall j : \matchpredicate(\testsess , \partnersess , [x,y]) = 1 \Rightarrow \neg\partnersess.\revstate[x,y]).$ %(i.e. $s=[1,0]$). $\freshpredicate_{\tau}(u,i,s) = \freshpredicate_{\asym}(u,i,[1,0])$

\begin{comment}
    $ \freshpredicate_{\asym}(u,i,[1,0])= \begin{cases}
    \quad if \testsess.\role = \init : \\
    \quad \quad \neg \testsess.\revrand[0,0] \land \\
    \quad \quad \neg \partnersess.\revrand[1,0] \\
    \quad if \testsess.\role = \resp : \\
    \quad \quad \neg \testsess.\revrand[1,0] \land \\
    \quad \quad \neg \partnersess.\revrand[0,0] \\
    \quad \lor (\freshpredicate_{\initial}(u,i,[0,0]) \\
    \quad \land \freshpredicate_{\state}(u,i,[1,0]))\\
\end{cases}  $

where

\textbf{Non-initial Asymmetric Stage where $s=[x,0] \land x\ge2$ :} Freshness at these stages is defined recursively, ensuring that each new asymmetric key exchange derives from unrevealed randomness or a clean prior state. Specifically,$\freshpredicate_{\tau}(u,i,s) = \freshpredicate_{\asym}(u,i,[1,0])$
here 
\end{comment}

$$ \freshpredicate_{\asym}(u,i,[x> 0,0]) = \freshpredicate_{\EE}(u,i,[x,0]) \lor 
\begin{cases} 
\freshpredicate_{\state}(u,i,[x-1,0])) \land\\
x>1: (\freshpredicate_{\asym}(u,i,[x-1,0])\\
x=0: (\freshpredicate_{\initial}(u,i,[0,0])
\end{cases}
$$

In $\LINEDR$ parties "take turns" to perform an asymmetric ratchet (for example, responders ratchet in odd stages), thus the ephemeral secrets our freshness predicate checks depend on the session's role. If the current stage index $x$ is even, we require that the initiator's current ephemeral secret and the responder's previous ephemeral secret have not been revealed. If odd, vice-versa. Let $b = (\testsess.\role == \init) \oplus (x\%2==0)$, and let $\partsess$ be any session such that $\matchpredicate(\testsess,\partsess,x)=1$, then
$\freshpredicate_{\EE} = \neg \testsess.\revrand[x-b,0] \land \neg \partsess.\revrand[x-\lnot b,0]$.

In \textbf{Symmetric Stages} ($s=[x,y>0]$), freshness is also defined recursively. Since the symmetric stage session keys are derived deterministically from the previous symmetric stage's state, we require that the previous stage is fresh and that the session state has not been revealed. Logically, this eventually requires that either the initial stage is fresh ($\freshpredicate_{\initial}$), or that the current asymmetric stage is fresh ($\freshpredicate_{\asym}$).  Specifically, $\freshpredicate_{\sym}(u,i,[x,y])= \freshpredicate_{\state}(u,i,[x,y-1]) ~\land \freshpredicate_{\sym}(u,i,[x,y-1])$.

\begin{theorem}[Key Indistinguishability Security of $\LINEDR$]\label{thm:linedr-secure}
Let $\LINEDR$ be the $\LINE$ double ratchet protocol described in Figure \ref{fig:line-dblrat-version}. Assuming that the hash function is collision-resistant, the $\GDH$ assumption holds and assuming all $\KDF$s are random oracles, then $\LINEDR$ is $\msind$-secure.
\end{theorem}

\begin{proof}

From Theorem \ref{thm:linedr-secure}, $\LINEDR$ is $\msind$-secure if we show that the advantage $\Adv{\msind}{\LINEDR,\numParties,\numSessions,\adversary}$ of any adversary ($\adversary$) of breaking the key indistinguishability security of $\LINEDR$ is negligibly small.
In what follows, our proof proceeds as a series of game hops. In each game hop, we make an incremental change to the initial security game. In the final game, we will be able to show that the computed session keys are already uniformly random keys independent of the protocol flow. As a result, the adversary learns nothing by being returned a real-or-random key, and thus, their advantage is equal to simply guessing. We define $\Adv{\msind}{\gameref{X}, \adversary}(\lambda)$ as the
advantage of the adversary $\adversary$ in $\gameref{X}$.
\game{0} 
This is the original multi-stage key exchange game described in Definition \ref{def:MSKE}. Thus we have $\Adv{\msind}{\LINEDR, \numParties,\numSessions, \adversary} (\lambda) = \Adv{\msind}{G0, \adversary} (\lambda) $.
\game{1}
In this game, we ensure that no collisions occur among honestly generated Diffie-Hellman ($\diffihell$) public keys. Specifically, the challenger $\challenger$ maintains a list $\PKL$ of all $\diffihell$ public values (for $g^x, g^y, g^a, g^b$) honestly generated during the game. If a $\diffihell$ public value appears twice, $\gameabort{1}$ event occurs and the adversary automatically loses. Otherwise we add the $\diffihell$ public key to $\PKL$: $(\ifsm ~\diffihell.\KGen \in \PKL: \gameabort{1}  ~\textbf{else}: \diffihell.\KGen \tou \PKL)$.

For an adversary's execution during the game, let $\numParties$ denote the total number of parties, $\numSessions$ the maximum number of sessions, and $\numStages$ the maximum number of stages. We note that each party generates a new ephemeral/ratchet key for each session, and possibly for each stage. This means a total maximum of $\numParties\numSessions\numStages$ $\diffihell$ keys in the list $\PKL$, every pair of which must not collide. There are $\binom{|\PKL|}{2}$ such pairs of $\diffihell$ keys to consider in the game. Each $\diffihell$ key in $\PKL$ is in the same group of order $q$, therefore any key collides with another key in $\PKL$ with probability $1 / q$. Thus, we have the following bound: $\Adv{\msind}{G0, \adversary} (\lambda) \leq \frac{\binom{(\numParties\numSessions\numStages)}{2}}{q} + \Adv{\msind}{G1, \adversary} (\lambda) $.
We now know that from this game onwards, each honestly generated $\diffihell$ public key is unique. In future game hops, we will replace certain $\diffihell$ values with ones sampled by a $\GDH$ challenger; this means that if these replacement values collide, we trigger the abort event $\gameabort{1}$.
\game{2}
In this game, the challenger guesses in advance the session $\testsess$ against which the $\Test(u', i', s')$ query is issued: the challenger guesses a pair of indices $(u,i,s) \in [\numParties \times \numSessions \times \numStages]$. Let $T$ be the event that the adversary issues a test query $\Test(u', i', s')$ where $(u',i',s') \neq (u,i,s)$. In this game, we abort if event $T$ occurs; it is a transition based on a large failure event.
$T$ will occur with probability $1 / \numSessions\numParties\numStages$ , and hence: $\Adv{\msind}{G1, \adversary} (\lambda) \leq \numParties\numSessions\numStages \cdot \Adv{\msind}{G2,\adversary} (\lambda) $.

From here on, we can assume the challenger's guessed session $(u,i,s)$ matches the adversary's $\Test(u',i',s')$; otherwise, the game aborts.
\game{3}

In this game, the challenger guesses in advance the partner session $\partsess$ against which the $\Test(u, i, s)$ query is issued: the challenger guesses a party $v \in [\numParties]$. Let $T$ be the event that the adversary issues a Test query $\Test(u
, i, s)$ where $\partsess.\peerid = v$. In this game, we abort if event $T$ doesn't occur; it is a transition based on a large failure event.
$T$ will occur with probability $1 / \numParties\numSessions$, and hence: $\Adv{\msind}{G2, \adversary} (\lambda) \leq \numParties\numSessions \cdot \Adv{\msind}{G3,\adversary} (\lambda) $.

At this point, we partition our analysis into four separate cases, each examining the adversary's advantage in winning under a certain freshness condition, with the ultimate aim of bounding $\Adv{}{G3}$ above. Since this is $\gameref{3}$, each different case begins with a hop to some $\gameref{4}$. 
Specifically, we split our analysis by examining the stage $s=[x,y]$ at which the adversary issues the $\Test(u,i,s)$ query. In $\LINEDR$, each stage index $s=[x,y]$ falls into exactly one of four cases:

\begin{enumerate}
    \item the initial stage $s=[0,0]$
    \item the asymmetric stages $s=[1,0]$ where $x\geq1$
    %\item all non-initial asymmetric stages $s=[x,0]$ where $x\geq2$
    \item the symmetric stages $s=[x,y]$ where $y\geq1$
\end{enumerate}

Given there are three cases capturing four distinct types of stages, then:

\begin{align*}\Adv{\msind}{G3, \adversary} (\lambda) \leq \Adv{\msind}{C1,\adversary} (\lambda) +& \Adv{\msind}{C2,\adversary} (\lambda)+\Adv{\msind}{C3,\adversary} (\lambda)
\end{align*}

\case{1} \textbf{Initial Stage ($\adversary$ issues $\Test(u,i,s)$ and $s=[0,0]$)}

In this case, we know that if $\Adv{\msind}{C1,\adversary} \geq 0$, then $\freshpredicate_{\initial}(\tested) = \true$. Otherwise, the experiment either aborts or its execution becomes independent of the challenge bit $b$. For $\freshpredicate_{\initial}(\tested) = \true$, then either $\freshpredicate_{\LL}$ or $\freshpredicate_{\EL}$ must be true. We note that while $\freshpredicate_{\LL}$ holds, the adversary cannot issue $\RevLongTermKey$ to either the test session or its communication partner. In $\freshpredicate_{\EL}$, the adversary cannot issue $\RevLongTermKey$ to the responder in the test session, and cannot issue $\RevRand$ to the initiator in the test session. 

We split in 2 subcases, each capturing the advantage of the adversary where one of these freshness predicates hold. Since one of these two conditions must hold, it follows that: $\Adv{\msind}{C1, \adversary} (\lambda)\leq \Adv{\msind}{C1.1, \adversary} (\lambda)+ \Adv{\msind}{C1.2, \adversary}(\lambda)$

\case{1.1} \textbf{Stage $s=[0,0]$ and $\freshpredicate_{\LL}(u,i,s)$ holds.} In this case we know that $\RevLongTermKey(u) \land \RevLongTermKey(v)$ not issued.

\game{4}

In this game, we replace the value $(g^y)^x$ used by the test session (and potentially, its communicating partner) as the KDF input for computing $(\rootkey_0,\chainkey_{[0,0]})$ with an independent random $g^z \gets \mathbb{G}$. We argue that if an adversary can distinguish this change, then we can build a reduction to break the $\GDH$ assumption. We introduce the following reduction $\bdversary$ that modifies the experiment in the following way:

At the beginning of the experiment, the reduction initialises a $\GDH$ challenger, which outputs $g^x$ and $g^y$, which the reduction uses to replace the long-term $\diffihell$ keys of $u$ and $v$, respectively. Now, whenever the adversary queries the random oracle with some input $in$, the reduction $\bdversary$ submits the Diffie-Hellman components ($g^x,g^y,in$) to the $\GDH$'s $\DDH$ oracle. If the $\DDH$ oracle returns $\trueflag$, then $in=g^{xy}$. In this case, the reduction submits $in$ to the $\GDH$ challenger as its $\GDH$ solution and wins, since $in = g^{xy}$. If the adversary never queries $g^{xy}$, then replacing the computation of $g^{xy}$ with $g^z$ (a uniformly random and independent value) is undetectable, since $g^{xy}$ is only used as input to the $\RO$, and $\RO$ outputs on unqueried inputs are uniformly random. The only way for the adversary distinguish this change then is to query $in$ to the $\RO$. However, doing so allows our reduction to solve the $\GDH$ challenge.

We note that $\bdversary$ no longer has access to the long-term secret keys of $u$ or $v$, respectively $x$ or $y$. Thus, whenever $\bdversary$ is required to compute using $x$ or $y$, it has to simulate these computations instead.  Sometimes, the reduction can simulate this computation using their access to all other secret keys in the experiment. Specifically, whenever $\bdversary$ is asked to compute $(g^t)^x$ for some other long-term public key $g^t$, it uses its knowledge of $t$ to compute $(g^x)^t$ instead. Since the adversary cannot inject long-term keys, so the reduction is sound.

However, long-term keys are sometimes combined with ephemeral keys $g^t$, which may be controlled by the adversary. Because $\bdversary$ does not know the exponents $x$ or $y$, it cannot directly compute the shared values of the form $g^{xt}$ or $g^{xt}$ that may arise when the adversary introduces its own ephemeral key $g^t$. To handle this, $\bdversary$ samples random group elements $g^z$ to represent each such computed $\diffihell$ value ($g^{xt}$, $g^{yt}$, etc.) and records the associated tuples $(g^a, g^t, g^z)$. Whenever the adversary queries the random oracle on some input $in$, $\bdversary$ checks whether this query corresponds to any stored tuple, i.e., uses the $\DDH$ oracle to check if $(g^{a}, g^t, in)$ is a valid $\DDH$ tuple. If so, the reduction reprograms the random oracle, such that whenever the adversary queries $in$ to the $\RO$, the reduction replaces it with a $\RO$ query $g^z$ instead. This ensures that if the adversary ever queries the computed $\diffihell$ value (such as $g^{xt}$ or $g^{at}$), the oracle's output remains consistent with the reduction's guessed values, preserving an indistinguishable simulation.

If the adversary can distinguish the change above, then the reduction solves $\GDH$. Thus: $\Adv{\msind}{G, \adversary} (\lambda)\leq \Adv{\msind}{G4,\adversary}(\lambda) + \Adv{\GDH}{\mathbb{G}, \adversary}(\lambda)$.

\game{5}

In this game, we replace the $\rootkey_0, \chainkey_{[0,0]} \getsr \bits^{\lambda}$ for the test session $\testsess$ (and its communication partner $\partsess$ in stage $s[0,0]$.  Because from $\gameref{4}$, we know that the adversary never queried the $\RO$ on the input $g^{xy}$. Thus, in the random oracle model, outputs on unqueried inputs are indistinguishable from uniform, and thus $\chainkey_{[0,0]}$ and $\rootkey_0$ are already uniformly random and independent values. We conclude that: $\Adv{\msind}{G4, \adversary} (\lambda)\leq \Adv{\msind}{G5, \adversary}(\lambda)$.

\game{6} In this game, we replace the message key $\msgkey_{[0,0]}$, and chain key $\chainkey_{[0,1]}$ computed by the test session (and its communicating partner) with uniformly random values. We note that these values are derived from a $\KDF$ computation using $\chainkey_{[0,0]}$ as input, i.e. $(\msgkey_{[0,0]}, \chainkey_{[0,1]}) = \KDF(\chainkey_{[0,0]})$. By \textbf{Game 5}, $\chainkey_{[0,0]}$ is already a uniformly random value independent of the experiment execution. Since we model $\KDF$ as a random oracle, it follows that message key $\msgkey_{[0,0]}$, and chain key $\chainkey_{[0,1]}$ is uniform too. We conclude that: $\Adv{\msind}{G5,\adversary} (\lambda) \leq \Adv{\msind}{G6,\adversary} (\lambda)$

Since the message key $\msgkey_{[0,0]}$ is uniformly random and independent of the experiment, then the output of the test query $\Test(u,i,s)$ is also uniformly random and independent of the experiment, regardless of the bit $b$. Hence, the adversary has no advantage in guessing the bit $b$ and thus we find $\Adv{\msind}{G6,\adversary} (\lambda)=0$.

We now turn to \caseref{1.2} where we know that $\RevLongTermKey(u) \land \RevRand(v)$ not issued.

\case{1.2} \textbf{Stage $s=[0,0]$ and $\freshpredicate_{\EL}(u,i,s)$ holds.}

Here $\freshpredicate_{\EL}$ holds, means that one party's initial ephemeral secret key is unrevealed and the peer's long-term secret is unrevealed. Therefore, the $_{\EL}$ Diffie-Hellman term (e.g $g^{ay}$) entering the initial $\KDF$ remains hidden to the adversary. Intuitively, this suffices for the initial keys to behave as random from the adversary's perspective, so security at the initial stage follows from the secrecy of this $_{\EL}$ term.

\game{4}
In this game, we resolve role symmetry by conditioning on the initiator role. Specifically, if the guessed test session's role is the initiator, we continue; otherwise, we abort the game.
The responder case is argued identically up to a change of notation. Equivalently, we may guess the role in advance and abort if the guess is wrong; the guess is correct with probability $\frac{1}{2}$. Thus: $\Adv{\msind}{G3, \adversary} (\lambda)  \leq 2\cdot \Adv{\msind}{G4,\adversary}(\lambda) $.

\game{5}
In this game, we replace the value $(g^y)^a$ used by the test session (and potentially, its communicating partner) as the KDF input for computing $(\rootkey_0,\chainkey_{[0,0]})$ with an independent random $g^z \getsr \mathbb{G}$. We argue that if an adversary can distinguish this change, then we can build a reduction to break the $\GDH$ assumption. We introduce the following reduction $\bdversary$ that modifies the experiment in the following way:

At the beginning of the experiment, the reduction initializes a $\GDH$ challenger and recieves $g^a$, $g^y$, and uses these to replace the test session's initial ephemeral secret key public key and the partner's long term public key, respectively. Because $\bdversary$ does not know the exponents $a$ or $y$, it cannot directly compute the shared value $g^{ay}$ or any derived values of the form $g^{at}$ or $g^{yt}$ that may arise when the adversary introduces its own ephemeral key $g^t$. To handle this, $\bdversary$ guesses random group elements to represent each such computed $\diffihell$ value ($g^{ay}$, $g^{at}$, $g^{yt}$, etc.) and records the associated tuples $(g^a, g^y)$ or $(g^a, g^t)$ and $(g^y, g^t)$ together with the guessed value. Whenever the adversary queries the random oracle on some input $in$, $\bdversary$ checks whether this query corresponds to any stored tuple, if so, reprograms the random oracle, so that the response is bound to the guessed value. This ensures that if the adversary ever queries the computed $\diffihell$ value (such as $g^{ay}$ or $g^{at}$), the oracle's output remains consistent with the reduction's guessed values, preserving an indistinguishable simulation. Whenever the adversary calls the Random Oracle with some input $in$, the reduction $\bdversary$ submits the Diffie-Hellman components ($g^a, g^y,in$) to the $\DDH$ oracle. If $\DDH$ returns $\trueflag$ then $in=g^{ay}$: the reduction outputs $in$ to the $\GDH$ challenger as its $\GDH$ solution and wins. If the adversary never queries $g^{ay}$, then replacing it with $g^z$ with uniformly random and independent value is undetectable since $\RO$ outputs on unqueried inputs are uniformly random and independent of the inputs. 

If the adversary does query $g^{ay}$, then the reduction solves $\GDH$. Thus $\gameref{3}$ and $\gameref{4}$ are indistinguishable unless $\GDH$ is solvable. Thus: $\Adv{\msind}{G4, \adversary} (\lambda)\leq \Adv{\msind}{G5,\adversary}(\lambda) + \Adv{\GDH}{\mathbb{G}, \adversary}(\lambda)$.

\game{6}

In this game, the experiment replaces the $\rootkey_0, \chainkey_{[0,0]} \getsr \bits^{\lambda}$ for $\Test(u,i,s)$ in stage $s=[0,0]$.  Because from $\gameref{4}$, we replace $g^{ay}$ with $g^z$ where $g^z$ is uniformly random and independent of the protocol flow. Thus, in the random oracle model, outputs on unqueried inputs are indistinguishable from uniform, and hence this hop is indistinguishable. We conclude that: $\Adv{\msind}{G5, \adversary} (\lambda)\leq \Adv{\msind}{G6, \adversary}(\lambda)$.

\game{7}

In this game, we replace $(\msgkey_{[0,0]}, \chainkey_{[0,1]}) = \KDF(\chainkey_{[0,0]})$ with uniformly random, independent values. Since in \gameref{6} the input $\chainkey_{[0,0]}$ is already uniform and independent, and the $\KDF$ is modeled as a random oracle, these outputs are indistinguishable from uniform. We conclude that: $\Adv{\msind}{G6,\adversary} (\lambda) \leq \Adv{\msind}{G7,\adversary} (\lambda)$.

Since the message key is uniformly random and independent of the protocol flow, and hence the adversary has no advantage in guessing and winning the game. So our \caseref{1} bound is: 

$$\Adv{\msind}{C1,\adversary} (\lambda) \leq 3 \cdot \Adv{\GDH}{\mathbb{G},\adversary} (\lambda)$$

\case{2} \textbf{Asymmetric Stage ($\adversary$ issues $\Test(u,i,s)$ and $s=[x,0]$ where $x\ge1$)}

In this case, we assume the test session at stage $s=[x,0]$ is fresh. We know that for $\Adv{\msind}{C2,\adversary} \geq 0$, then $\freshpredicate(\tested) = \true$, since otherwise the execution of the experiment is independent of the challenge bit $b$ sampled by the challenger, or the game aborts. For $\freshpredicate(\tested) = \true$, either $\freshpredicate_{\initial}$ (as shown in \caseref{1}, the session is fresh when the adversary has not revealed any long-term or ephemeral keys.) or $\freshpredicate_{\EE}$ holds. In $\freshpredicate_{\EE}$, the adversary cannot issue $\RevRand$ to either the test session or its communication partner. We split in 2 subcases, each capturing the advantage of the adversary where one of these freshness predicates hold. Since one of these two conditions must hold, it follows that:
$\Adv{\msind}{C2, \adversary} (\lambda)\leq \Adv{\msind}{C2.1, \adversary} (\lambda)+ \Adv{\msind}{C2.2, \adversary}(\lambda)$.

\case{2.1} \textbf{Stage $s=[1,0]$ and  $ \freshpredicate_{\initial}(u,i,s)$ holds.}
We note that for $\freshpredicate_{\initial}$ from \caseref{1}, either $\freshpredicate_{\LL}$ (both long-term keys unrevealed) or $\freshpredicate_{\EL}$ (one side's initial ephemeral secret key and the peer's long-term secret unrevealed) holds.
We can further split into two subcases. $\Adv{\msind}{C2, \adversary} (\lambda)\leq \Adv{\msind}{C2.1.1, \adversary} (\lambda)+ \Adv{\msind}{C2.1.2, \adversary}(\lambda)$.

\case{2.1.1} \textbf{Stage $s=[1,0]$ and $\freshpredicate_{\LL}(u,i,s)$ holds.} In this case we know that $\RevLongTermKey$ not issued for either the test session or the communicating partner.

\game{4} Following the same reasoning as in \caseref{1.1}, we replace the $\diffihell$ component $(g^{y})^{x}$ used by the test session to derive $(\rootkey_0, \chainkey_{[0,0]})$ with an independent random group element $g^{z} \getsr \mathbb{G}$. The reduction proceeds analogously: a simulator $\mathcal{B}$ embeds the $\GDH$ challenge $(g^{x}, g^{y})$ into the experiment and answers random-oracle queries by forwarding tuples $(g^{x}, g^{y}, in)$ to the $\DDH$ oracle. If the adversary distinguishes the change, $\mathcal{B}$ can extract $g^{xy}$ and solve the $\GDH$ instance; otherwise, the modification is indistinguishable, since random-oracle outputs on unqueried inputs remain uniformly random. Hence, the advantage gap between \gameref{3} and \gameref{4} is bounded by the adversary's advantage in solving $\GDH$. $\Adv{\msind}{G4, \adversary} (\lambda)\leq \Adv{\msind}{C2.1.1,\adversary}(\lambda) + \Adv{\GDH}{\mathbb{G}, \adversary}(\lambda)$.

\game{5} In this game, we replace $\rootkey_1, \chainkey_{[1,0]} \gets \KDF((g^a)^b, \rootkey_0)$ with uniformly random and independent values. We note $\rootkey_0$ from \gameref{4}) is already uniformly random and independent, and we model $\KDF$ as $\RO$, so the outputs $\rootkey_1, \chainkey_{[1,0]}$ are also uniformly independent. Thus $\Adv{\msind}{G4, \adversary} (\lambda)\leq \Adv{\msind}{G5,\adversary}(\lambda)$.

\game{6} In this game, we replace $\msgkey_{[1,0]}, \chainkey_{[1,1]} \gets \KDF(\chainkey_{[1,0]})$ with uniformly random and independent values. We note as $\chainkey_{[1,0]}$ (from \gameref{5}) is also uniformly random and independent and $\KDF$ is modeled as a $\RO$, we get $\msgkey_{[1,0]}, \chainkey_{[1,1]}$ as uniform as well. $\Adv{\msind}{G5, \adversary} (\lambda)\leq \Adv{\msind}{G6,\adversary}(\lambda)$.

\case{2.1.2} \textbf{Stage $s=[1,0]$ and $\freshpredicate_{\EL}(u,i,s[1,0])$ holds.} In this case we know that $\RevRand$ not issued for one side's initial ephemeral secret key and $\RevLongTermKey$ not issued for the peer's long-term secret.

\game{4} The experiment proceeds identically to \gameref{4} of \caseref{1.2}, In this game, we resolve role symmetry by conditioning on the initiator role. Specifically, if the guessed test session's role is the initiator, we continue; otherwise, we abort the game.
The responder case is argued identically up to a change of notation. Equivalently, we may guess the role in advance and abort if the guess is wrong; the guess is correct with probability $\frac{1}{2}$. Thus: $\Adv{\msind}{G3, \adversary} (\lambda)  \leq 2\cdot \Adv{\msind}{G4,\adversary}(\lambda) $.

\game{5}The experiment proceeds identically to \gameref{5} of \caseref{1.2}, we replace the value $(g^y)^a$ used for deriving $(\rootkey_0,\chainkey_{[0,0]})$ with an independent random $g^z \getsr \mathbb{G}$. The same $\GDH$ $\RO$ reduction applies: if adversary queries $g^{ay}$, then the reduction solves $\GDH$. Otherwise the change is indistinguishable. Thus: $\Adv{\msind}{G4, \adversary} (\lambda)\leq \Adv{\msind}{G5,\adversary}(\lambda) + \Adv{\GDH}{\mathbb{G}, \adversary}(\lambda)$.

\game{6} The experiment proceeds identically as \gameref{5} in \caseref{2.1.1}, replace $\rootkey_1, \chainkey_{[1,0]} \gets \KDF((g^a)^b, \rootkey_0)$ with uniformly random and independent values. We note $\rootkey_0$ in already uniformly random and independent, and we model $\KDF$ as $\RO$, the outputs $\rootkey_1, \chainkey_{[1,0]}$ are uniformly and independent. Thus $\Adv{\msind}{G6, \adversary} (\lambda)\leq \Adv{\msind}{G7,\adversary}(\lambda)$.

\game{7} The experiment proceeds identically as \gameref{6} in \caseref{2.1.1}, replace $\msgkey_{[1,0]}, \chainkey_{[1,1]} \gets \KDF(\chainkey_{[1,0]})$ with uniformly random and independent values. We note as $\chainkey_{[1,0]}$ (from \gameref{6}) is also uniformly random and independent and $\KDF$ is modeled as a $\RO$, we get $\msgkey_{[1,0]}, \chainkey_{[1,1]}$ as uniform as well. $\Adv{\msind}{G7, \adversary} (\lambda)\leq \Adv{\msind}{G8,\adversary}(\lambda)$.

\case{2.2} \textbf{Stage $s=[x,0]$ at $x\ge2$ and  $\freshpredicate_{\EE}(u,i,s[x,0])$ holds.} In this case we know that the adversary cannot issue $\RevRand$ for either the sender's new ephemeral at $s=[x,0]$ or the peer's previous ephemeral at $s=[x-1,0]$ are unrevealed.

\game{4} The experiment proceeds identically to \gameref{4} of \caseref{1.2}, In this game, we resolve role symmetry by conditioning on the initiator role. Specifically, if the guessed test session's role is the initiator, we continue; otherwise, we abort the game.
The responder case is argued identically up to a change of notation. Equivalently, we may guess the role in advance and abort if the guess is wrong; the guess is correct with probability $\frac{1}{2}$. Thus: $\Adv{\msind}{G3, \adversary} (\lambda)  \leq 2\cdot \Adv{\msind}{G4,\adversary}(\lambda) $.

\game{5}
In this game we replace the $\EE$ $\diffihell$ value at stage $x$, the computed $\diffihell$ value of the sender's new ephemeral at stage $s=[x,0]$ and the partner's ephemeral at $s=[x-1,0]$ ($(g^a)^b$) that enters the $\KDF$ deriving ($\rootkey_x, \chainkey_{[x,0]}$) with an independent and random value $g^z \getsr \mathbb{G}$. We embed a $\GDH$ challenge ($g^a, g^b$) as the two ephemeral public keys (current sender at $s=[x,0]$ uses $g^a$ and partner's previous at $s=[x-1,0]$ uses $g^b$) and we never reuse these challenge elements in any other session nor with adversary values.

We show the change is indistinguishable via a the following reduction. When $\Send(u,i,\Activate)$ is called, the reduction sets the test session's outgoing ephemeral $epk_{\testsess} \gets g^a$. The reduction thus modifies the computation of $\rootkey_{x}$ by instead setting $\rootkey_x, \chainkey_{[x,0]} \gets \KDF(g^z, \rootkey_{x-1})$. When $\Send(v,j,\msg)$ is queried for the matching partner $j$ in stage $x$, the reduction sets the matching partner's new ephemeral public key as $epk_{\partsess} \gets g^b$. Thus, reduction computes $\rootkey_x, \chainkey_{[x,0]}$ by calling the $\KDF$ on the input where $(g^a)^b$ is replaced by $g^z$ together with $\rootkey_{[x-1]}$. On any $\RO$ query whose input could contain the $\EE$ value $g^{ab}$, our reduction checks whether the adversary has solved the $\GDH$ challenge by submitting ($g^a, g^b, in$) to the $\DDH$ oracle; if the $\DDH$ oracle returns $\trueflag$, then $in=g^{ab}$ and the reduction submits $in$ to the $\GDH$ challenger and wins; otherwise the reduction calls $in$ to the $\RO$ and outputs the uniformly random value to the adversary. This argument follows the same reasoning as in \gameref{4} in \caseref{1.1}.

If the adversary ever queries $g^{ab}$, we solve $\GDH$ otherwise the replacement of $(g^a)^b \Rightarrow g^z$ is undetectable in $\RO$ model. Thus $\Adv{\msind}{G4, \adversary} (\lambda)\leq \Adv{\msind}{G5,\adversary}(\lambda) + \Adv{\GDH}{\mathbb{G}, \adversary}(\lambda)$.

\game{6}
In this game, we replace $(\rootkey_{x},\chainkey_{[x,0]})$ for the test session (and its matching partner) with uniformly random and independent values sampled from their key spaces. From \gameref{5}, the $\KDF$ input for deriving $(\rootkey_{x},\chainkey_{[x,0]})$ already contains the independent uniform group element $g^{z}$; in the random-oracle model, outputs on inputs containing an unknown uniform component are indistinguishable from uniform unless the adversary queries that exact input. Hence $\Adv{\msind}{G5, \adversary} (\lambda)\leq \Adv{\msind}{G6,\adversary}(\lambda)$.

\game{7}
In this game, we replace $(\msgkey_{[x,0]},\chainkey_{[x,1]})=\KDF(\chainkey_{[x,0]})$ with uniformly random and independent values. By \textbf{Game 6}, $\chainkey_{[x,0]}$ is already uniform and independent of the execution, and since $\KDF$ is modeled as a random oracle, the derived pair is uniform as well. Therefore $\Adv{\msind}{G6, \adversary} (\lambda)\leq \Adv{\msind}{G7,\adversary}(\lambda)$.

\case{2.3} \textbf{Stage $s=[x,0]$ where $x\ge1$ and $\freshpredicate_{\state}(u,i,[x-1,0]) \land \freshpredicate_{\asym}(u,i,[x-1,0])$ holds.}

We highlight that our argument proceeds by induction. $\freshpredicate_{\state}(u,i,[x-1,0])$ means the adversary did not issue $\RevState(u,i,[x-1,0])$. By definition of $\freshpredicate_{\asym}$, either $\freshpredicate_{\EE}(u,i,[x-1,0])$ holds, or $\freshpredicate_{\asym}(u,i,[x-2,0])$ holds. Iterating this backwards, there exists some $x' \in \{1,\ldots,x-1\}$ such that $\freshpredicate_{\EE}(u,i,[x',0])$ holds. If $x'=1$, we invoke \caseref{2.1} (\gameref{4-6}) to replace $(\rootkey_{1},\chainkey_{[1,0]})$ by uniform values; if $x'\ge 2$, we invoke \caseref{2.2} at index $x'$ (\gameref{4-6}) to replace $(\rootkey_{x'},\chainkey_{[x',0]})$ by uniform values. From this uniform base, we propagate forward using only the random-oracle model and the fact that $\RevState(u,i,[x-1,0])$ did not occur.

\game{4} 
In this game, by the backward-induction argument we import the uniformity of the previous asymmetric stage: by applying \caseref{2.2} (\gameref{4-6}) at index $x-1$. We may treat $\rootkey_{x-1}, \chainkey{[x-1,0]}$ for the test session (and its matching partner session) as uniformly random and independent. $\Adv{\msind}{G3, \adversary} (\lambda)\leq \Adv{\msind}{G4,\adversary}(\lambda)$.

\game{5}
In this game, we replace $(\rootkey_{x},\chainkey_{[x,0]}) = \KDF((g^a)^b,\rootkey_{x-1})$ with uniformly random and independent values. Since $\rootkey_{x-1}$ is already uniform and independent values from \gameref{4}, the random-oracle output at this fresh input is indistinguishable from uniform unless that exact input is queried, which is impossible without knowing $\rootkey_{[x-1]}$. Therefore $\Adv{\msind}{G4, \adversary} (\lambda)\leq \Adv{\msind}{G5,\adversary}(\lambda)$.

\game{6}
In this game, we replace $(\msgkey_{[x,0]},\chainkey_{[x,1]}) = \KDF(\chainkey_{[x,0]})$ with uniformly random and independent values. Since $\chainkey_{[x,0]}$ is uniform from \gameref{5} and $\KDF$ is modeled as a random oracle, the derived pair is uniform as well. Hence $\Adv{\msind}{G5, \adversary} (\lambda)\leq \Adv{\msind}{G6,\adversary}(\lambda)$.  So our \caseref{2} bound is: 

$$\Adv{\msind}{C2,\adversary} (\lambda) \leq 6 \cdot \Adv{\GDH}{\mathbb{G},\adversary} (\lambda)$$

\case{3} \textbf{Symmetric Stages ($\adversary$ issues $\Test(u,i,s)\land s=[x,y]$ where $y\ge1$ )}

In this case we assume the test session at stage $s=[x,y]$ with $y\ge 1$ is fresh; otherwise the execution is independent of the challenge bit $b$ or the game aborts. The freshness conditions for the symmetric stages are inductive: a stage being fresh at stage $s=[x,y]$ requires that stage $s=[x,y-1]$ is also fresh. Eventually, this requires that stage $s=[x,0]$ is fresh. Thus, freshness at $s=[x,y]$ can be achieved in exactly two ways: (i) the initial stage's freshness predicate $\freshpredicate_{\initial}(u,i,[0,0])$ holds, and $\freshpredicate_{\state}(u,i,[x',y'])$ holds for all $0\leq x' \leq x$, $0\leq y' \leq y$ or; (ii) some asymmetric stage's freshness predicate $\freshpredicate_{\asym}(u,i,[x',0])$ holds (such that $x' \leq x$ and $\freshpredicate_{\state}(u,i,[x',y'])$ holds for all $0\leq x' \leq x$. In either case, our proof proceeds as follows: in some stage $x'$ we know a "base" Diffie-Hellman exchange existed such that neither Diffie-Hellman secret was revealed by the adversary; here we embed a $\GDH$ challenge. 

If condition (i) holds, then by the same arguments as \caseref{1} we can replace the initial root key $rk_0$ with a uniformly random value independent of the protocol flow. If condition (ii) holds, then by the same arguments as \caseref{2} we can replace the stage $x'$ root key $rk_{x'}$ with a uniformly random value independent of the protocol flow. From there, we can iteratively replace all root keys and chain keys $(rk_{x'+1},ck_{[x'+1,0]}), \ldots, (rk_{x},ck_{[x,0]})$ with uniformly random values. This is because computing $(rk_{x'+1},ck_{[x'+1,0]}), \ldots$ is the output of the RO using $rk_{x'},\ldots$ as input, which we know is a uniformly random value that the adversary has not revealed, and therefore cannot distinguish from random. From there, we can iteratively replace all message keys and chain keys $(mk_{[x,1]},ck_{[x,1]}), \ldots, (mk_{[x,y]},ck_{[x,y]})$ with uniformly random values. This is because computing $(mk_{[x,1]},ck_{[x,1]}), \ldots$ is the output of the RO using $ck_{[x,0]},\ldots$ as input, which we know is a uniformly random value that the adversary has not revealed, and therefore cannot distinguish from random.

At this point, we have replaced the tested message key $\msgkey_{[x,1]}$ with a uniformly random value regardless of the value of the challenge bit, and thus the challenge bit $b$ is uniformly random and independent of the security experiment. Thus the advantage of the adversary is $0$.

We now split into two subcases: 

$\Adv{\msind}{C4, \adversary} (\lambda)\leq \Adv{\msind}{C4.1, \adversary} (\lambda)+ \Adv{\msind}{C4.2, \adversary}(\lambda)$.

\case{3.1} \textbf{Stage $s=[0,y]$ and $\freshpredicate_{\state}(u,i,[0,y-1]) \land \freshpredicate_{\initial}(u,i,[0,0])$ holds.}

In this case we assume the test session at $s=[0,y]$ is fresh. By $\freshpredicate_{\sym}(u,i,[0,y])$ we in particular have that $\freshpredicate_{\initial}(u,i,[0,0])$ holds, and that $\freshpredicate_{\state}(u,i,[x',y'])$ holds for all $0\leq x' \leq x$, $0\leq y' \leq y$, so the initial stage state is not revealed. 

\game{4} 

In this game we replace the initial rook key $\rootkey_{0}$ and chain key $\chainkey_{[0,0]}$ with uniformly random values for the test session (and its matching partner). This follows the arguments made in \caseref{1} with the same bounds. Hence 
$\Adv{\msind}{G3, \adversary} (\lambda)\leq \Adv{\msind}{G4,\adversary}(\lambda)+  3\cdot \Adv{\GDH}{\Group,\adversary}(\lambda)$.

\game{5}

In this game, we iteratively replace all message keys and chain keys \sloppy$(mk_{[0,1]},ck_{[0,1]}), \ldots, (mk_{[0,y]},ck_{[0,y]})$ with uniformly random values. We justify this change by appealing to the RO: computing $(mk_{[0,1]},ck_{[0,1]}), \ldots$ is the output of the RO using $ck_{[0,0]},\ldots$ as input, which we know is a uniformly random value that the adversary has not revealed, and therefore cannot distinguish from random. At this point, we have replaced the tested message key $\msgkey_{[0,y]}$ with a uniformly random value regardless of the value of the challenge bit, and thus the challenge bit $b$ is uniformly random and independent of the security experiment. The advantage of the adversary is $0$, and the bound in this subcase is:
$\Adv{\msind}{G4, \adversary} (\lambda) \leq 3\cdot \Adv{\GDH}{\Group,\adversary}(\lambda)$.

\case{3.2} \textbf{Stage $s=[x,y]$ where $x\ge1$ and $\freshpredicate_{\state}(u,i,[x,y-1]) \land \freshpredicate_{\asym}(u,i,[x',0])$ holds.}

In this case we assume the test session at $s=[x,y]$ is fresh. By $\freshpredicate_{\sym}(u,i,[x,y])$ we in particular have that $\freshpredicate_{\asym}(u,i,[x',0])$ holds, and that $\freshpredicate_{\state}(u,i,[x^*,y^*])$ holds for all $\leq x' \leq x^* \leq x$, $0\leq y^* \leq y$, so all states from each stage since $s=[x',0]$ has not revealed. 

\game{4} 

In this game we replace the rook key $\rootkey_{x'}$ and chain key $\chainkey_{[x',0]}$ with uniformly random values for the test session (and its matching partner). This follows the arguments made in \caseref{2} with the same bounds. Hence 
$\Adv{\msind}{G3, \adversary} (\lambda)\leq \Adv{\msind}{G4,\adversary}(\lambda)+  6\cdot \Adv{\GDH}{\Group,\adversary}(\lambda)$.

\game{5}

In this game, we iteratively replace all root keys and chain keys $(rk_{x'+1},ck_{[x'+1,0]}), \ldots, (rk_{x},ck_{[x,0]})$ with uniformly random values. We justify this change by appealing to the RO: computing $(rk_{x'+1},ck_{[x'+1,0]}), \ldots$ is the output of the RO using $rk_{x'},\ldots$ as input, which we know is a uniformly random value that the adversary has not revealed, and therefore cannot distinguish from random. Thus
$\Adv{\msind}{G4, \adversary} (\lambda)\leq \Adv{\msind}{G5,\adversary}(\lambda)$.

\game{6}

In this game, we iteratively replace all message keys and chain keys $(mk_{[x,1]},ck_{[x,1]}), \ldots, (mk_{[x,y]},ck_{[x,y]})$ with uniformly random values. We justify this change by appealing to the RO: computing $(mk_{[x,1]},ck_{[x,1]}), \ldots$ is the output of the RO using $ck_{[x,0]},\ldots$ as input, which we know is a uniformly random value that the adversary has not revealed, and therefore cannot distinguish from random. At this point, we have replaced the tested message key $\msgkey_{[x,y]}$ with a uniformly random value regardless of the value of the challenge bit, and thus the challenge bit $b$ is uniformly random and independent of the security experiment. The advantage of the adversary is $0$, and the bound in this subcase is:
$\Adv{\msind}{G4, \adversary} (\lambda) \leq 5\cdot \Adv{\GDH}{\Group,\adversary}(\lambda)$.

 So our \caseref{3} bound is: 

$$\Adv{\msind}{C3,\adversary} (\lambda) \leq 9 \cdot \Adv{\GDH}{\mathbb{G},\adversary} (\lambda)$$

\end{proof}

\mehr{From the above analysis, we conclude that $\LINEDR$ achieves strong cryptographic security guarantees within the multi-stage key exchange model. In particular, $\LINEDR$ provides key indistinguishability for all fresh stages, FS through the symmetric ratchet, and PCS via repeated asymmetric Diffie–Hellman ratcheting that injects fresh entropy into the root key. The protocol further achieves KCI resilience, since successful key derivation at asymmetric stages requires unrevealed ephemeral secrets, preventing an adversary who compromises long-term keys from impersonating an honest party. Finally, replay protection is ensured by the strict progression of asymmetric and symmetric stage indices, ensuring that previously accepted keys cannot be reused without detection. Together, these properties demonstrate that $\LINEDR$ significantly strengthens the security guarantees of $\LINE$’s messaging protocol.}

{\section {Benchmarking Double Ratchet $\LINE$}

As part of our implementation efforts, we developed a reference implementation of our improved $\LINE$ double ratchet ($\LINEDR$) protocol in Rust, leveraging existing cryptographic libraries from the Rust ecosystem. Specifically, our prototype is built using the libraries $\mathtt{x25519\mhyphen dalek}$ \cite{x25519-dalek}, $\mathtt{aes\mhyphen gcm}$ \cite{aes-gcm-crate}, $\mathtt{hkdf}$ \cite{hkdf-crate} and $\mathtt{sha2}$ \cite{sha2-crate}. We benchmarked both reference implementations using the $\mathtt{criterion}$ crate \cite{criterion} to evaluate their wall-time performance as well as cryptographic costs. Table~\ref{tab:wall} compares the end-to-end cost incurred during a single message transaction, where Alice sends an encrypted message that Bob successfully decrypts. In contrast, Table~\ref{tab:crypto} compares the cost of individual cryptographic operations across all protocol variants and stages. All results were evaluated over 100 iterations of the respective protocols.
We note that this reference implementation is not a production-ready secure messaging system, but rather a prototype designed to capture the core functionality of the $\LINEDR$ protocol for demonstration and analysis purposes. Certain engineering considerations, such as message loss handling, are omitted in order to emphasize the interplay between the cryptographic components.

The $\LINEDR$ protocol relies on three primary classes of cryptographic primitives: Diffie--Hellman key exchange, key derivation, and symmetric encryption with AES-GCM. In our implementation, the asymmetric ratchet is instantiated using X25519 Diffie-Hellman \cite{bernstein2006curve25519} via the $\mathtt{x25519\mhyphen dalek}$ library, providing forward secrecy and post-compromise security through periodic key agreement and evolution steps. 
Key material for both the root key and the sending/receiving chain keys is derived using the HMAC-based Key Derivation Function (HKDF) \cite{HKDF} with SHA-256, and implemented with the $\mathtt{hkdf}$ and $\mathtt{sha2}$ libraries in Rust. Each invocation of HKDF produces fresh chain keys and per-message encryption keys, ensuring that message compromise does not leak future session secrets. 
After successful key derivation, all messages exchanged between sender and receiver are encrypted and decrypted using the corresponding sending and receiving keys in AES-256 in Galois/Counter Mode (AES-GCM) \cite{nist80038d}, instantiated via the  $\mathtt{aes\mhyphen gcm}$ library from the RustCrypto project.

In $\LINEtwo$, the session-establishing $\mathsf{Init}$ stage involves both Alice and Bob agreeing on a session key, derived from a combination of key-exchange and HKDF operations. In contrast, during the $\mathsf{Init}$ stage of our $\LINEDR$ implementation, Alice generates an ephemeral key and combines it with Bob's long-term key to derive the initial root, chain, and message keys. Upon receiving Alice's message containing her new ephemeral key, Bob derives the corresponding keys and generates his own ephemeral key pair in preparation for the next exchange. This explains the  difference in wall-time between  $1^\text{st}\text{Message}/\mathsf{Init}$ in $\LINEtwo$ and $\LINEDR$ summarized in Table~\ref{tab:wall}.

After the $\mathsf{Init}$ stage, Alice constructs a message to transmit to Bob, consisting of header information, an AES-GCM encrypted ciphertext, and associated additional data (AAD). In $\LINEDR$, to support the ongoing ratcheting mechanism between Alice and Bob, the message header contains the sender's corresponding 32-byte public key and an additional $\mathtt{sending\_counter}$ variable to synchronize the double ratchet. At the receiving end, $\LINEtwo$ allows Bob to recover the message using the session key derived in the session-establishing $\mathsf{Init}$ stage. This encrypted exchange between Alice and Bob in $\LINEtwo$ forms a \emph{secure session}, corresponding to each $i^\text{th}$ message after the first. By contrast, in the \emph{asymmetric} $\LINEDR$, Bob first checks whether Alice has sent a new ephemeral public key. If so, this triggers Bob to update the root and receiving chain keys by combining Alice's new ephemeral key with their current ephemeral key. This enables Bob to derive the correct session key to recover the message from Alice. Next, Bob generates a new ephemeral Diffie-Hellman key pair and updates the root and sending chain keys by combining the fresh ephemeral key with Alice's recently received ephemeral key, thereby preparing for the next transaction. These additional computations account for the increased wall-time for $\LINEDR$ during the $\mathrm{Asymmetric \;Ratchet \; Decryption}$ stage, as shown in Table~\ref{tab:wall}. During the \emph{symmetric} $\LINEDR$, Alice and Bob update only their chain keys using a key-derivation function, without performing any new key agreements. 
Table \ref{tab:crypto} illustrates the average cost of cryptographic operation for each protocol stage.

\begin{table}[htbp!]
\centering
\resizebox{\textwidth}{!}{
\begin{tabular}{|p{3.5cm}|c|c|c|}
\hline
\diagbox[width=12.5em]{Protocol}{ Walltime} & E2E Avg. Cost & \makecell{Encryption Avg. \\ Cost} &  \makecell{Decryption Avg. \\ Cost}\\
\hline
$\LINEtwo$ $1^{\mathrm{st}}$ Message & $130.65 \pm 16.10$ & $61.98 \pm 5.6$ & $62.08 \pm 6.19$ \\
\hline
$\LINEtwo$ $i^{\mathrm{th}}$ Message& $2.58 \pm 0.22$ & $1.33 \pm 0.5$  & $1.43 \pm 0.97$\\
\hline
$\LINEDR \; \mathsf{Init}$ & $370.65\pm 30.76$ & $185.3 \pm 20.2$  & $186.26 \pm 24.6$\\
\hline
 $\LINEDR$ Asym Ratchet & $185.74 \pm 25.4$ & $1.58 \pm 0.46$  & $183.46 \pm 22.6$\\
\hline
 $\LINEDR$ Sym Ratchet & $3.15\pm 0.84 $ & $1.5 \pm 0.43$ & $1.63 \pm 0.48$\\
% \hline
% Difference & $133.349 \:\mu\mathrm{s}$ & $28.2674 \:\mu\mathrm{s}$ \\
\hline
\end{tabular}
}
\vspace{0.1cm}
\caption{Comparison of end-to-end walltime (in $\mu s$) and standard deviations on the \textbf{Intel(R) Core(TM) i7-11370H CPU @ 3.30GHz}.}
\label{tab:wall}
\end{table}

%\vspace{-0.9cm}

\setlength{\tabcolsep}{8pt}
\renewcommand{\arraystretch}{1.2}
\begin{table}[htbp!]
\Huge
\centering
\resizebox{\textwidth}{!}{
\begin{tabular}{|c|l|l|l|l|l|}
\hline
\diagbox[width=7.5em]{Operation}{\hspace{-3.5em}Protocol} & \makecell{$\LINEtwo$ $1^{\mathrm{st}}$\\ Message} &  \makecell{$\LINEtwo$ $i^{\mathrm{th}}$ \\ Message} & \makecell{$\LINEDR$ \\$\textsf{Init}$} & \makecell{$\LINEDR$ \\ Asym }&  \makecell{$\LINEDR$ \\ Sym } \\
\hline
$\mathsf{DH}$ & $2 \times 60.65 $ & $-$ & $3 \times 60.65 $ & $3 \times 60.65 $&  $-$\\
\hline
$\mathsf{KDF}$ & $2 \times	1.01 $ & $2 \times 	1.01 $ &$4 \times  1.34 $& $4 \times  1.34 $  & $2 \times   1.34 $\\
\hline
$\mathsf{AES\mhyphen GCM}$ & $2 \times 1.01$ & $ 2 \times 1.01$ & $-$ & $ 2\times 1.01$ & $ 2 \times  1.01$ \\
\hline
\end{tabular}
}
\vspace{0.1cm}
\caption{Comparison of Average Cryptographic Cost (in $\:\mu\mathrm{s}$) on \textbf{Intel(R) Core(TM) i7-11370H CPU @ 3.30GHz}.}
\label{tab:crypto}
\end{table}
%\vspace{-1cm}

While the ongoing double ratchet in $\LINEDR$ requires additional key generation, evolution, and derivation operations compared to the one-time key-establishment process of $\LINEtwo$, we argue that the enhanced security guarantees of forward secrecy and post-compromise security outweigh the marginal increase in computational cost. We note that $\LINEDR$ introduces a slight overhead due to the inclusion of a 32-byte public key and an additional $\mathtt{sending\_counter}$ in the message header. However, when run on modern mobile phones, $\LINEtwo$ and $\LINEDR$ would communicate over high-speed 5G bandwidths based on the 3GPP specification \cite{etsi_133501_v16}, with an uplink data rate of 25~Mbps and a downlink rate of 50~Mbps in a wide-area scenario. Therefore, we argue that this additional overhead can be easily accommodated by the bandwidth provided by the current 5G standard, without affecting the performance of the generic $\LINE$ protocol. Our goal is not absolute performance but relative overhead comparison between LINEv2 and LINEvDR under identical cryptographic primitives.

\section{Conclusion} \label{sec:con}

In this work, we present the first rigorous cryptographic analysis of the $\LINEtwo$ messaging protocol, which was designed as the successor to $\LINEone$ with the explicit goal of addressing the  security flaws identified in its predecessor. Our cryptographic security  investigation revealed that $\LINEtwo$ still fails to meet several essential security guarantees expected from a modern secure messaging protocol. Through our analysis of $\LINEtwo$, within a tailored Multi-Stage Key Exchange model, we demonstrated that although $\LINEtwo$ attains key indistinguishability and message authentication, its dependence on a static-static Diffie-Hellman key exchange precludes forward secrecy and post-compromise security. Additionally, our network-level analysis uncovered that $\LINEtwo$ does not apply end-to-end encryption to bot communications. With LINE bots increasingly used in e-commerce, banking, and healthcare, this design flaw poses a significant privacy risk for millions of users.
% Notably, our analysis shows that, despite being presented as a security-enhanced iteration of $\LINEone$, $\LINEtwo$ remains susceptible to long-term key compromise and replay attacks, thereby undermining its claimed cryptographic guarantees.

To address these weaknesses, we proposed $\LINEDR$, a strengthened variant of $\LINEtwo$ that introduces ratcheting-based key updates to achieve both FS and PCS while maintaining backward compatibility with the existing $\LINEtwo$ architecture. We formally proved that $\LINEDR$ satisfies all desired security properties within our adapted MSKE framework. Our empirical evaluations further demonstrated that $\LINEDR$ provides substantially improved security guarantees with reasonable performance overhead compared to the baseline $\LINEtwo$ protocol. We emphasize that while our empirical evaluations validate the feasibility of the proposed security enhancements to $\LINEtwo$, it was our formal analysis---often regarded as a purely theoretical exercise---that revealed previously overlooked  security flaws in the protocol design. This underscores the importance of rigorous cryptographic analysis of real-world systems, not as an abstract academic pursuit, but as a necessary means to substantiate the security guarantees such protocols claim to provide. As future work, we aim to extend our formal framework to analyse privacy guarantees and explore post-quantum extensions within the structural limitations of the existing LINE ecosystem.
\bibliography{mybibliography}

\appendix
\section{Definition and Execution Flow of the MSKE Model} \label{appen:MSKE}

The MSKE model formalizes the process of establishing multiple cryptographic keys between parties over the course of a protocol execution. In this model, each stage of the exchange yields a fresh session key, typically intended for use in a symmetric cryptographic primitive, such as encryption or message authentication providing security guarantees for the secure composition of these keys within symmetric-key application protocols. \cite{MSKE-1} introduce the MSKE model and demonstrate that the QUIC protocol satisfies its security requirements, qualifying as a secure MSKE protocol. A distinctive feature of MSKE is its support for key-dependent stages, allowing later stages to rely on keys established in earlier ones. The core security goals of the MSKE model include key indistinguishability, key freshness, and FS. To ensure the integrity of session keys, the model mandates that each key be unique to its session and not reused across different sessions. We note that we do not prove FS, and thus must modify the MSKE model to omit this guarantee.

\begin{figure}[htpb!]
\begin{center}
\scalebox{0.7}{\centering
\begin{tikzpicture}[yscale=-0.5,>=latex]
\edef\InitX{0}
\edef\ArrowLeft{0}
\edef\ArrowCenter{5}
\edef\ArrowRight{10}
\edef\AdvX{5}
\edef\RespX{10}
% Set the starting Y coordinate
\edef\Y{0}

% Draw header boxes
\node [rectangle,draw,inner sep=6pt,right] at (\InitX,\Y) {\textbf{Alice} };

\node [rectangle,draw,inner sep=6pt,left] at (\RespX,\Y) {\textbf{Bob}};

\NextLine[2]
\ClientAction{$\KGEN() \tor (\pk_A,\sk_A)$}
\ServerAction{$\KGEN() \tor (\pk_B,\sk_B)$}

\NextLine[2]
\ClientAction{$\Activate(\sk_A,\init,\id_B) \to (\pi_A, m_A)$}
\NextLine[1]
\ServerAction{$\Activate(\sk_B,\resp,\id_A) \to (\pi_B, \bot)$}

\NextLine[1.5]
\ClientToServer{$m_A$}{}

\NextLine[1]
\ServerAction{$\Run(\sk_B,\pi_B,m_A) \to (\pi^{\prime}_B, m^{\prime}_B)$}

\NextLine[1.5]
\ServerToClient{$m^\prime_B$}{}

\NextLine[1]
\ClientAction{$\Run(\sk_A,\pi'_A,m'_B) \to (\pi^{\prime\prime}_A, m^{\prime\prime}_A)$}

\end{tikzpicture}
 }
\caption{An example execution of MSKE protocol}
\label{fig:MSKE-flow}
\end{center}
\end{figure}
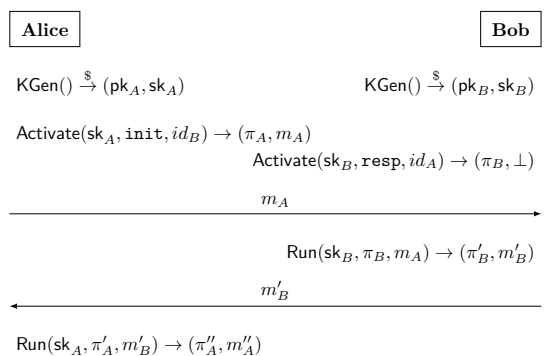

The MSKE framework can be divided into three distinct components: the MSKE protocol, the MSKE formalisation, and the MSKE security game. The MSKE protocol is a key exchange mechanism executed between two parties, establishing multiple keys throughout its execution. The MSKE formalisation introduces this protocol's notation and underlying algorithms. Lastly, the MSKE security game captures the security aspects of the MSKE protocol, ensuring that the keys generated are secure and robust against potential threats. We formalise the MSKE protocol in Definition \ref{def:MSKE}, adapting the MSKE definition presented in \cite{signal-1}.  To illustrate the MSKE protocol definition with an Alice and Bob communication scenario, we can break down the process step-by-step, showcasing how Alice and Bob interact with each other and utilize the algorithms defined in the MSKE protocol as shown in  Figure \ref{fig:MSKE-flow}.

\begin{definition} [$\MSKE$] \label{def:MSKE}
    A multi-stage key exchange protocol is a tuple of algorithms $\Pi$, along with a keyspace ($\KeySpace$) and a security parameter ($\lambda$) indicating a measure that scales the size of the system parameters to ensure adequate security against adversarial attacks. The algorithms $\Pi = \{\KGen, \Activate, \Run\}$ are: $\KGen(\secpar) \tor (\pk, \sk)$:  A probabilistic long-term key generation algorithm that takes a security parameter $\secpar$ input and outputs a long-term public key $\pk$ and a long-term secret key $\sk$. $\Activate(\sk,\role,\peerid) \tor (\pi, m)$: A probabilistic protocol activation algorithm that takes as input the party's long-term secret key $\sk$, the role $\role$ of the party in the session ($\role \in \{\init, \resp \}$), the identifier of the intended peer $\peerid$ and outputs a per-session state $\pi$ and a (possibly empty) outgoing message $m$. $\Run(\sk,\pi,m) \tor (\pi^\prime, m^\prime)$: A probabilistic protocol execution algorithm that takes an input as a long-term secret key $\sk$, per-session state $\pi$ and a (possibly empty) incoming protocol message $m$ and outputs an updated state $\pi'$ and a possibly empty outgoing protocol message $m'$.
\end{definition}

\section{$\LINEtwo$ Letter Sealing Protocol version 2 (Detailed)} 
\label{app:line2-detail}
Figure \ref{fig:line-vers-2(a)} presents the complete message flow of $\LINEtwo$, including the registration, session establishment, encryption and decryption phases. The diagram is provided here for reference and completeness, illustrating the unmodified baseline design of $\LINEtwo$.

\begin{figure}[htbp!]
\begin{center}
\scalebox{0.85}{
	\centering
\begin{tikzpicture}[yscale=-0.5,>=latex]
\edef\InitX{0}
\edef\ArrowLeft{0}
\edef\ArrowCenter{7}
\edef\ArrowRight{14}
\edef\AdvX{7}
\edef\RespX{14}
% Set the starting Y coordinate
\edef\Y{0}

% Draw header boxes
\node [rectangle,draw,inner sep=5pt,right] at (\InitX,\Y) {\textbf{Alice} };
\node [rectangle, draw, inner sep=5pt, centered] at (\AdvX,\Y) {\textbf{Server}};
\node [rectangle,draw,inner sep=5pt,left] at (\RespX,\Y) {\textbf{Bob}};
\NextLine[0.5]
\Separator{Registration Phase}{}
\NextLine[2]
\ClientAction{$\diffihell.\KeyGen(1^\secpar) \tor (x,g^x)$}
\ServerAction{$\diffihell.\KeyGen(1^\secpar) \tor (y,g^y)$}
\NextLine[1]
\ClientToAdversary{$g^x$}{}
\NextLine[0.5]
\AdversaryAction{$\kid_A \getsr \KID$}{}
\NextLine[1]
\AdversaryAction{$g^x \to \PK[\kid_A]$}
\NextLine[0.5]
\AdversaryToClient{$\kid_A$}{}
\NextLine[0.25]
\ServerToAdversary{$g^y$}{}
\NextLine[0.5]
\AdversaryAction{$\kid_B \getsr \KID$}{}
\NextLine[1]
\AdversaryAction{$g^y \to \PK[\kid_B]$}
\NextLine[0.5]
\AdversaryToServer{$\kid_B$}{}

\Separator{Session Establishment}{}
%%%%%%%%%%%%%%%% Registration Phase %%%%%%%%%%%%%%%%%%%%
\NextLine[2]
\ClientToAdversary{$\kid_B$}{}
\NextLine[0.5]
\AdversaryAction{$g^y \gets \PK[\kid_B] $}
\NextLine[0.5]
\AdversaryToClient{$g^y$}{}
\NextLine[0.5]
\ClientAction{$(g^y)^x \gets \pms_{AB}$}
\ServerToAdversary{$\kid_A$}{}
\NextLine[1]
\ClientAction{$\ctr[32]_{AB} \gets 0$}
\AdversaryAction{$g^x \gets \PK[\kid_A]$}
\NextLine[0.5]
\AdversaryToServer{$g^x$}{}
\NextLine[0.5]
\ServerAction{$(g^x)^y \gets \pms_{AB}$}
\NextLine[1]
\ServerAction{$\ctr[32]_{BA} \gets 0$}
\NextLine[0.5]

\Separator{Encryption}{}
%%%%%%%%%%%%%%%%%%%%%%% Activation Phase %%%%%%%%%%%%%%%%%%
\NextLine[2.5]
\ClientAction{$\salt \getsr \{0,1\}^{128}$}
\NextLine[1]
\ClientAction{$\key_e \gets \SHAtwo(\pms_{AB} \| \salt \| ``\keylbl" )$}
\NextLine[1]
\ClientAction{$\rand \getsr \bits^{32}$}
\NextLine[1]
\ClientAction{$\nonce \gets \ctr \| \rand$}
\NextLine[1]
\ClientAction{$\AD \gets \RID \| \SID \| \kid_A \| \kid_B \| \vers \| \ctype$}
\NextLine[1]
\ClientAction{$c \gets \AESGCM.\Enc(\key_e, \nonce, m, \AD)$}
\NextLine[1]
\ClientAction{$\ctr_{AB}++$}
\NextLine[0.5]
\ClientToServer{$\vers,\ctype, \salt, c, \nonce, \kid_A, \kid_B$}{}

\Separator{Decryption}{}
%%%%%%%%%%%% Sending Message %%%%%%%%%%%%%%%%%%%%%%%
\NextLine[1.5]
\ServerAction{$\key_e \gets \SHAtwo(\pms_{AB} \| \salt \| ``\keylbl" )$}
\NextLine[1]
\ServerAction{$\AESGCM.\Dec(\key_e,\nonce,c,\AD)\to m$}

\end{tikzpicture}
}
\end{center}
\caption{$\LINE$ Letter Sealing Protocol version 2 ($\LINEtwo$).
Registration and Session Establishment are depicted here simultaneously for both parties, but can occur separately.}
\label{fig:line-vers-2(a)}
\end{figure}
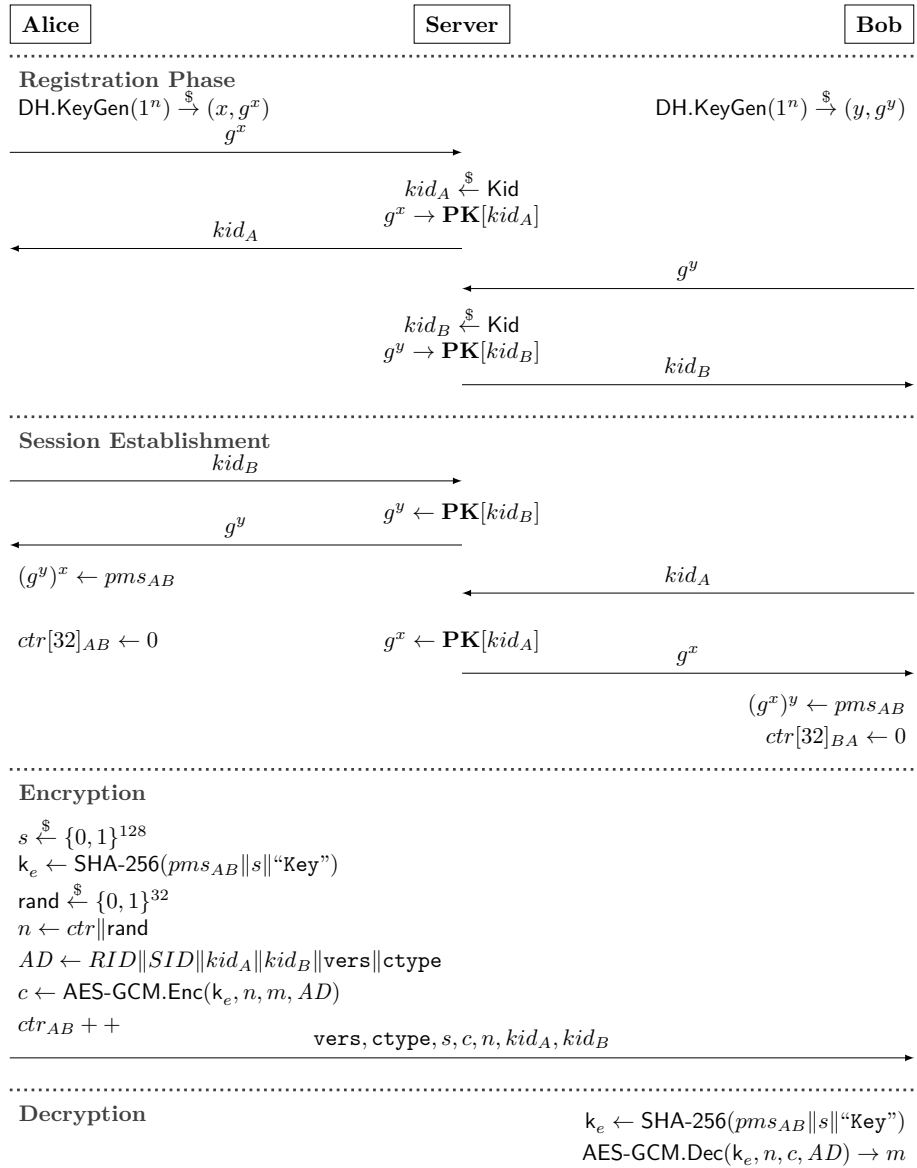

\section{$\LINEDR$ Letter Sealing Protocol version Double Ratchet (Detailed)}
\label{app:lineDR-detail}
Figure \ref{fig:line-vers-DR(a)} depicts the full message flow of $\LINEDR$, including the registration, initial key exchange, and DR-driven ratcheting phases. In the main body, we presented a condensed version showing only the parts where modifications were introduced to achieve the ratcheting mechanism, for clarity and space efficiency. Here, the complete protocol is provided for reference and completeness, illustrating how $\LINEDR$ builds upon the baseline $\LINEtwo$ protocol by integrating the double ratchet mechanism to provide forward secrecy and post-compromise security.

\begin{figure}[htbp!]
\begin{center}
\scalebox{0.85}{
	\centering
\begin{tikzpicture}[yscale=-0.5,>=latex]
\edef\InitX{0}
\edef\ArrowLeft{0}
\edef\ArrowCenter{7}
\edef\ArrowRight{14}
\edef\AdvX{7}
\edef\RespX{14}
% Set the starting Y coordinate
\edef\Y{0}

% Draw header boxes
\node [rectangle,draw,inner sep=5pt,right] at (\InitX,\Y) {\textbf{Alice} };
\node [rectangle, draw, inner sep=5pt, centered] at (\AdvX,\Y) {\textbf{Server}};
\node [rectangle,draw,inner sep=5pt,left] at (\RespX,\Y) {\textbf{Bob}};
\NextLine[0.5]
\Separator{Registration Phase}{}
%\draw [draw=blue] (5.3,16.3) rectangle (0,22);
%\draw [draw=blue] (6.5,26) rectangle (0,32.5);
%\draw [draw=blue] (14,35.5) rectangle (3,47.5);
\NextLine[2]
\ClientAction{$\diffihell.\KeyGen(1^\secpar) \tor (x,g^x)$}
\ServerAction{$\diffihell.\KeyGen(1^\secpar) \tor (y,g^y)$}
\NextLine[1]
\ClientToAdversary{$g^x$}{}
\NextLine[0.5]
\AdversaryAction{$\kid_A \getsr \KID$}{}
\NextLine[1]
\AdversaryAction{$g^x \to \PK[\kid_A]$}
\NextLine[0.5]
\AdversaryToClient{$\kid_A$}{}
%\NextLine[0.1]
\ServerToAdversary{$g^y$}{}
\NextLine[0.5]
\AdversaryAction{$\kid_B \getsr \KID$}{}
\NextLine[1]
\AdversaryAction{$g^y \to \PK[\kid_B]$}
\NextLine[0.5]
\AdversaryToServer{$\kid_B$}{}
\NextLine[0.5]
\Separator{Session Establishment}{}
\NextLine[2.5]
\ClientAction{\color{blue}$\rootkey_{[0]} \gets \emptyset ,$$\asymindexsender \gets 0, ~\symindexsender \gets 0 $}
% \NextLine[1]
% \ClientAction{\color{blue}$\asymindexsender \gets 0, ~\stageindex \gets 0 $}
\NextLine[1]
\ClientAction{\color{blue}$\diffihell.\KeyGen(1^\secpar) \tor (a,g^{a})$}
\NextLine[1]
\ClientAction{{\color{blue}$\rootkey_{[0]}, \chainkey_{[0,0]}\gets \KDF((g^y)^{a}$} $,(g^y)^x)$}
\NextLine[0.5]

\Separator{Encryption}{}
%%%%%%%%%%%%%%%%%%%%%%% Activation Phase %%%%%%%%%%%%%%%%%%
\NextLine[2.5]
%\ClientAction{$\diffihell.\KeyGen(1^\secpar) \tor (x',g^{x'})$}
%\NextLine[1]
%\ClientAction{$\salt \getsr \{0,1\}^{128}$}
%\NextLine[1]
\ClientAction{\color{blue}$\msgkey_{[\asymindexsender,\symindexsender]},\chainkey_{[\asymindexsender,\symindexsender+1]} \gets \KDF(\chainkey_{[\asymindexsender,\symindexsender]})$}
\NextLine[1]
\ClientAction{\color{blue} $\rand \getsr \bits^{32}$, $\nonce  \gets \asymindexsender \| \rand$, }
% \NextLine[1]
% \ClientAction{\color{blue}$(\nonce,\nonce_r)  \gets \asymindexsender \| \rand$}
\NextLine[1]
\ClientAction{$\AD \gets \kid_A \| \kid_B \| \vers \| \ctype \| \color{blue}g^{a} \| \symindexsender$}
\NextLine[1]
\ClientAction{$\ciphertxt \gets \AEAD.\Enc({\color{blue}\msgkey_{[\asymindexsender,\symindexsender]}, \nonce }, \msg, \AD)$}
\NextLine[1]
\ClientAction{\color{blue}$\symindexsender++$}
\NextLine[1.5]
\ClientToServer{$ \AD, \ciphertxt, {\color{blue}\nonce}$}{}

%%%%%%%%%%%% Decryption %%%%%%%%%%%%%%%%%%%%%%%
\Separator{Decryption}{}
\NextLine[2]
%\ServerAction{$\ifsm (\rootkey_{[0]} = \emptyset)$}
%\NextLine[1]
\ServerAction{\color{blue}$\ifsm (\rootkey_{[0]} = \emptyset)$:}
\NextLine[1]
\ServerAction{\color{blue}$\asymindexrcvr \gets 0, ~\symindexrcvr \gets 0 $}
% \NextLine[1]
% \ServerAction{\color{blue}$\asymindexrcvr \gets 0, ~\stageindex \gets 0 $}
\NextLine[1]
\ServerAction{\color{blue}$\rootkey_{[0]}, \chainkey_{[0,0]}\gets \KDF((g^a)^{y},(g^x)^y)$}
\NextLine[1]
\ServerAction{\color{blue}$\msgkey_{[0,0]}, \chainkey_{[0,1]} \gets \KDF(\chainkey_{[0,0]})$}
\NextLine[1]
\ServerAction{\color{blue}$\symindexrcvr++$}
\NextLine[1]
\ServerAction{\color{blue}$\textbf{end~if}$}
\NextLine[1]
\ServerAction{\color{blue}$\ifsm (\nonce.\asymindexsender > \asymindexrcvr)$:}
\NextLine[1]
\ServerAction{\color{blue}$(\rootkey_{[\asymindexrcvr+1]}, \chainkey_{[\asymindexrcvr+1,0]} \gets \KDF((g^{a})^{b},\rootkey_{[\asymindexrcvr]}))$}
\NextLine[1]
\ServerAction{\color{blue}$\textbf{end~if}$}
\NextLine[1]
\ServerAction{\color{blue}$\whilesm (\symindexsender > \symindexrcvr)$:}
\NextLine[1]
\ServerAction{\color{blue}$\asymindexsender \gets \nonce.\asymindexsender$:}
\NextLine[1]
\ServerAction{\color{blue}$\msgkey_{[\asymindexsender,\symindexrcvr]},\chainkey_{[\asymindexsender,\symindexrcvr+1]} \gets \KDF(\chainkey_{[\asymindexsender,\symindexrcvr]}) $}
\NextLine[1]
\ServerAction{\color{blue}$ \symindexrcvr++$}
\NextLine[1]
\ServerAction{\color{blue}$\textbf{end~while}$}
\NextLine[1]
\ServerAction{$\AEAD.\Dec({\color{blue}\msgkey_{[\asymindexrcvr,\symindexrcvr]},\nonce}, \ciphertxt,\AD)\to \msg$}
\NextLine[1]
\ServerAction{\color{blue}$\ifsm (\nonce.\asymindexsender \geq \asymindexrcvr)$:}
\NextLine[1]
\ServerAction{\color{blue}$(\diffihell.\KeyGen(1^\secpar) \tor (b,g^{b}))$}
\NextLine[1]
\ServerAction{\color{blue}$(\rootkey_{[\asymindexrcvr+2]}, \chainkey_{[\asymindexrcvr+2,0]}\gets \KDF((g^{a})^{b},\rootkey_{[\asymindexrcvr+2]}) )$}
\NextLine[1]
\ServerAction{\color{blue}$\asymindexrcvr \gets \asymindexrcvr+2$}
\NextLine[1]
\ServerAction{\color{blue}$\textbf{end~if}$} 
\end{tikzpicture}
}
\end{center}
\caption{In this figure, we present the $\LINEDR$ (Letter Sealing with double ratchet) version of the protocol. Text in {\color{blue}blue} highlights our modifications to $\LINEtwo$, specifically the integration of double ratcheting mechanisms inspired by the Signal protocol.}
\label{fig:line-vers-DR(a)}
\end{figure}
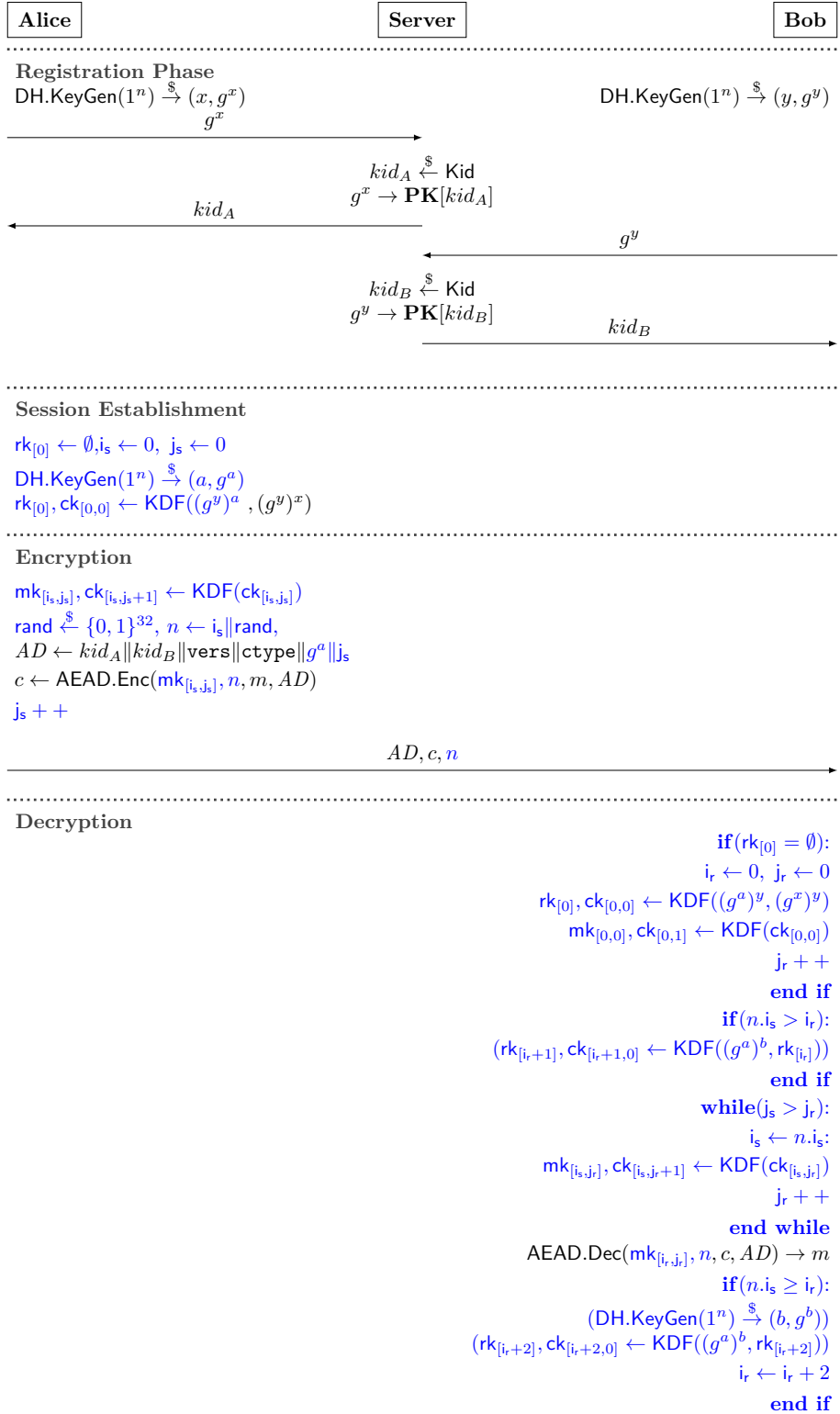

\section{$\LINEone$ Letter Sealing Protocol version 1} \label{Sec:line1}

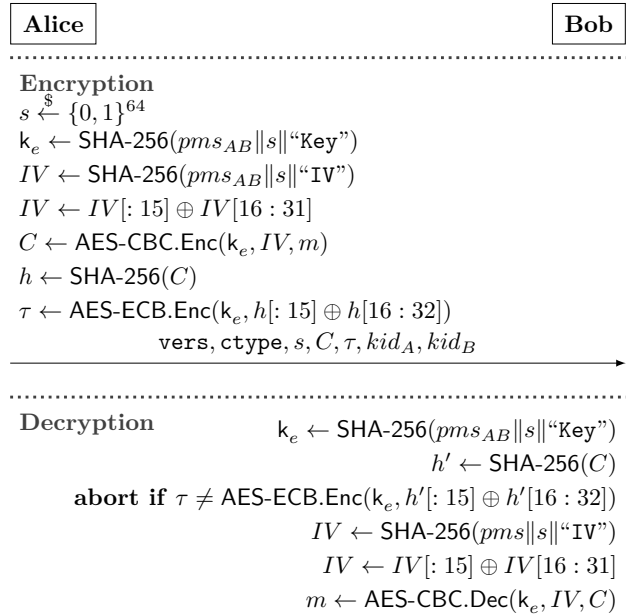
\begin{figure}[h]
\begin{center}
\scalebox{0.9}{
	\centering
\begin{tikzpicture}[yscale=-0.5,>=latex]
\edef\InitX{0}
\edef\ArrowLeft{0}
\edef\ArrowCenter{4.5}
\edef\ArrowRight{9}
\edef\AdvX{4.5}
\edef\RespX{9}
% Set the starting Y coordinate
\edef\Y{0}

% Draw header boxes
\node [rectangle,draw,inner sep=5pt,right] at (\InitX,\Y) {\textbf{Alice} };
\node [rectangle,draw,inner sep=5pt,left] at (\RespX,\Y) {\textbf{Bob}};
\NextLine[0.5]
\Separator{Encryption}{}
\NextLine[2]
\ClientAction{$\salt \getsr \bits^{64}$}
\NextLine[1]
\ClientAction{$\key_e \gets \SHAtwo(\pms_{AB} \| \salt \| ``\keylbl" )$}
\NextLine[1]
\ClientAction{$\IV \gets \SHAtwo(\pms_{AB} \| \salt \| ``\ivlbl")$}
\NextLine[1]
\ClientAction{$\IV \gets \IV[:15] \oplus \IV[16:31]$}
\NextLine[1]
\ClientAction{$C \gets \AESCBC.\Enc(\key_e,\IV,m)$}
\NextLine[1]
\ClientAction{$h \gets \SHAtwo(C)$}
\NextLine[1]
\ClientAction{$\tau \gets \AESECB.\Enc(\key_e, h[:15] \oplus h[16:32])$}
\NextLine[1]
\ClientToServer{$\vers, \ctype, \salt, C, \tau, \kid_A, \kid_B$}{}

\Separator{Decryption}{}
\NextLine[1.5]
\ServerAction{$\key_e \gets \SHAtwo(\pms_{AB} \| \salt \| ``\keylbl" )$}
\NextLine[1]
\ServerAction{$h' \gets \SHAtwo(C)$}
\NextLine[1]
\ServerAction{\textbf{abort if} $\tau \neq \AESECB.\Enc(\key_e, h'[:15] \oplus h'[16:32])$}
\NextLine[1]
\ServerAction{$\IV \gets \SHAtwo(\pms \| \salt \| ``\ivlbl")$}
\NextLine[1]
\ServerAction{$\IV \gets \IV[:15] \oplus \IV[16:31]$}
\NextLine[1]
\ServerAction{$m \gets \AESCBC.\Dec(\key_e,\IV,C)$}

\end{tikzpicture}
}
\end{center}
\caption{$\LINE$ Letter Sealing Protocol Version 1 or $\LINEone$. The Registration Phase and Session Establishment are identical to $\LINEtwo$ (except $\LINEone$ does not maintain $\ctr$), and thus these steps are omitted.}
\label{fig:line-vers-1}
\end{figure}

An algorithmic description of $\LINE$'s Letter Sealing protocol (version 1, $\LINEone$) is shown in figure \ref{fig:line-vers-1}. As in $\LINEtwo$, $\LINEone$ clients generate a long-term elliptic-curve Diffie-Hellman secret key pair $(x, g^x)$ (using Curve 25519) and register them with the $\LINE$ messaging server as described in Section \ref{sec:linetwo}. The $\LINE$ server then registers and associates each public key with a unique key ID $\kid$. At some point, the client (hereafter referred to as the sender) wishes to send an encrypted message to some recipient. The sending client will retrieve the current public key of the recipient $g^y$ by communicating with the $\LINE$ messaging server as in $\LINEtwo$. The sender will generate a shared secret $\pms_{AB} \gets (g^y)^{x}$. These public keys can be verified out-of-band using $\LINE$'s fingerprinting algorithm.

For each message, the sender will generate a fresh encryption key $\key_e$ and an initialization vector ($\IV$). The sender generates an 8-byte value $\salt \getsr \bits^{64}$ from a secret random number generator. The encryption key $\key_e \gets \SHAtwo(\pms_{AB}\|\salt\|``\keylbl'')$ is computed by hashing the shared secret $\pms_{AB}$, $\salt$, and the string $``\keylbl"$, and $\IV \gets \SHAtwo(\pms_{AB} \| \salt \| ``\ivlbl'')$ is obtained in a similar way, only replacing $``\keylbl''$ with $``\ivlbl''$. The final $\IV$ is computed as the XOR of $\IV$'s two 16-byte halves, i.e. $\IV \gets \IV[:15] \oplus \IV[16:32]$. The message payload $m$ is encrypted using $\AESCBC$ and outputs the ciphertext $c \gets \AESCBC.\Enc(\key_e,\IV,m)$. A message authentication code $\tau \gets \AESECB.\Enc $ $(\key_e, \SHAtwo(C)[:15] \oplus \SHAtwo(C)[16:32])$ is calculated for the ciphertext by encrypting a hash of the ciphertext using AES in ECB mode, using the encryption key $\key_e$, and the result is included in the ciphertext sent to the recipient. The ciphertext is the concatenation of the Letter Sealing protocol version $\vers$, the content type $\ctype$, $\salt$, ciphertext $C$, $\tau$ and key identifiers $\kid_{A}, \kid_{B}$.

Recipients use the sender key identifier $\kid_{A}$ to retrieve the sender's public key $g^x$ used to generate the shared secret $\pms_{AB}$, and the receiver's key identifier $\kid_{B}$ to verify that the message can be decrypted using the receiver's current local private key. Once the recipient determines they can decrypt a message, they derive the shared secret, symmetric encryption key, and IV using the same process as the sender. The recipient then locally calculates the MAC of the received ciphertext and compares it with $\tau$, and rejects if $\tau \neq \AESECB.\Enc(\key_e,\SHAtwo(C)[:15]\oplus\SHAtwo(C)[16:32])$. Otherwise, the receiver decrypts the ciphertext and outputs the message $m \gets \AESCBC.\Dec(\key_e,$ $\IV,C)$.

\section{Evaluation and Experimental Results} \label{sec:Exp}

In this experiment, we set up a testbed environment to examine the structure of packets sent and received by the $\LINE$ messaging application, aiming to determine whether they align with the specifications outlined in its white paper.  

\begin{figure}[ht!b]
\begin{center}
%\vspace{-1cm}
\includegraphics[width=0.45\textwidth]{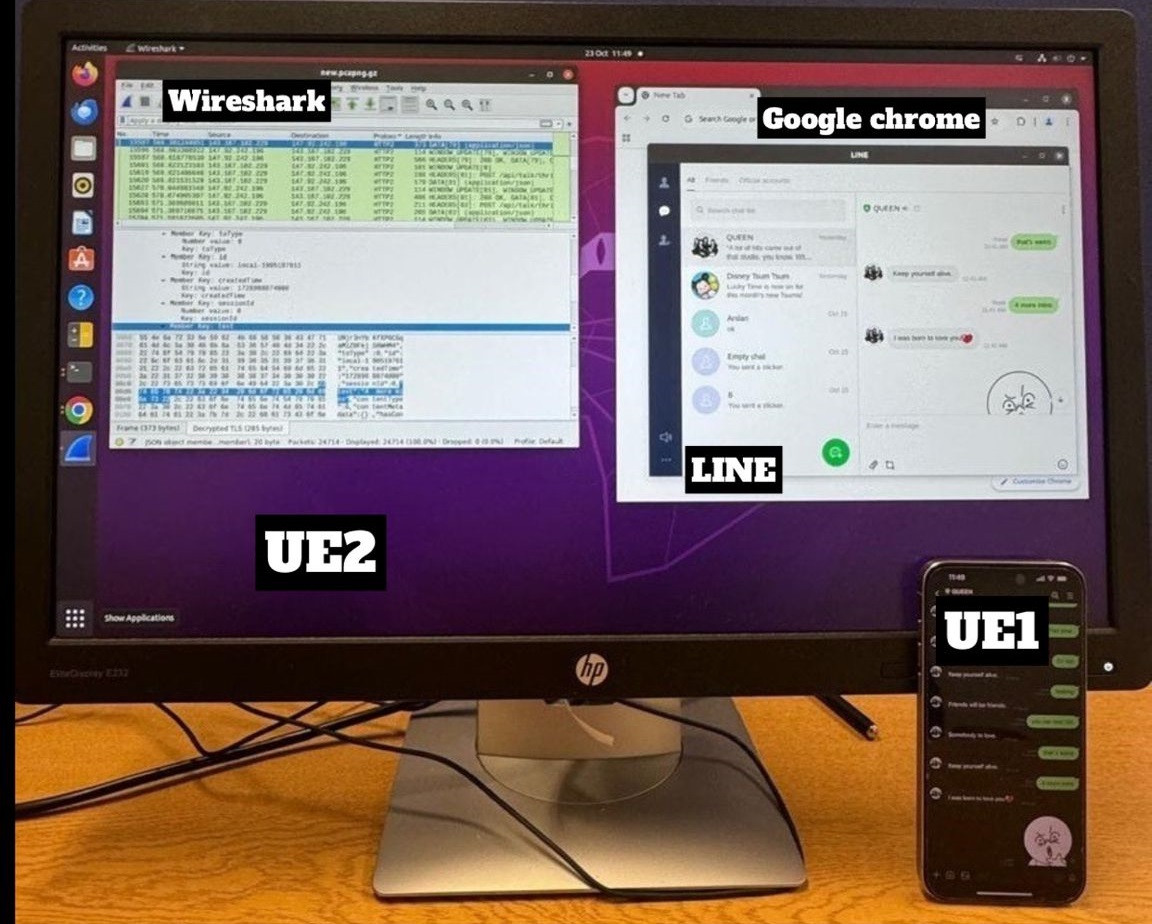}
\end{center}
\caption{A display of our physical experiment testbed.}
\label{fig:testbed}
\end{figure}

As illustrated in Figure \ref{fig:testbed}, the setup involved an iPhone device denoted UE1 (User Environment) running iOS version 18.1 with $\LINE$ version 14.16.1 installed. On the desktop side, we used an HP PC denoted UE2 machine with a Core i5, 3.3 GHz CPU, 16 GB of RAM, and a 64-bit operating system running Ubuntu 20.04.6. The PC ran Google Chrome version 129.0.6668.100 with the $\LINE$ Chrome extension version 3.5.1. To capture and decrypt the network traffic, we used Wireshark version 3.2.3. The experiment began with installing and registering the $\LINE$ application on the iPhone. After setting up the account, the user verified the application's functionality by sending and receiving messages. Following this, the $\LINE$ Chrome extension was installed on the Ubuntu machine. $\LINE$ allows users to log in by scanning a QR code, which was scanned using the iPhone application, enabling seamless communication between the mobile app and the Chrome extension. During this step, Google Chrome was used to access the $\LINE$ extension. At the same time, Wireshark ran in the background to capture the network traffic between the $\LINE$ client and the server, and the SSL/TLS decryption setup described in the article \cite{wiresharkExp} was followed to analyze encrypted data packets. The data capture process lasted approximately 15 minutes, during which typical one-to-one messaging interactions were conducted. This included communications with other users and interactions with $\LINE$'s automated bots. %All captured and sorted data can be found in our repository \cite{fullversion}. 
The detailed dataset collected during this experiment is available in an anonymous repository$\footnote {\href{https://osf.io/npdtw/?view_only=84621b38349a4f858f5f897c195e80c4}{Anonymous Repository}}$. Upon analysis of the captured data, we identified two key findings. First, the structure of $\LINE$ messages was evident where we demonstrate the format and metadata of the communication. Even though the content of user-to-user messages remained encrypted, the structure itself could still be analyzed. Second, conversations with $\LINE$'s automated bots appeared to lack encryption, as we could view the entire exchange of messages, including both sent and received content. This raises significant concerns about the security of bot interactions, mainly as unencrypted bot communications can expose sensitive user information. Independently and concurrently to our work Chou \etal~ discovered similar behaviours across a range of secure messaging services including $\LINE$, which allows malicious actors can exploit such vulnerabilities in bot communication to compromise user identities or harvest personal data \cite{chatbotExploitation}. While the encryption applied to standard user-to-user communication was effective, these unencrypted bot conversations highlight potential security risks within the $\LINE$ platform. We note that $\LINE$'s 2022 encryption transparency report \cite{LineWP2022} discusses on the scope of E2EE, through their "Letter Sealing" feature. They provide a table that discusses encryption deployment status for each message (read: content) type, and state that test messages apply "Letter Sealing" encryption. They also provide a list of "Letter Sealing Exceptions".

\begin{figure}[t]
\centering
\noindent
\begin{minipage}{0.5\textwidth}
\vspace{0.8cm}
  \begin{verbatim}
    struct {
        long from;
        long to;
        char toType;
        long id;
        long createdTime;
        long deliveredTime;
        bool hascontent;
        char contentType;
        struct {
            char e2eeVersion;
            long seq;
            } contentMetadata;
        long sessionId;
        long chunks;
        }
\end{verbatim}

\label{fig:pkt_strct}
\end{minipage}%
\begin{minipage}{0.5\textwidth}
\begin{verbatim}
    struct {
        long from;
        long to;
        char toType;
        long id;
        long createdTime;
        long deliveredTime;
        bool hascontent;
        char contentType;
        struct {
            long BOT_TAG2;
            char BOT_ORIGIN;
            bool BOT_CHECK;
            char BOT_TRACK;
            } contentMetadata;
        long sessionId;
        long text;
        }
\end{verbatim}
\end{minipage}
\caption{Left: $\LINEtwo$ Ciphertext Structure and Right: $\LINEtwo$ Bot Ciphertext Structure}
\label{fig:ct&botpkt_strct}
\end{figure}

The ciphertext packet structure of $\LINEtwo$ is shown at the left of Figure \ref{fig:ct&botpkt_strct}. Curiously, the ciphertext structure in practice diverges from that described in the Technical Whitepaper (and described in Figure \ref{fig:line-vers-2}), particularly including additional metadata: $\timefield$, $\delfield$, $\idfield$, $\seqfield$. %\ben{figure needs to be updated to include these fields}. 
The observed structure encapsulates various fields providing detailed metadata about the message. It includes a sender ID ($\fromfield$) and a recipient ID ($\tofield$), representing the unique identifiers of the message's sender and recipient. Additionally, an (empty) recipient type ($\totypefield$). Each message is assigned a unique message ID ($\idfield$), a timestamp ($\timefield$) marks when the message was created, and a timestamp ($\delfield$) marks when the message was delivered. A boolean flag ($\hasfield$) indicates whether the message contains actual content, while the content type ($\contentfield$) specifies the nature of the message, such as text, image, or video. Embedded within this structure is metadata related to encryption: the end-to-end encryption (E2EE) version is indicated by $\eversfield$, and sequence ($\seqfield$) that matches with message ID ($\idfield$). Additionally, the structure includes an (empty) session identifier ($\sidfield$). The chunks field ($\chunksfield$) contains the ciphertext, representing the encrypted parts of a message.  Specifically, $\LINEtwo$ parses $\chunksfield$ as the $\salt$, $\AESGCM$ ciphertext, $\ctr\|\rand$, $\kid_{A}$ and $\kid_{B}$ respectively.

The right of Figure \ref{fig:ct&botpkt_strct} details the structure of messages to-and-from bots. It is primarily the same structure as the Letter Sealing messages, but with distinct fields in its \\ $\contentMDfield$ section, such as $\bottagfield$ (a constant base64 string used across different bots), $\botoriginfield$ (indicating the source of the bot message, specifically \texttt{line:chat}), $\botcheckfield$ (a boolean field that indicates that the message is from a bot), and $\bottrackfield$ (which tracks the bot's API use, specifically \texttt{api=pu}).  Additionally, the bot message replaces the $\chunksfield$ (containing the $\LINEtwo$ ciphertext) with a $\textfield$ field, which contains plainttext content rather than encrypted messages. It is important to note that, according to the $\LINE$ Encryption Report \cite{LineWP2022}, end-to-end encryption applies to user messages, but this structure shows that bot messages are not encrypted in $\LINEtwo$.

%\begin{notebox}
\textbf{Remarks:} However, the lack of an end-to-end encryption application for bot communications is not mentioned in either discussion. We note that bots exist for both banking and health services on $\LINE$, for which users might reasonably expect private communications. By not encrypting bot conversations, $\LINE$ exposes sensitive user data to unnecessary risks since they fall outside of commonly expected E2EE guarantees to access unrelated private messages, thus failing to achieve data minimization.
%\end{notebox}

\section{Cryptographic Preliminaries} \label{appen:CP}

\begin{wrapfigure}{r}{0.5\textwidth}\label{fig:prf-sec}
    \procedureblock{$\Exp{\prf}{\PRF}(\adversary)$}{
    \key \getsr \KeySpace; b \getsr \bits \\
    y \gets \adversary^{\mathcal{O}(x)\to \PRF(k,x)}() \\
    \pcif y \in (x_0, \ldots, x_i)\\
    \t \pcreturn \bot\\
    z_0 \gets \PRF(k,y); z_1 \getsr \bits^{\lambda}\\
    b' \gets \adversary^{\mathcal{O}(x)\to \PRF(k,x)}(z_b) \\
    \pcif y \in (x_{i+1}, \ldots, x_j)\\
    \t \pcreturn \bot\\
    \pcreturn (b = b')} 
\caption{The $\PRF$ security experiment.}
\end{wrapfigure}

This section introduces the cryptographic primitives that underlie the assurance of confidentiality, integrity, and authenticity in secure messaging protocols. These cryptographic primitives are essential building blocks for the development of robust security frameworks.

\subsection{Hash Functions and Collision-Resistance}

A hash function is a deterministic algorithm $\Hash: \bits^* \to \bits^\lambda$ which, given a bit string $m$ of arbitrary length, outputs a hash value $w = \Hash(m)$ in the space $\bits^\lambda$.

\begin{definition}[Collision-Resistant Hash Function]\label{def:coll-sec} We define the advantage of a probabilistic polynomial-time (PPT) algorithm $\adversary$ in breaking the collision-resistance security of the hash function $\Hash$ as 
$\Adv{\coll}{\Hash}(\adversary) = \Pr[\Hash(m) = \Hash(m') \land m \ne m':(m, m') \gets \adversary].$ We say that a hash function $\Hash$ is collision-resistant if for all non-trivial PPT $\adversary$, $\Adv{\coll}{\Hash}(\adversary)$ is negligible.
\end{definition}

\subsection{Pseudo-Random Functions}

Our formalization of a pseudo-random function $\PRF$ and the security of $\PRF$ follows \cite{PRF-def} 

A $\PRF$ with key space $\KeySpace$ and input space $\bits^*$ is a deterministic algorithm. On input a key $k \in \KeySpace$ and an input string $x \in \bits^*$, the algorithm outputs a value $\PRF(k,x) \in \bits^{\lambda}$.

\subsubsection{Security of $\PRF$}

The security of a $\PRF$ is given as a game played between a challenger $\challenger$ and an adversary $\adversary$.

\begin{definition}[Security of Pseudorandom Functions]\label{def:prf-sec}
We define the advantage of a probabilistic polynomial-time (PPT) algorithm $\adversary$ in breaking the security of the pseudo-random function $\PRF$ as 
$\Adv{\prf}{\PRF}(\adversary) = \Pr(\Exp{\prf}{\PRF}(\adversary)=1)$. We say that $\PRF$ is secure if for all PPT $\adversary$ $\Adv{\prf}{\PRF} (\adversary)$ is negligible in the security parameter $\lambda$.
\end{definition}

\subsection{AEAD: Authenticated Encryption with Associated Data} \label{sec:AEAD}

The $\LINE$ protocol uses AES-GCM  \cite{LineWP} %\ben{todo: add reference}, 
which is an $\AEAD$ scheme. On a high level, an \emph{authenticated encryption with associated data} scheme is a cryptographic primitive that achieves authenticity and confidentiality of the plaintext but also achieves authenticity of a public data field known as associated data.

An \emph{authenticated encryption with associated data} ($\AEAD$) scheme defined in \cite{AEADdef} is a tuple of algorithms $\AEAD = \{\KGen,\ENC, $ $\DEC\}$. $\KGen$ is a probabilistic function that takes as input a security parameter $1^\lambda$ and outputs a randomly-sampled key $k$ from the keyspace $\KeySpace = \bits^{\lambda}$. $\ENC$ is a deterministic function that takes as input a key $k$, a nonce $N$, an associated data field $AD$ and a message $\msg$ , and outputs a ciphertext $\ciphertxt$. Finally, $\DEC$ is a deterministic function that takes as input a key $k$, a nonce $N$, an associated data field $AD$ and a ciphertext $\ciphertxt$, and outputs a message $\msg$  or a failure symbol $\bot$.

\subsubsection{Security of Authenticated Encryption with Associated Data}

We now turn to defining the security of $\AEAD$ schemes. First, we define $\auth$-security, which broadly captures the security of an $\AEAD$ scheme against ciphertext forgery attacks. This notion is formalised in Definition \ref{def:aead-auth-sec}. Then, we define $\conf$-security, which captures the inability of an adversary to distinguish encryptions of different plaintexts under chosen-ciphertext queries, as formalised in Definition \ref{def:aead-conf-sec}.

\begin{definition}
[$\auth$-security of AEAD schemes]\label{def:aead-auth-sec}
We define the advantage of a probabilistic polynomial-time (PPT) algorithm $\adversary$ in breaking the $\auth-security$ of the authenticated encryption with associated data scheme $\AEAD$ (as given in Figure \ref{fig:aead-sec}) as 
$\Adv{\auth}{\AEAD}(\adversary) = \Pr(\Exp{\auth}{\AEAD}(\adversary)=1)$. We say that $\AEAD$ is $\auth$-secure if for all PPT $\adversary$ $\Adv{\auth}{\AEAD} (\adversary)$ is negligible in the security parameter $\lambda$.
\end{definition}

Now we define $\conf$-security taken from \cite{conf-AEAD} as: the advantage of a probabilistic polynomial-time (PPT) algorithm $\adversary$ in breaking the $\conf$-security of the authenticated encryption with associated data scheme $\AEAD$ as:

\begin{definition}[$\conf$-security of AEAD schemes]\label{def:aead-conf-sec}
Let $\AEAD$ = ($\KeyGen$, $\ENC$, $\DEC$) be an $\AEAD$ scheme. Let $\init_0$ be the
security experiment (see Figure \ref{fig:aead-sec}) with $b$ set to 0 and $\init_1$ be the security experiment with $b$ set to 1. We say that the adversary wins the $\AEAD$ game if $b' = b$ and define the advantage function:
$$\Adv{\aead}{\AEAD,\adversary} = |\Pr ([1 \gets \adversary|b = 1]) - \Pr ([1 \gets \adversary|b = 10| )$$

Intuitively $\AEAD$ is secure, if for all PPT algorithms $\adversary$ it holds $\Adv{\aead}{\AEAD,\adversary}$ is negligible.
\end{definition}

\begin{figure*}[t]
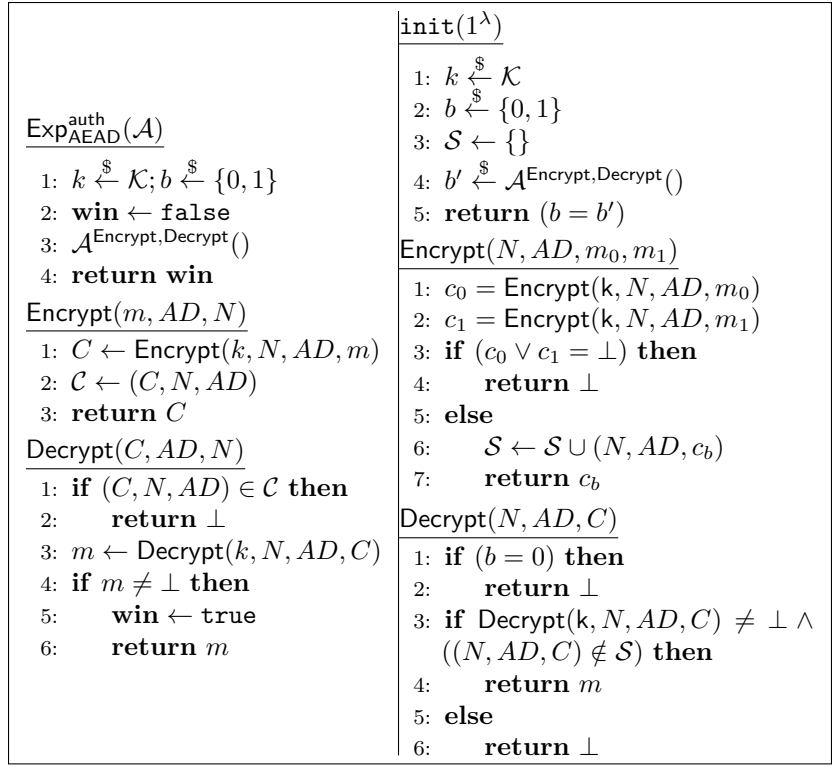

\centering
\fbox{
\begin{minipage}{0.35\textwidth}
    %\hspace{0.25cm}
    \underline{$\Exp{\auth}{\AEAD}(\adversary)$}
    \begin{algorithmic}[1]
    \State $k \getsr \KeySpace; b \getsr \bits$ 
    \State $\win \gets \falseflag$
    \State $\adversary^{\ENC,\DEC}()$ 
    \State \Return $\win$ 
    \end{algorithmic}

%%%%%%%%%%%%%%%%%%%%%%%%%%%%%%%%%%%%%%%%%%%%%%%%%%%%%%%%%5
    
    \underline{$\ENC(m,AD,N)$}
    \begin{algorithmic}[1]
    \State $C \gets \ENC(k,N,AD,m)$ 
    \State $\mathcal{C} \gets (C,N,AD)$
    \State \Return $C$
    \end{algorithmic}

%%%%%%%%%%%%%%%%%%%%%%%%%%%%%%%%%%%%%%%%%%%%%%%%%%%%%%%
   \underline{$\DEC(C,AD,N)$}
    \begin{algorithmic}[1]
    \If{$(C,N,AD) \in \mathcal{C}$}
    \State \Return $\bot$
    \EndIf
    \State $\msg \gets \DEC(k,N,AD,C)$
    \If {$\msg \neq \bot$}
    \State $\win \gets \trueflag$
    \State \Return $\msg$ 
    \EndIf
    \end{algorithmic}
    \end{minipage}
    \vline
\begin{minipage}{0.4\textwidth}
    %\hspace{0.25cm}
    \underline{$\init(1^\lambda)$}   
    \begin{algorithmic}[1]
    \State $k \getsr \KeySpace$ 
    \State $b \getsr \bits $
    \State $\SaltSpace \gets \{\} $
    \State $b' \getsr \adversary^{\ENC,\DEC}()$
    \State \Return $(b=b') $
    \end{algorithmic}

%%%%%%%%%%%%%%%%%%%%%%%%%%%%%%%%%%%%%%%%%%%%%%%%%%%%%%%%%%%%%

    \underline{$\ENC(N,AD,\msg_0,\msg_1)$}
    \begin{algorithmic}[1]
    \State $\ciphertxt_0 = \ENC(\key, N, AD, \msg_0)$ 
    \State$\ciphertxt_1 = \ENC(\key, N, AD, \msg_1) $
    \If {$(\ciphertxt_0 \lor \ciphertxt_1 = \bot) $}
    \State \Return $\bot $
    \Else 
    \State $\SaltSpace \gets \SaltSpace \cup {(N,AD,\ciphertxt_b)}$ 
    \State \Return $\ciphertxt_b $
    \EndIf
    \end{algorithmic}

%%%%%%%%%%%%%%%%%%%%%%%%%%%%%%%%%%%%%%%%%%%%%%%%%%%%%%%%%%%%%%%   
    \underline{$\DEC(N,AD,C)$}
    \begin{algorithmic}[1]
    \If {$(b=0)$} 
    \State \Return $\bot$ 
    \EndIf
    \If {$\DEC(\key,N,AD,C) \neq \bot \land ((N,AD,C) \notin \SaltSpace) $}
    \State \Return $\msg$
    \Else
    \State \Return $\bot$
    \EndIf
    \end{algorithmic}
    \end{minipage}
}
\caption{The left half illustrates the $\AEAD-\auth$ security experiment. while the right half illustrates $\AEAD-\conf$ security experiment}
\label{fig:aead-sec}
\end{figure*}

\subsection{Key Derivation Functions}

A Key Derivation Function ($\KDF$) is a cryptographic primitive used in key agreement and encryption protocols to derive one or more cryptographic keys from a private input string. The derived keys ensure confidentiality and integrity of stored or transmitted data. A secure $\KDF$ produces outputs that are computationally indistinguishable from random binary strings of equal length, even when public inputs are known to the adversary, ensuring no efficient distinguisher can succeed beyond random guessing.

\begin{figure}[htbp]
    \centering
    \includegraphics[width=0.44\textwidth]{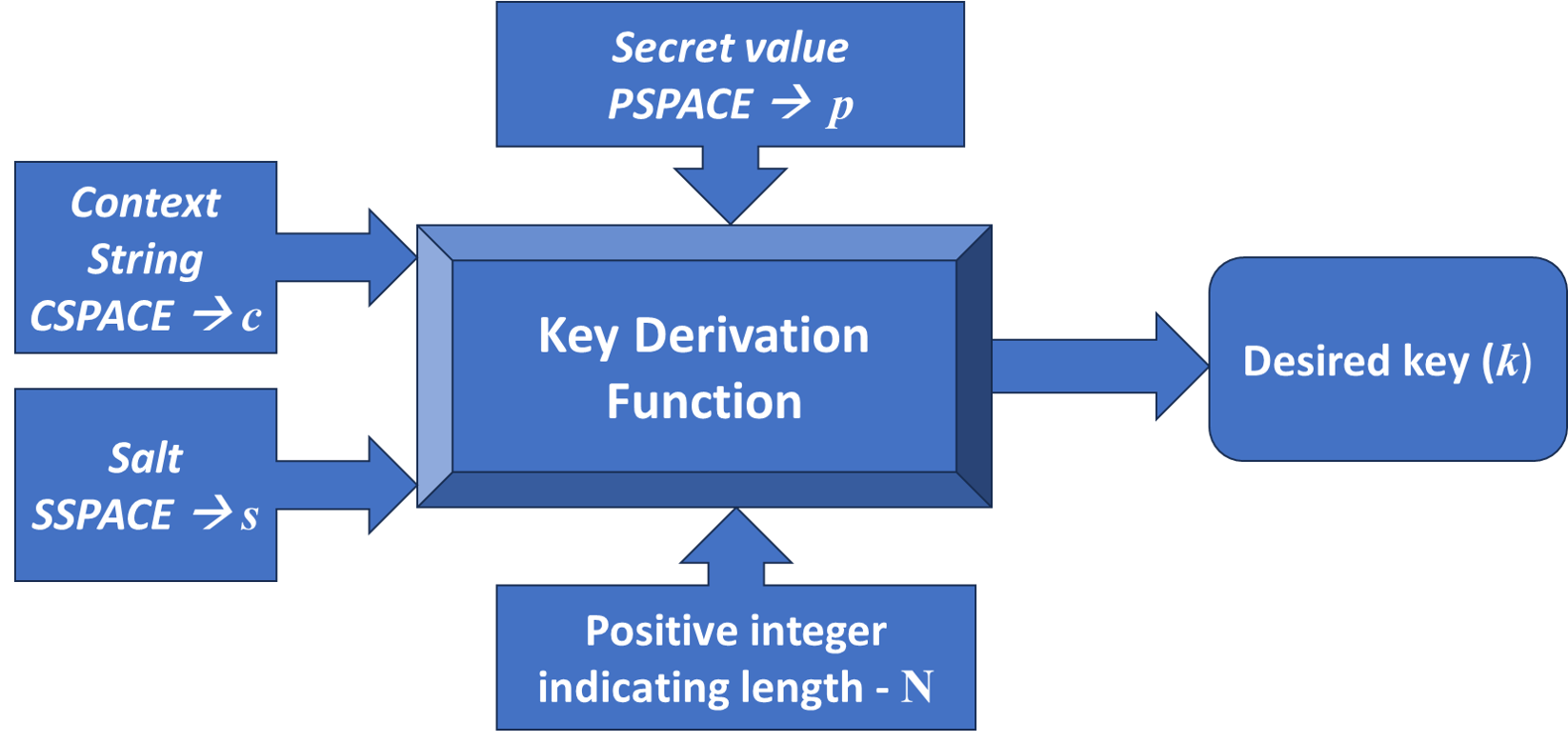} 
    \caption{Key Derivation Function} 
\label{fig:kdf} 
\end{figure}

A $\KDF$ typically takes as input a private string, a salt to introduce randomness and prevent precomputed attacks, and a context string carrying additional protocol information or party identities. It transforms potentially non-uniform secret entropy into uniform keys of suitable length and distribution for symmetric-key protocols, resisting attacks that aim to recover the underlying secret. Selecting an appropriate $\KDF$ depends on the desired security properties; a well-known example is HKDF (HMAC-based Key Derivation Function) \cite{HKDF}, which employs a hash function such as SHA-256 within the HMAC construction \cite{KDF}.

\begin{definition}[$\KDF$] \label{def:kdf} 
Key Derivation Function is a deterministic function denoted as $\KDF(\rho, s, c, N)$ that takes inputs as:
\begin{itemize}
\item[-] $\rho$: The secret value chosen from a set of possible private strings ($PSPACE$).
$$ \rho \leftarrow PSPACE $$
\item[-] $s$: The salt value chosen from a set of possible public random strings ($SSPACE$).
$$ s \leftarrow SSPACE $$
\item[-] $c$: The public context string chosen from a set of possible public context strings ($CSPACE$).
$$ c \leftarrow CSPACE $$
\item[-] $N$: A positive integer indicating the desired number of bits to be produced by the $\KDF$ function.
\end{itemize}
The $\KDF$ function transforms these inputs into an $N$-bit cryptographic key $k$ according to the defined transformation rules:
$ \KDF(\rho,s,c,N) \rightarrow \key $ 
and generates an output as:
\begin{itemize}
\item[-] $k$ is the derived cryptographic key.
\end{itemize}
\end{definition}

 %It is utilized in a wide range of cryptographic applications and protocols e.g. secure communication protocols like Transport Layer Security (TLS), password-based key derivation, key establishment schemes, disk encryption systems and many others. Use of $\KDF$ algorithm helps these protocols to generate secure and uniformly distributed cryptographic keys from potentially non-uniform secret inputs. This not only enables the establishment of confidential and secure communication channels but also helps in providing protection against sensitive data, and overall strengthening of the cryptographic systems.

\begin{wrapfigure}{l}{0.5\textwidth}
    \centering
    \scalebox{0.85}{\centering
\begin{tikzpicture}[yscale=-0.5,>=latex]
\edef\InitX{0}
\edef\ArrowLeft{0}
\edef\ArrowCenter{3.5}
\edef\ArrowRight{7}
\edef\AdvX{3.5}
\edef\RespX{7}
% Set the starting Y coordinate
\edef\Y{0}

% Draw header boxes
\node [rectangle,draw,inner sep=5pt,right] at (\InitX,\Y) {\textbf{Adversary $\adversary$} };

\node [rectangle,draw,inner sep=5pt,left] at (\RespX,\Y) {\textbf{Challenger $\challenger$}};

\NextLine[2]
\ServerAction{$\rho \gets PSPACE$}
\NextLine[1]
\ServerAction{$SSPACE \tor s$}

\NextLine[1]
\ServerToClient{$s$}{}

\NextLine[1]
\ClientAction{$c_i \gets CSPACE $}

\NextLine[1]
\ClientToServer{$c_i$}{}

\NextLine[1]
\ServerAction{$\key_i \gets \KDF(\rho, s, c_i, N)$}

\NextLine[1]
\ServerToClient{$\key_i$}{}

\NextLine[1.5]
\ClientAction{$c \gets CSPACE$}

\NextLine[1]
\ServerAction{$\{0,1\} \tor b$}
\NextLine[1]
\ServerAction{if $(b=0)$}
\NextLine[1]
\ServerAction{$\key^\prime \gets \KDF(\rho, s, c_i, N)$}
\NextLine[1]
\ServerAction{else}
\NextLine[1]
\ServerAction{$\{0,1\}^N \tor \key^\prime$}

\NextLine[1]
\ServerToClient{$\key^\prime$}{}

\end{tikzpicture}
}
    \caption{KDF Security Experiment}
\label{fig:kdf_security}
\end{wrapfigure}
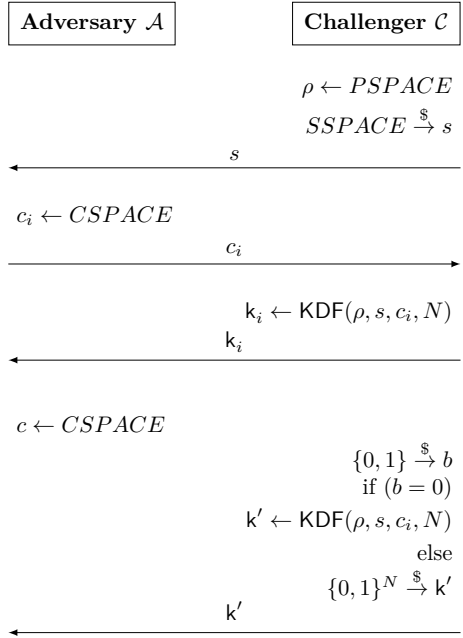

\subsubsection {Security of $\KDF$}
\par A secure $\KDF$ generates a derived key that is computationally indistinguishable from a randomly generated key. For any efficient adversary $\adversary$, the advantage of distinguishing between the output of the $\KDF$ and a truly random key is negligible: \\
$ Adv_{\adversary}^{\KDF} = |Pr [\adversary(K) = 1] - Pr [ \adversary (random(N)) = 1 ]| \le negl (\lambda) $

where $K$ represents the derived key output by the $\KDF$, $\text{random}(N)$ denotes a randomly generated key of length $N$, $\Pr$ denotes probability, and $\text{negl}(\lambda)$ is a negligible function depending on the security parameter $\lambda$.

A $\KDF$ is considered secure if it satisfies the following conditions in the context of a security game:
\begin{itemize}
    \item[1.] Challenger: A trusted party responsible for generating and providing the derived keys.
    \item[2.] Adversary: An efficient entity trying to distinguish the output of the $\KDF$ from a randomly generated key.
    \item[3.] Oracle Queries: The adversary can make queries to the $\KDF$ oracle, providing inputs and receiving derived key outputs.
\end{itemize}

The adversary $\adversary$ tries to distinguish between the output of the $\KDF$ and a random key by making queries to the $\KDF$ oracle. The security of the $\KDF$ is quantified by the advantage of the adversary, which measures the distinguishability between the $\KDF$ output and a random key. The game played between the adversary $\adversary$ and the challenger $\challenger$ is show in the figure \ref{fig:kdf_security}, in which the adversary $\adversary$  can interact with the challenger $\challenger$ to demand cryptographic keys corresponding to $\adversary$'s choice of public input $c$ with $p$ and $s$ chosen by $\challenger$. In this game, $p$ is secret and known only to $\challenger$, while $s$ is known by an $\adversary$. At the challenge stage, $\adversary$ is provided a challenge output $\key^\prime$. Lastly, $\adversary$ has to distinguish whether the challenge output is the derived cryptographic key from the $\KDF$ or just a random string.

\subsection{Decisional Diffie-Hellman}

\cite{DDH-def} defines Decisional Diffie-Hellman as follows: 
%Let $Gen$ be a ppt algorithm that on input $1^\pi$ returns a description $\mathcal{G} = (\mathbb{G}, q, \mathcal{P})$ of cyclic group $\mathbb{G}$ and a generator $\mathcal{P}$ of $\mathbb{G}$. Similarly, let $PGen$ be a ppt algorithm that returns a description $\mathcal{PG} = (\mathbb{G},\mathbb{G}_T,q,e,\mathcal{P})$ of a pairing group. We informally recall a previously considered decisional Diffie-Hellman assumption. Diffie-Hellman ($\DDH$) Assumption. It is hard to distinguish ($\mathcal{G}, [x], [y], [xy]$) from ($\mathcal{G}, [x], [y], [xy]$), for $\mathcal{G} = (\mathbb{G},q,\mathcal{P}) \leftarrow \mathsf{Gen}$, $x,y,z \gets \mathbb{Z}_q$

\subsubsection{Security of $\DDH$}

\par \begin{definition}{Decisional Diffie-Hellman Problem}\label{def:ddh-sec}
Let $\Group$ be a group of prime order $q$ and $g$ be a generator of $\Group$, and let $a,b,c \getsr \ZZ_q$. The advantage of an algorithm $\adversary$ in solving the $\DDH$ problem for $(g, q)$ is $\Adv{\ddh}{g,p}(\adversary)$ is defined as
$\Adv{\ddh}{g,p}(\adversary) = |\Pr\big(\adversary(g,g^a,g^b,g^{ab})=1\big)-\Pr\big(\adversary(g,g^a,g^b,g^{c})=1\big)|$. We say that the $\DDH$ problem is hard if for all PPT $\adversary$ $\Adv{\ddh}{g,q}(\adversary)$ is negligible.
\end{definition}

\section{Data Availability}

All data used in this study, are available on the Open Science Framework through this ${\href{https://osf.io/npdtw/?view_only=84621b38349a4f858f5f897c195e80c4}{anonymous ~repository}}$. Additional data for further validation can also be generated by following the steps mentioned in section $\ref{sec:Exp}$.

\end{document}